\documentclass[11pt]{article}
\usepackage{amsthm}
\usepackage{amsmath}
\usepackage{amssymb}
\pdfoutput=1
\usepackage{bbm}
\usepackage{multirow}
\usepackage{mathrsfs}
\usepackage{graphicx}
\usepackage{enumerate}
\usepackage{bm}
\usepackage{verbatim}
\usepackage{natbib}
\usepackage{setspace}
\usepackage{subfigure}
\usepackage{threeparttable}
\usepackage{float}
\usepackage{epstopdf}
\usepackage{color}

\newtheorem{assumption}{Assumption}
\newtheorem{corollary}{Corollary}
\newtheorem{definition}{Definition}

\newtheorem{lemma}{Lemma}
\newtheorem{proposition}{Proposition}

\begin{document}
\onehalfspacing
\title{Systemic Risk of Optioned Portfolios: Controllability and Optimization}
\author{Xiaochuan Pang\thanks{School of Business, Sun Yat-Sen University, Guangzhou 510275, China. Email: pangxch@mail2.sysu.edu.cn}\quad Shushang Zhu\thanks{Corresponding author. School of Business, Sun Yat-Sen University, Guangzhou 510275, China. Email: zhuss@mail.sysu.edu.cn}\quad Xueting Cui\thanks{School of Mathematics, Shanghai University of Finance and Economics, Shanghai 200433, China. Email: cui.xueting@shufe.edu.cn}\quad Jiali Ma\thanks{College of Big Data Statistics, Guizhou University of Finance and Economics, Guiyang 550025, China. Email: majli@mail2.sysu.edu.cn}}
\maketitle
\begin{abstract}
We investigate the portfolio selection problem against the systemic risk which is measured by CoVaR. We first demonstrate that the systemic risk of pure stock portfolios is essentially uncontrollable due to the contagion effect and the seesaw effect. Next, we prove that it is necessary and sufficient to introduce options to make the systemic risk controllable by the correlation hedging and the extreme loss hedging. In addition to systemic risk control, we show that using options can also enhance return-risk performance. Then, with a reasonable approximation of the conditional distribution of optioned portfolios, we show that the portfolio optimization problem can be formulated as a second-order cone program (SOCP) that allows for efficient computation. Finally,  we carry out comprehensive simulations and empirical tests to illustrate the theoretical findings and the performance of our method.
\end{abstract}

\section{Introduction}
The outbreak of the global financial crisis in 2008 seriously hurt the financial system and even the world economy, highlighting the importance of systemic risk management. Since that, systemic risk has attracted extensive attention from both regulators and researchers, and consequently, we have witnessed a rapid growth of the literature on it.

A typical definition of systemic risk is the risk of collapse of an entire financial institutional system or an entire financial market caused by the cascading loss contagion from some part of the system or the market due to the interlinkages and interdependencies among entities (see Wikipedia). In addition to pointing out the common essence of systemic risk, i.e., contagion, this definition also divides systemic risk into two types: systemic risk of  institutional systems and systemic risk of financial markets. In the existing literature, researches on these two types of systemic risk both focus on risk measurement. The difference is that the former is mainly based on risk contagion channels such as cross holding (see, e.g., \cite{eisenberg2001}, \cite{elliott2014}, \cite{ma2021}) to explore the contagion mechanism, while the latter mainly uses market data to quantify the risk contagion effect by upgrading the traditional risk measures. As \cite{benoit2017} pointed out, over the last decade, the most commonly used systemic risk measures are CoVaR of \cite{adrian2011}, SRISK of \cite{acharya2012} and \cite{brownless2017}, and Marginal Expected Shortfall (MES) and Systemic Expected Shortfall (SES) of \cite{acharya2017}. The systemic risk of a portfolio studied in this paper refers to the systemic risk of financial markets. Among the popular systemic risk measures mentioned above, CoVaR is especially suitable for quantifying the systemic risk of a portfolio. Thus we use it in this paper as a representative of systemic risk measures.

Although portfolio selection models originated from \cite{markowitz1952} consider risk management, they do not efficiently deal with the extreme portfolio loss caused by some systemic events. In the classical portfolio theory, the risk is divided into two parts, the individual risk and the systematic risk, and diversification is emphasized since it not only eliminates the individual risk, but also optimizes and reshapes the systematic risk without harming the expected return. However,  the systemic risk of a portfolio, i.e., the risk of collapse of the entire portfolio caused by cascading contagion from the distress of some assets in the portfolio, is evidently different from the systematic risk. In fact, the correlation among financial assets in the portfolio provides the possibility for the contagion, amplification and formation of systemic risk. Thus, diversification among correlated assets may increase systemic risk since a high level of diversification implies more interdependencies.

\begin{figure}[H]
    \centering
    \includegraphics[width=10cm]{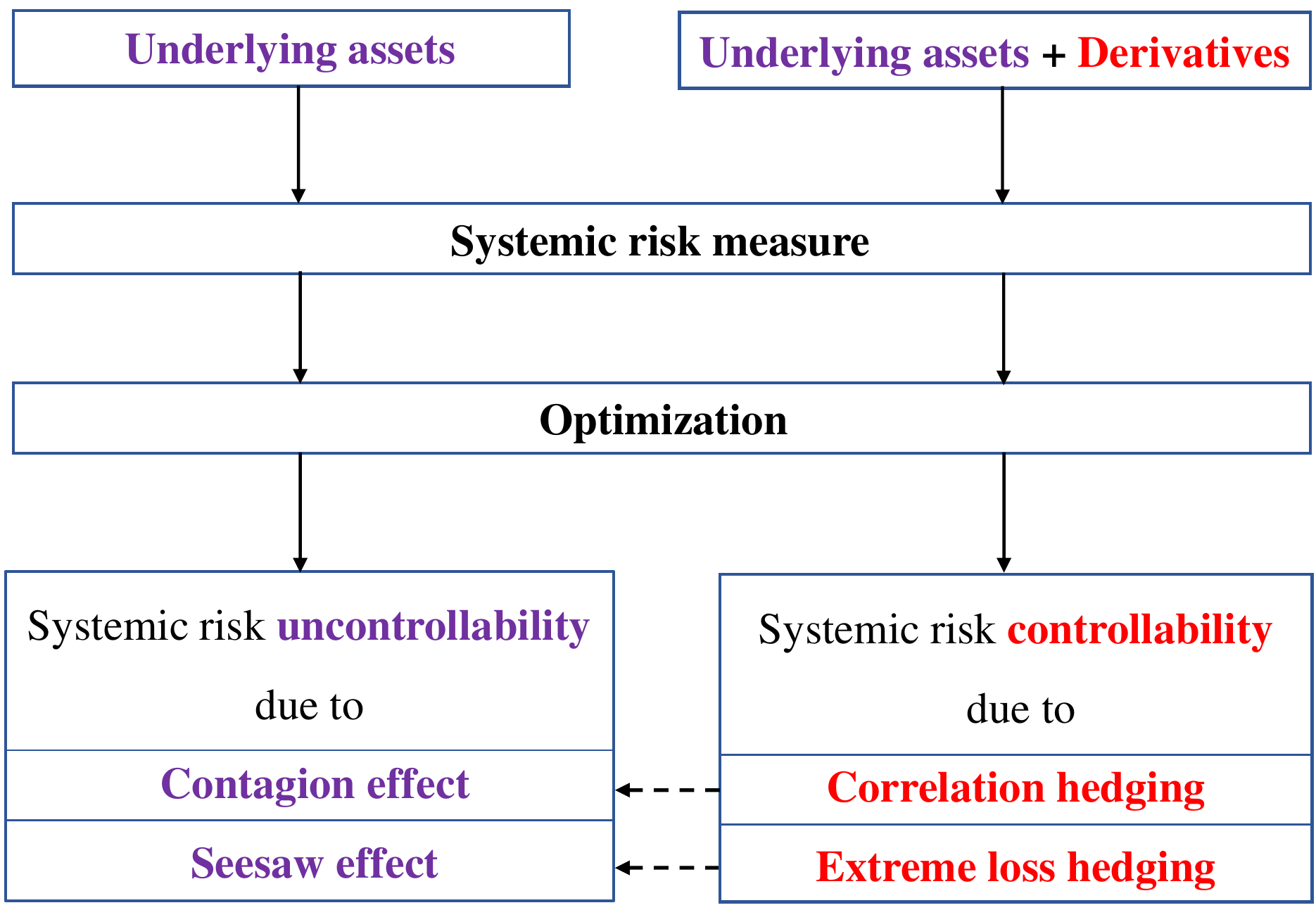}
    \caption{Framework of portfolio with systemic risk control}
    \label{framework}
\end{figure}

As shown in Figure {\ref{framework}}, in this paper, we propose a portfolio optimization model with the consideration of systemic risk control, which also involves conventional procedures: the selection of constituent assets, risk measure and optimization. In addition, we try to answer a more fundamental question: whether the systemic risk is controllable for a given portfolio.

We demonstrate that the systemic risk of portfolios only containing underlying assets is usually uncontrollable due to the contagion effect and the  seesaw effect. More specifically, the contagion effect refers to that, even if only one asset gets distressed, the portfolio usually suffers large losses since underlying assets are positively correlated, especially in an extreme market case, while the seesaw effect means that if the systemic risk caused by the collapse of an asset is restricted to a low level, then the systemic risk caused by the collapse of another asset might rise to a high level at the same time. If investors want to control systemic risk to a safe level, they have to execute a more conservative strategy with a relatively low expected return since there is a trade-off between the systemic risk and the expected return. To effectively prevent the systemic risk and maintain a relatively high return-risk performance in the traditional sense, we suggest adding derivatives into the portfolio to reconcile the conflict between diversification and systemic risk. By combining each underlying asset with derivatives written on it, we show theoretically that the systemic risk of portfolios containing derivatives is controllable compared with portfolios only containing underlying assets, since the derivatives can be used effectively to hedge the correlations among underlying assets and hedge the extreme losses, and thus alleviate the contagion effect and the seesaw effect simultaneously.

Surprisingly, to the best of our knowledge, researches on portfolio selection problems relevant to systemic risk are very limited. \cite{das2004} proposed a model with a common jump-diffusion process in stock returns to capture the systemic risk. They found that there is little difference between the portfolio of risky assets with and without the consideration of systemic risk and that the total position on risky assets that accounts for systemic risk is smaller. \cite{biglova2014} used Co-Expected Tail Loss as the systemic risk measure and Co-Expected Tail Profit as the reward and maximized the ratio of reward to systemic risk to select the portfolio. They verified their model with stock indexes of 14 countries and found that their model had a good performance during the period of financial instability. \cite{Capponi2022} used two types of Conditional Expected Shortfall as the systemic risk measure and discussed how the correlation between the portfolio and the market affects the systemic risk. Given the distribution of stock returns as Student's t, they provided the solutions to the systemic risk minimization problems. In the empirical study, similarly, their model performed well during the period of the market downturn. Although these models had a satisfactory performance in the out-of-sample test, they just consider the systemic risk control of pure stock portfolios. In this paper, we demonstrate that optioned portfolios are much more effective than stock portfolios in the sense of systemic risk control.

The literature on portfolio selection problems containing derivatives is also relatively limited and can be roughly divided according to risk measures. \cite{liang2008} used variance as the risk measure and solved the multi-stage portfolio selection problem with options under the mean-variance framework. Given the price changes of the underlying assets are normally distributed, \cite{britten-jones1999} and \cite{cui2013} provided the moments of the portfolio containing stocks and options and solved the mean-VaR optimization model. \cite{alexander2006} found that the CVaR minimization problems including derivatives are often ill-posed and they accounted for the transaction and management cost to alleviate this problem. \cite{zymler2013} considered worst-case VaR as the risk measure and provided the solution to the portfolio optimization model with European potions and American options respectively. \cite{deng2014} focused on the minimal risk of the portfolio given that the underlying stock returns and volatility variations are bounded within a circle or a rectangle. Instead of using return-risk framework, some papers derived the optimal portfolio of options by maximizing an expected utility function, such as \cite{liujun2003}, \cite{driessen2007} and \cite{faias2017}.

In the literature, the use of derivatives for the purpose of optimally hedging the extreme risk of a portfolio is originated from \cite{wilmott2007} and extended by \cite{zhu2020}, which are relevant to this paper. However, there exists essential differences between those works and our study. In this paper, we use CoVaR that is defined as the quantile of the random return of the portfolio as the systemic risk measure, while \cite{wilmott2007}  and  \cite{zhu2020} use the worst-case loss realized within a box or an ellipse as the extreme risk, treating the loss as a deterministic variable. Actually, the risk discussed in  \cite{wilmott2007}  and \cite{zhu2020} is the crash risk, which considers the case where all the assets fall down together. Obviously, crash risk is more severe than the systemic risk discussed in this paper that results from contagion initiated by the distress of some assets.

In this paper, we regard stocks, the underlying assets, as risk factors and investigate how and why options, the derivatives, can be used to hedge the systemic risk. We call the portfolios involving options as {\it optioned portfolios} for simplicity. The main contributions are in three aspects. First, we explore why the contagion effect and seesaw effect can lead to the uncontrollability of the systemic risk of pure stock portfolio. Second, we demonstrate the necessity and merit of introducing options for systemic risk control from both the perspectives of correlation hedging and extreme loss hedging related to the contagion effect and  seesaw effect respectively. Third, we investigate how the systemic risk of optioned portfolios can be controlled optimally. By approximating the distribution of the optioned portfolio with its asymptotic distribution, we derive the analytical form of the systemic risk measure, based on which we further show that the optimization model can be formulated as an SOCP that allows for efficient computation.

The rest of this paper is organized as follows. In Section 2, we propose the systemic risk constrained model with CoVaR as the systemic risk measure and clarify why the systemic risk of stock portfolios is uncontrollable. In Section 3, we comprehensively show why and how options improve the efficiency of systemic risk control in portfolio selection. In Section 4, we carry out simulations and empirical studies to examine the theoretical results and the performance of the method. Conclusions and discussions are summarised in Section 5.

\textbf{Notations:} Throughout this paper, lowercase boldface letters represent column vectors and uppercase letters represent matrices. For an index subset $\mathcal{I}$, $\bm{x}_{\mathcal{I}}$ is the subvector of $\bm{x}$ with elements indexed by $\mathcal{I}$. $A_{\mathcal{I}\mathcal{J}}$ represents the submatrix of $A$ consisting of rows and columns indexed by $\mathcal{I}$ and $\mathcal{J}$. $A_{\mathcal{I}\cdot}~ (A_{\cdot\mathcal{I}})$ represents the submatrix of $A$ consisting of rows (columns) indexed by $\mathcal{I}$. For a countable finite set $\mathcal{I}$, denote $|\mathcal{I}|$ as the cardinality of the set. Denote $||\cdot||$ as the $l_2$ norm. Denote ${\rm I}$  as the identity matrix, and $\bm{0}$ the vector of all entries being zeros.

\section{Problem setup}
Consider a portfolio consisting of assets such as stocks, bonds or derivatives over a given period. Denote $\bm{u}$ as the vector of asset prices at the beginning of the investment period and $\Delta\bm{u}$ as the random vector of asset price changes during this period. A portfolio is characterized by a weight vector $\bm{w}$ representing the holding amount of assets. Let $\Delta v(\bm{w})=\Delta\bm{u}^{\top}\bm{w}$
denote the change in the value of the portfolio. In this paper, we consider the portfolio selection problem with systemic risk control as follows
\begin{eqnarray*}
(P)&\max\limits_{\bm{w}}&\mu(\bm{w})\\
&{\rm s.t.}&\sigma(\bm{w})\leq \bar\sigma,\\
&&\rho_j(\bm{w})\leq \bar\rho_j,~~~~ j=1,\cdots,h, \\
&&\bm{w}\in\Omega,
\end{eqnarray*}
where $\mu(\bm{w})$, $\sigma(\bm{w})$ and $\rho_j(\bm{w})$ are the expected return, standard deviation and systemic risk driven by the $j$th of $h$ distressed events given the portfolio $\bm{w}$, $\bar\sigma$ and $\bar\rho_j$ are pre-determined parameters representing the tolerances of investors on traditional risk and systemic risk, and $\Omega$ is the set of admissible portfolio.

Following \cite{adrian2011}, we adopt CoVaR as the systemic risk measure in this paper. More specifically,
denote
\begin{eqnarray*}&&\mathscr{G}=\{\Delta u_i=-k_i, i\in\mathcal{I}\}
\end{eqnarray*}
as the distressed event indexed by $\mathcal{I}$. Here, $k_i$ is set to be the Value-at-Risk (VaR) of asset $i$, which is defined by
\begin{eqnarray}\label{VaR}
&& VaR_i^p=\inf\left\{a\in\mathbb{R}:{\rm Pr}(-\Delta u_i\geq a)\leq1-p\right\},
\end{eqnarray}
where ${\rm Pr(\cdot)}$ denotes the probability, $p(\geq0.5)$ is the confidence level. Then the systemic risk of the portfolio conditioning on the occurrence of $\mathscr{G}$ is defined by
\begin{eqnarray}\label{CoVaR}
&& CoVaR_q^{\Delta v(\bm{w})|\mathscr{G}}=\inf\left\{a\in\mathbb{R}:{\rm Pr}(-\Delta\bm{u}^{\top}\bm{w}\geq a|\mathscr{G})\leq1-q\right\},
\end{eqnarray}
where $q(\geq0.5)$ is the confidence level.

Since the essence of systemic risk is the risk contagion triggered by some distressed assets, it is necessary to attach multiple risk constraints to control the potential systemic risk and distinguish the different initial distresses. Without loss of generality, in the sequel, we consider only one systemic risk constraint, i.e., $h=1$, in the model to simplify the discussion, except that multiple systemic risk constraints are indeed necessary to be introduced into the discussion.

\subsection{Systemic risk measure for pure stock portfolio}

Now we specify the portfolio as a pure stock portfolio consisting of $m$ stocks. Denote the vector of stock prices and random vector of stock price changes by $\bm{p}=(p_i)_{m}$ and $\Delta\bm{p}=(\Delta p_i)_{m}$. Let $\bm{y}=(y_i)_{m}$ denote the weight vector of holding amount of stocks. Thus for the pure stock portfolio, we have $\bm{u}=\bm{p}$, $\Delta\bm{u}=\Delta\bm{p}$ and $\bm{w}=\bm{y}$.

Throughout this paper, we assume the returns/price changes of stocks follow a multivariate normal distribution, which is a common assumption in risk analysis of optioned portfolios for the sake of tractability (see, e.g., \cite{britten-jones1999}, \cite{cui2013}). Actually, fat tail distributions can be modeled by heteroscedasticity models such as the GARCH type models (see, e.g., \cite{bollerslev1986} and \cite{Engle2002}), where the stock returns are usually assumed to follow a time-varying conditional normal distribution. This also provides the rationality for the assumption of normal distribution in a single period case.
\begin{assumption}
The distribution of the price changes of stocks is normally distributed,  i.e.  $\Delta\bm{p} \thicksim \mathcal {N}(\bm{\mu},\Sigma)$, where $\bm{\mu}=(\mu_i)_{m}$ is the mean vector and $\Sigma=(\sigma_{ij})_{m\times m}$ is the nonsingular covariance matrix.
\end{assumption}
Under Assumption 1, the Value-at-Risk of stock $i$ defined by (\ref{VaR}) can be analytically expressed as
\begin{eqnarray*}
&&k_i=\alpha_p\sqrt{\sigma_{ii}}-\mu_i,
\end{eqnarray*}
where $\alpha_p$ is the $p$-quantile of the standard normal distribution. Denote $\bm{k}=(k_i)_{|\mathcal{I}|}$ and  $\mathcal{J}=\{1,\cdots,m\}\backslash\mathcal{I}$ as the complementary set of $\mathcal{I}$.
By the properties of multivariate normal distribution, the conditional distribution of the portfolio remains normally distributed. Furthermore, given the condition $\mathscr{G}=\{\Delta p_i=-k_i, i\in\mathcal{I}\},$
the conditional mean and conditional variance of the portfolio are
\begin{eqnarray*}
&&\mathbb{E}(\Delta v(\bm{y})|\mathscr{G}) = \bm{c}^{\top}\bm{y}_{\mathcal{J}}-\bm{k}^{\top}\bm{y}_{\mathcal{I}}~\mbox{and}~
\mathbb{V}(\Delta v(\bm{y})|\mathscr{G}) = \bm{y}_{\mathcal{J}}^{\top}E\bm{y}_{\mathcal{J}},
\end{eqnarray*}
where
\begin{eqnarray*}
&&\bm{c} = (c_j)_{|\mathcal{J}|}=\bm{\mu}_{\mathcal{J}}-\Sigma_{\mathcal{J}\mathcal{I}}\Sigma_{\mathcal{I}\mathcal{I}}^{-1}(\bm{k}+\bm{\mu}_{\mathcal{I}})
~\mbox{and}~
E = (e_{jj})_{|\mathcal{J}|\times |\mathcal{J}|}= \Sigma_{\mathcal{J}\mathcal{J}}-\Sigma_{\mathcal{J}\mathcal{I}}\Sigma_{\mathcal{I}\mathcal{I}}^{-1}\Sigma_{{\mathcal{I}}{\mathcal{J}}}
\end{eqnarray*}
are the conditional mean vector and conditional covariance matrix of price changes of stocks indexed by $\mathcal{J}$. Then
CoVaR can be explicitly expressed as the $q$-quantile of the conditional distribution of the portfolio
\begin{eqnarray}\label{CoVaR_normal}
&& CoVaR_q^{\Delta v(\bm{y})|\mathscr{G}}=\alpha_q\sqrt{\bm{y}_{\mathcal{J}}^{\top}E\bm{y}_{\mathcal{J}}}-\bm{c}^{\top}\bm{y}_{\mathcal{J}}+\bm{k}^{\top}\bm{y}_{\mathcal{I}}.
\end{eqnarray}

\subsection{Uncontrollability of systemic risk for pure stock portfolio}

A straightforward meaning of uncontrollability is that the systemic risk of pure stock portfolios can not be controlled to any level. We show that the systemic risk measured by CoVaR is bounded below. We assume that the initial wealth of the portfolio is one and short selling is not permitted. Denote by
\begin{eqnarray*}
&&\mathcal{K}(\bar\rho)=\left\{\bm{y}\in\mathbb{R}^{m}: \bm{p}^{\top}\bm{y}=1,\bm{y}\geq\bm{0},CoVaR_q^{\Delta v(\bm{y})|\mathscr{G}}\leq\bar\rho\right\}
\end{eqnarray*}
the feasible set of the portfolio given the acceptable systemic risk level $\bar\rho$.
Obviously, the emptiness of $\mathcal{K}(\bar\rho)$ is equivalent to the uncontrollability of systemic risk.

The following proposition provides the  conditions of uncontrollability/controllability of systemic risk for the pure stock portfolio. To reconcile the ordinal of elements in vectors with different length, we further assume
\begin{eqnarray*}
&&\mathcal{J}=\{1,\cdots,m'\}~\mbox{ and }~\mathcal{I}=\{m'+1,\cdots,m\}
\end{eqnarray*}
without loss of generality (e.g., under this assumption, for $i\in\mathcal{J}$, $\mu_i$ and $c_i$ represent the mean and conditional mean of the price change of the same stock).

\begin{proposition}\label{controlability}
$\mathcal{K}(\bar\rho)=\emptyset$ if and only if
\begin{eqnarray*}
&&\bar\rho<\min\limits_{i\in\mathcal{I}}\left\{\frac{\alpha_p\sqrt{\sigma_{ii}}-\mu_i}{p_i}\right\} ~\mbox{and}~ \bar\rho<\max\limits_{||\bm{r}||\leq 1}\min\limits_{j\in\mathcal{J}}\left\{\frac{\alpha_q F_{\cdot j}^{\top}\bm{r}-c_{j}}{p_j}\right\},
\end{eqnarray*}
where $F^{\top}F=E$.
\end{proposition}

\begin{proof}
Recalling equation (\ref{CoVaR_normal}), the constraints that construct $\mathcal{K}(\bar\rho)$ can be reformulated as the following inequality system
\begin{eqnarray}\label{P1.1}
&&\bm{p}^{\top}\bm{y}=1,~\bm{y}\geq0~\mbox{ and}~||F\bm{y}_{\mathcal{J}}||\leq\frac{\bar\rho-\bm{k}^{\top}\bm{y}_{\mathcal{I}}+\bm{c}^{\top}\bm{y}_{\mathcal{J}}}{\alpha_q}
\end{eqnarray}
with respect to $\bm{y}$. According to the Theorem of Alternatives  of Chapter 5.8-5.9 of \cite{boyd2004}, the alternative of inequality system (\ref{P1.1}) is
\begin{eqnarray}\label{P1.2}
&&s-\frac{\bar\rho l_0}{\alpha_q}>0,~ s\bm{p}_{\mathcal{I}}-\frac{l_0}{\alpha_q}\bm{k}\leq0,~ F^{\top}\bm{l}+\frac{l_0}{\alpha_q}\bm{c}+s\bm{p}_{\mathcal{J}}\leq0~\mbox{and}~ ||\bm{l}||\leq l_0
\end{eqnarray}
with respect to $s$, $l_0$ and $\bm{l}$. The Theorem of Alternatives implies that the inequality system (\ref{P1.1}) is infeasible if and only if (\ref{P1.2}) is feasible. Now we consider the relationship between the value of $\bar\rho$ and the feasibility of (\ref{P1.2}). Noting that $\bm{p}>\bm{0}$, replacing the first inequality of (\ref{P1.2}) into the second and the third inequalities yields the following equivalent inequality system
\begin{eqnarray}\label{P1.3}
&&\frac{l_0}{\alpha_q}(\bar\rho\bm{p}_{\mathcal{I}}-\bm{k})<0,~ \frac{l_0}{\alpha_q}(\bar\rho\bm{p}_{\mathcal{J}}+\bm{c})+F^{\top}\bm{l}<0~\mbox{and}~||\bm{l}||\leq l_0
\end{eqnarray}
with respect to $l_0$ and $\bm{l}$. Let $\bm{r}=-\bm{l}/l_0$ and then we can reformulate (\ref{P1.3})  as
\begin{eqnarray}\label{P1.4}
&&\bar\rho\bm{p}_{\mathcal{I}}-\bm{k}<0,~ \bar\rho\bm{p}_{\mathcal{J}}+\bm{c}-\alpha_q F^{\top}\bm{r}<0 ~\mbox{and}~||\bm{r}||\leq 1
\end{eqnarray}
with respect to $\bm{r}$. Notice that inequality system (\ref{P1.4}) is feasible if and only if
\begin{eqnarray*}
&&\bar\rho<\min\limits_{i\in\mathcal{I}}\left\{\frac{\alpha_p\sqrt{\sigma_{ii}}-\mu_i}{p_i}\right\}~\mbox{and}~\bar\rho<\max\limits_{||\bm{r}||\leq 1}\min\limits_{j\in\mathcal{J}}\left\{\frac{\alpha_q F_{\cdot j}^{\top}\bm{r}-c_{j}}{p_j}\right\},
\end{eqnarray*}
which implies that inequality system (\ref{P1.1}) is infeasible and $\mathcal{K}(\bar\rho)=\emptyset$. The proof is completed.
\end{proof}

Proposition \ref{controlability} indicates that, in a general case, the systemic risk for any pure stock portfolio is bounded below and the bound depends on the individual VaRs of stocks indexed by $\mathcal{I}$ and the conditional means and conditional covariances of stocks indexed by $\mathcal{J}$. Specially, consider a case when there exists a systemically important stock $i$ such that $\rho_{ij}\approx 1$ for any $j\neq i$, $j\in\{1,\cdots,m\}$. Then we take $\bm{r}=F_{\cdot i}/||F_{\cdot i}||$. Notice that $F^{\top}_{\cdot j}F_{\cdot i}=e_{ij}$ and $||F_{\cdot i}||=\sqrt{e_{ii}}$. Thus
\begin{eqnarray*}
&&\max\limits_{||\bm{r}||\leq 1}\min\limits_{j\in\mathcal{J}}\left\{\frac{\alpha_q F_{\cdot j}^{\top}\bm{r}-c_{j}}{p_j}\right\}\geq\min\limits_{j\in\mathcal{J}}\left\{\frac{\alpha_q\rho_{ij}\sqrt{e_{jj}}-c_j}{p_j}\right\}\approx\min\limits_{j\in\mathcal{J}}\left\{\frac{\alpha_q \sqrt{e_{jj}}-c_j}{p_j}\right\}.
\end{eqnarray*}
It means that if $\bar\rho<\min\limits_{i\in\mathcal{I}}\{\frac{\alpha_p\sqrt{\sigma_{ii}}-\mu_i}{p_i}\}$ and $\bar\rho<\min\limits_{j\in\mathcal{J}}\{\frac{\alpha_q\sqrt{e_{jj}}-c_j}{p_j}\}$, then $\mathcal{K}(\bar\rho)=\emptyset$. In other words, to make the portfolio feasible, i.e., $\mathcal{K}(\bar\rho)\neq\emptyset$, $\bar\rho$ must be set relatively large, which further implies that the systemic risk cannot be controlled within a relatively low level for a stock portfolio, especially when stocks are highly correlated.

Besides the existence of a lower bound on the systemic risk, a more fundamental explanation for the uncontrollability of systemic risk is the contagion effect and the seesaw effect. To simplify the analysis, we use specific examples rather than general analysis to demonstrate the existence of  these two effects, since these examples are sufficient to illustrate the uncontrollability of systemic risk of pure stock portfolios. More general examples are shown later in the empirical section.

The contagion effect in the pure stock portfolios means that the systemic risk of the portfolio is amplified through risk contagion. We consider a portfolio consisting of two stocks. The initial wealth of the portfolio and the stock prices are both set at one. Conditioning on $\mathscr{G}_1=\{\Delta p_1=-k_1\}$,  by equation (\ref{CoVaR_normal}), we have
\begin{eqnarray}\label{sr}
&&CoVaR_q^{\Delta v(\bm{y})|\mathscr{G}_1}=\left(\psi\sqrt{\sigma_{22}}-\mu_2\right)y_2+\left(\alpha_p\sqrt{\sigma_{11}}-\mu_1\right)y_1,
\end{eqnarray}
where $\psi=\alpha_q\sqrt{1-\rho_{12}^2}+\alpha_p\rho_{12}$. We can see from equation (\ref{sr}) that the systemic risk depends on the correlation, the weights on stocks and the quantiles.

It is well known that the correlation between asset returns/losses is usually enhanced in extreme situations, which may lead to financial contagion. We use numerical examples to illustrate the relationship between systemic risk and correlation in the sequel. Set $\mu_1=\mu_2=0$, $y_1=y_2=0.5$ and $\sigma_{11}=\sigma_{22}=0.01$. In the left sub-figure of Figure \ref{contagion}, where ${\rm CoVaR_1}$ is the CoVaR conditioning on stock 1 being under distressed, we investigate how the systemic risk varies with the loss of stock 1 under different correlation settings. It is clear that the systemic risk is an increasing function of the initial loss and that  the larger the correlation, the larger the marginal effect of initial loss on systemic risk.

The right sub-figure of Figure 2 shows how the systemic risk changes with respect to the correlation given quantiles $p$ and $q$. It is clear that, in most of the interval $[0,1]$, the systemic risk is an increasing function of the correlation $\rho_{12}$, which again indicates that more correlation means more risk contagion. However, we also find that the systemic risk is not an increasing function of $\rho_{12}$ while $\rho_{12}$ is close to 1, except for the case $q=0.5$. Roughly speaking, the conditional variance of the portfolio, which represents the uncertainty caused by factors other than the systemic event $\mathscr{G}_1=\{\Delta p_1=-k_1\}$, is a decreasing function of correlation $\rho_{12}$. If $q=0.5$, the weight on conditional variance $\alpha_q=0$, and CoVaR turns to the conditional mean of losses of the portfolio and is an increasing function of correlation. If $q>0.5$ and $\alpha_q>0$, CoVaR is a weighted sum of conditional variance and conditional mean of losses. When the correlation increases, the conditional variance decreases and CoVaR does not necessarily increase.

\begin{figure}[H]
  \begin{minipage}{0.5\linewidth}
  \centerline{\includegraphics[width=8cm]{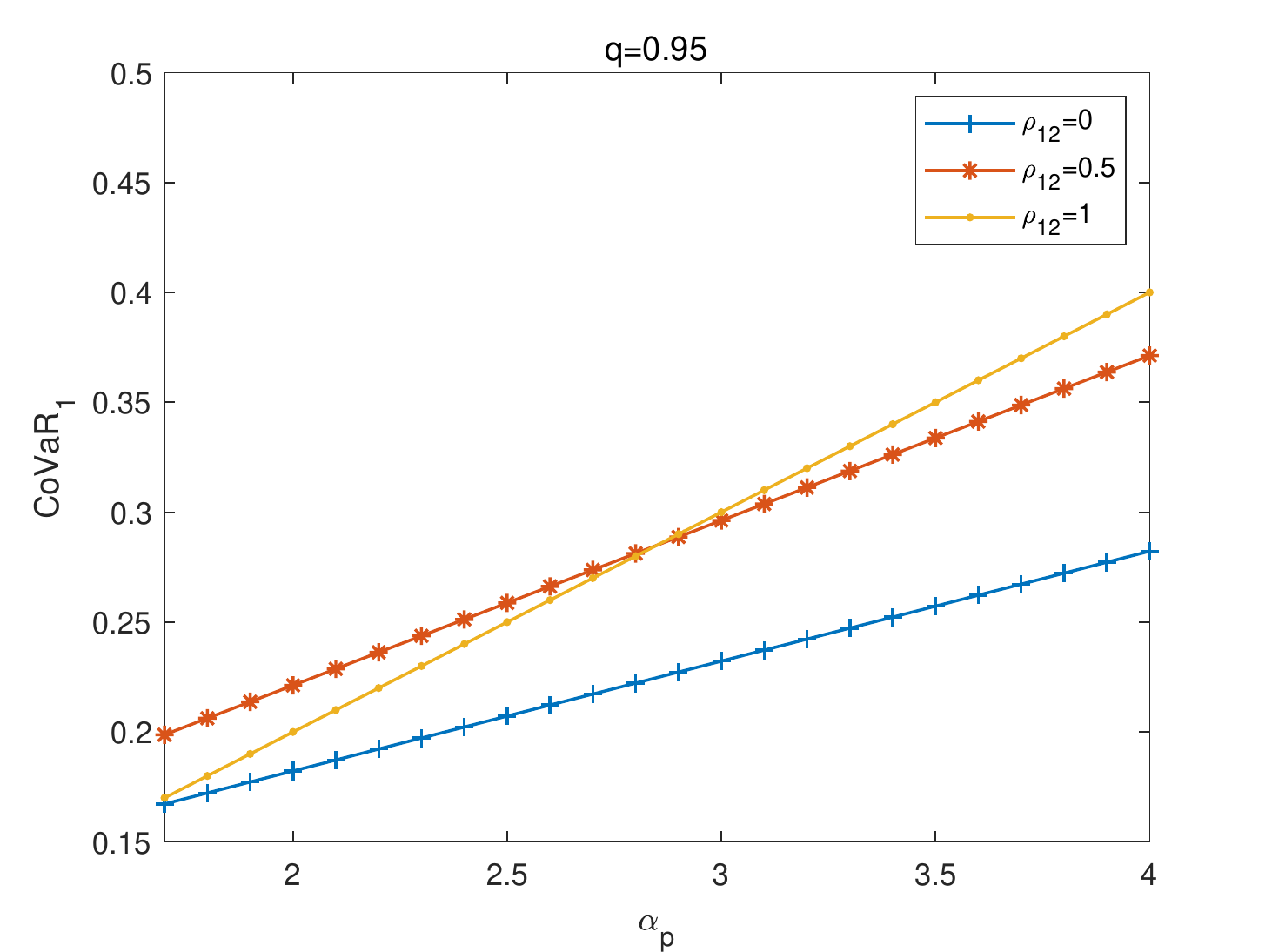}}
  \end{minipage}
  \hfill
  \begin{minipage}{0.5\linewidth}
  \centerline{\includegraphics[width=8cm]{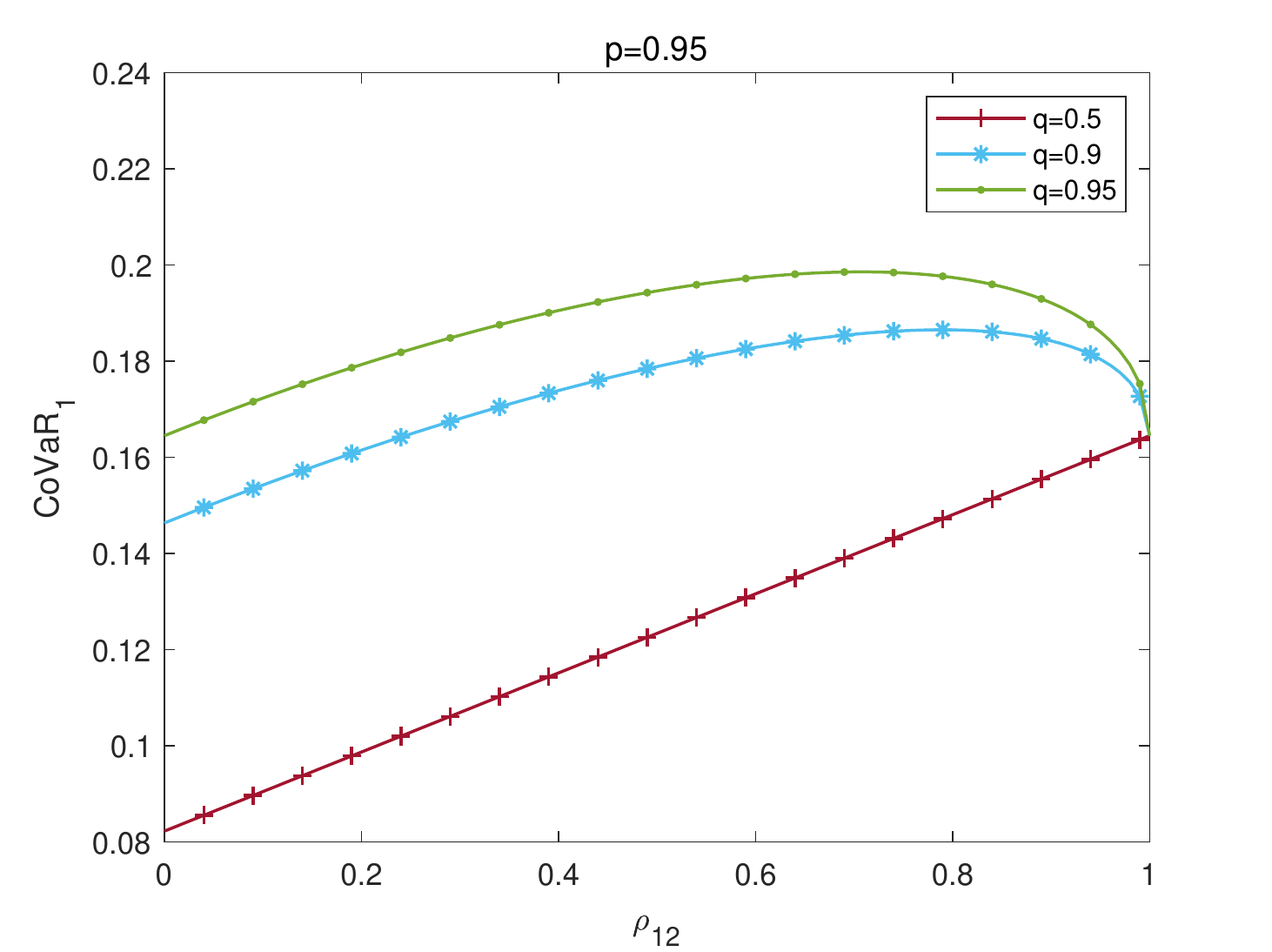}}
  \end{minipage}
  \caption{Contagion effect of systemic risk}
  \label{contagion}
\end{figure}

As we consider multiple systemic risk constraints, there exists the seesaw effect in the systemic risk control for pure stock portfolios, which means that a requirement of low systemic risk conditioning on some stocks being under distress usually forces systemic risk conditioning on other stocks being under distress to rise up.

We still consider the example with two stocks to illustrate the  seesaw effect. Similar to (\ref{sr}), the CoVaR of the portfolio conditioning on stock 2 being under distress is
\begin{eqnarray}
&&CoVaR_q^{\Delta v(\bm{y})|\mathscr{G}_2} = \left(\psi\sqrt{\sigma_{11}}-\mu_1\right)y_1+\left(\alpha_p\sqrt{\sigma_{22}}-\mu_2\right)y_2.\label{seesaw2}
\end{eqnarray}
Since we assume that the prices of stocks are one and the initial wealth is one, we have $y_2=1-y_1$. Replacing it into (\ref{sr}) and (\ref{seesaw2}) derives
\begin{eqnarray*}
CoVaR_q^{\Delta v(\bm{y})|\mathscr{G}_1} &=& \left(\alpha_p\sqrt{\sigma_{11}}-\mu_1-\psi\sqrt{\sigma_{22}}+\mu_2\right)y_1+\psi\sqrt{\sigma_{22}}-\mu_2,\\
CoVaR_q^{\Delta v(\bm{y})|\mathscr{G}_2} &=& \left(\psi\sqrt{\sigma_{11}}-\mu_1-\alpha_p\sqrt{\sigma_{22}}+\mu_2\right)y_1+\alpha_p\sqrt{\sigma_{22}}-\mu_2.
\end{eqnarray*}
Notice that both of the CoVaRs are linear functions of $y_1$. If the product of the coefficients of $y_1$ is negative, i.e.,
\begin{eqnarray*}
&&\left(\alpha_p\sqrt{\sigma_{11}}-\mu_1-\psi\sqrt{\sigma_{22}}+\mu_2\right)\left(\psi\sqrt{\sigma_{11}}-\mu_1-\alpha_p\sqrt{\sigma_{22}}+\mu_2\right)<0,
\end{eqnarray*}
then $CoVaR_q^{\Delta v(\bm{y})|\mathscr{G}_1}$ and $CoVaR_q^{\Delta v(\bm{y})|\mathscr{G}_2}$ move in the opposite direction no matter whether $y_1$ increases or decreases.

Figure \ref{seesaw} provides two concrete examples with $\mu_1=\mu_2=0$, $\rho_{12}=0.1$ and $\sigma_{11}=0.01$ but different $\sigma_{22}$. Here, VaR displayed in the figure is the VaR of the portfolio, $\rm CoVaR_1$ and $\rm CoVaR_2$ denote the CoVaR of the portfolio conditioning on stock 1 and stock 2 being at their VaR, respectively. Three important implications are exhibited by the examples. First, in the left figure, $\rm CoVaR_1$ and $\rm CoVaR_2$ decrease or increase in the same direction. However, in the right figure, as the weight on stock 1 increases (weight on stock 2 decreases), the systemic risk caused by stock 1 decreases while the systemic risk caused by stock 2 increases. In this case, the systemic risk is uncontrollable in the sense of seesaw effect. Second, the systemic risk could be larger than the tail risk (VaR) of the whole portfolio. Third, the diversification that can effectively lower the VaR of a portfolio no longer works well in reducing the systemic risk.
\begin{figure}
  \begin{minipage}{0.5\linewidth}
  \centerline{\includegraphics[width=8cm]{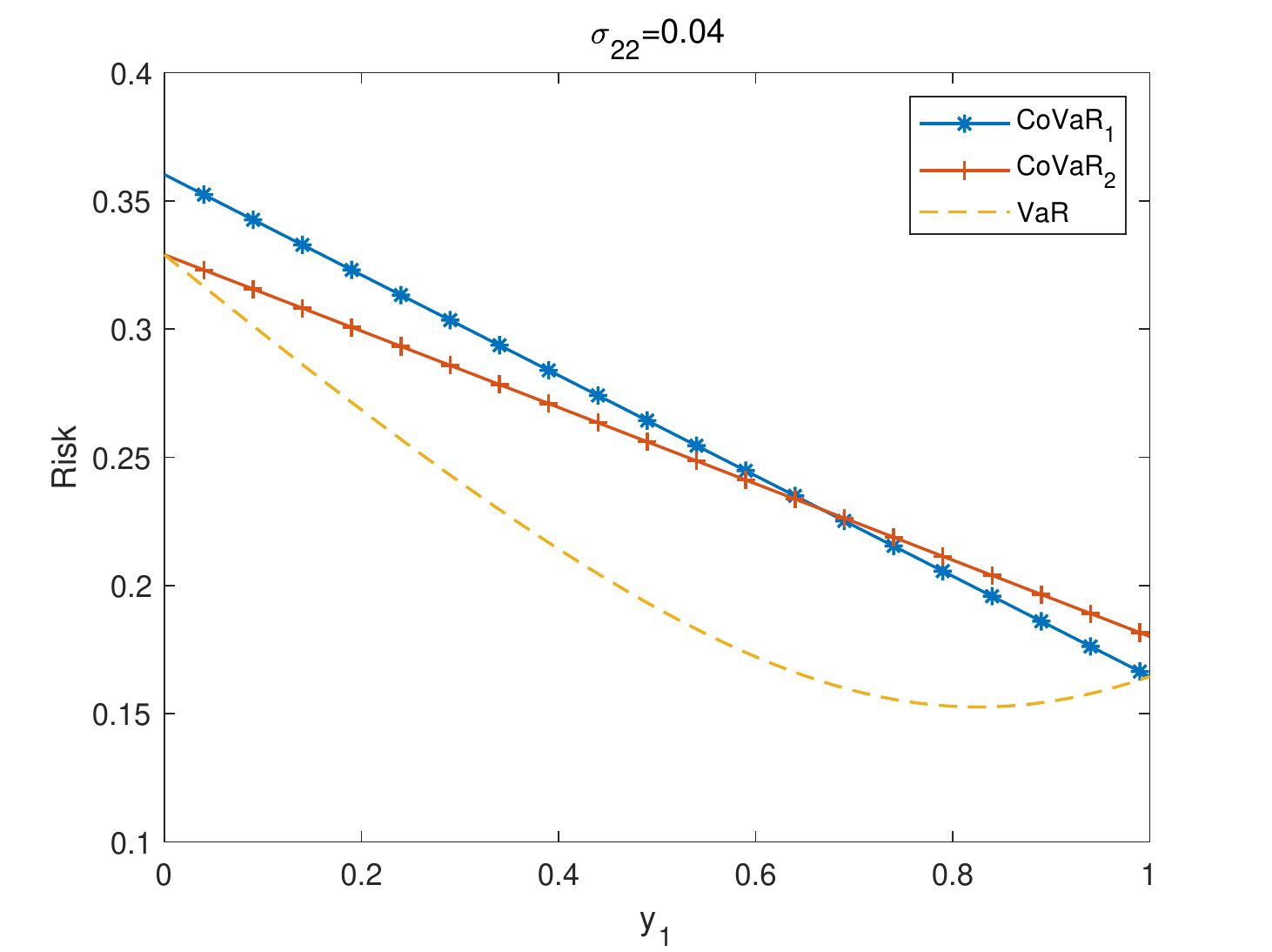}}
  \end{minipage}
  \hfill
  \begin{minipage}{0.5\linewidth}
  \centerline{\includegraphics[width=8cm]{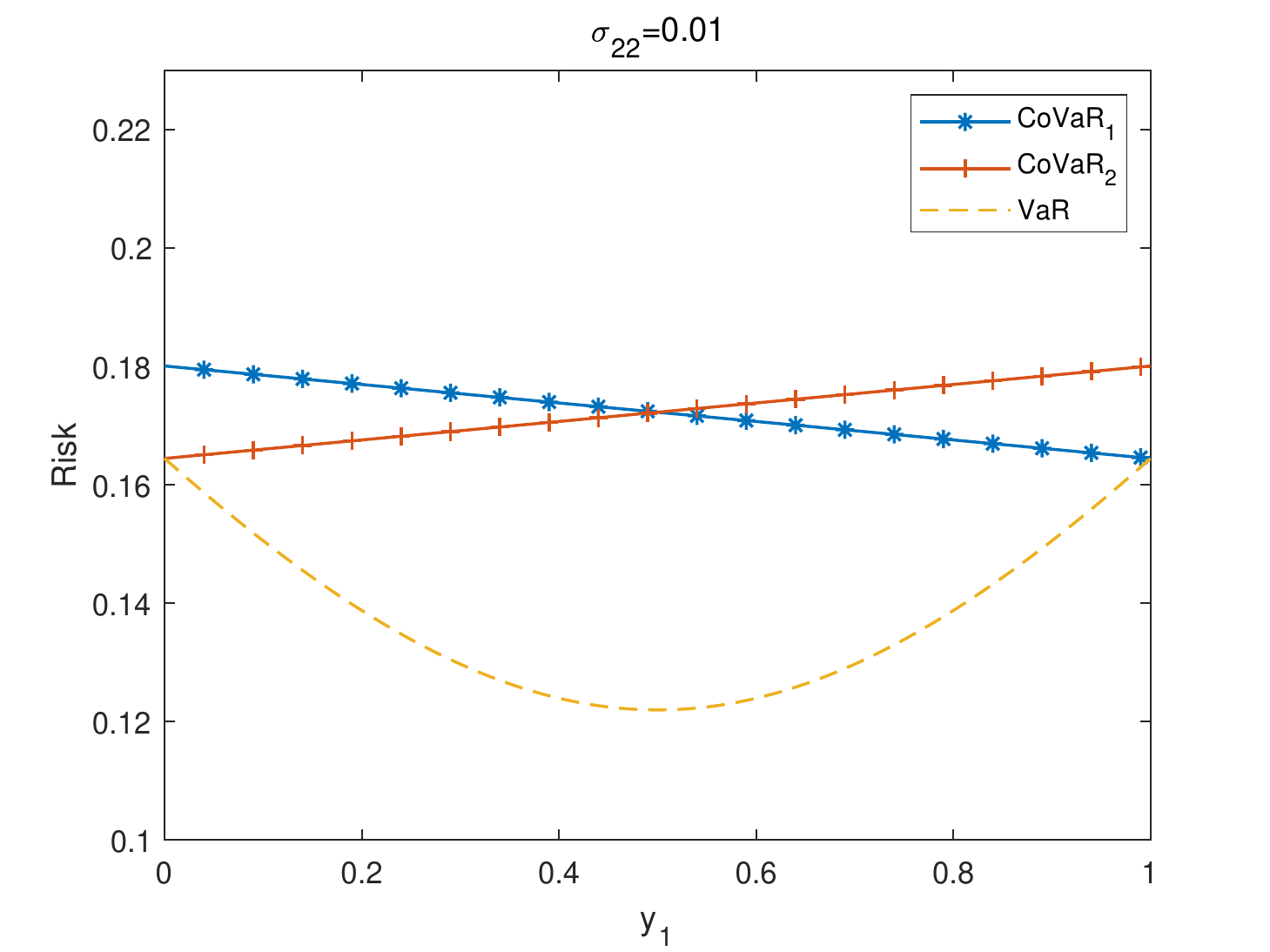}}
  \end{minipage}
  \caption{Seesaw effect in systemic risk control}
  \label{seesaw}
\end{figure}
The seesaw effect illustrated by the right sub-figure of Figure \ref{seesaw} is essentially caused by the limited investment universe. More specifically, given some stocks are under distressed, lowering the systemic risk might lead to more allocation of funds on other stocks that are less correlated with the distressed stocks. However, on the other hand, when the high weighting stocks are under distressed, the portfolio still suffers large losses due to high proportion of investment on them.

To alleviate the uncontrollability of systemic risk of pure stock portfolios, a natural alternative is enlarging the assets universe, thereby expanding the feasible set of the portfolio and allowing the minimal bound to emerge at a lower level. We suggest to use derivatives to hedge the systemic risk, since derivatives can hedge the risk contagion and also the extreme losses. We show these merits in the next section.

\section{Systemic risk control for optioned portfolio}
In this section, we discuss the systemic risk control of optioned portfolios. We first provide the systemic risk measure for optioned portfolios, and then compare optioned portfolios with pure stock portfolios from different aspects to illustrate the controllability of optioned portfolios.

\subsection{Systemic risk measure for optioned portfolio}
Suppose there are $n$ options written on the $m$ underlying stocks which can be selected in the portfolio. Let $\bm{d}=(d_i)_{n}$ be the vector of option prices, $\Delta\bm{d}=(\Delta d_i)_{n}$ be the random vector of option price changes, and $\bm{x}=(x_i)_{n}$ be the vector of holding amount of options. Partition the notations accordingly as $\bm{u}=(\bm{p};\bm{d})$, $\Delta\bm{u}=(\Delta\bm{p};\Delta\bm{d})$ and $\bm{w}=(\bm{x};\bm{y})$. The value change of the optioned portfolio $\Delta v(\bm{x},\bm{y})$ can be expressed as
\begin{eqnarray}
&&\Delta v(\bm{x},\bm{y})=\sum_{i=1}^nx_i\Delta d_i+\sum_{i=1}^my_i\Delta p_i=\Delta\bm{d}^\top\bm{x}+\Delta\bm{p}^\top\bm{y}.
\end{eqnarray}

Generally, the price of derivatives is a nonlinear function of factors related to underlying stocks. A common method to calculate the value change of derivatives is the Delta-Gamma approximation, which is the second-order Taylor expansion of the price of derivatives with respect to the prices of underlying stocks. According to the definition of ``Greeks''(see, e.g., \cite{glasserman2004}, \cite{hull2009}), denote
\begin{eqnarray*}
&&\bm{\delta}^i = \left(\begin{array}{ccc}\displaystyle{\frac{\partial d_i}{\partial{p_1}}}, & \cdots, & \displaystyle{\frac{\partial d_i}{\partial{p_m}}}\end{array}\right)^{\top},~\Gamma^i =\left(\begin{array}{ccc}
                \displaystyle{\frac{\partial^2 d_i }{\partial p_{1}^2}} &\cdots & \displaystyle{\frac{\partial^2 d_i }{\partial p_{1}\partial p_{m}}}\\ \vdots &\ddots &\vdots\\
                \displaystyle{\frac{\partial^2 d_i }{\partial p_{m}\partial p_{1}}} &\cdots & \displaystyle{\frac{\partial^2 d_i }{\partial p_{m}^2}}\end{array}
\right) ~\mbox{and}~ \theta^i = \frac{\partial d_i}{\partial t}.
\end{eqnarray*}
Then, for a given portfolio $(\bm{x},\bm{y})$, the value change of the portfolio can be approximated as
\begin{eqnarray}
\Delta v(\bm{x},\bm{y})&=&\left(\sum_{i=1}^nx_i\bm{\delta}^i+\bm{y}\right)^\top\Delta
\bm{p}+\frac{1}{2}\Delta \bm{p}^\top\left(\sum_{i=1}^n
x_i\Gamma^i\right)\Delta \bm{p}+\left(\sum_{i=1}^nx_i\theta^i\right)\Delta
t\nonumber\\
&=&\bm{\delta}^{\top}\Delta
\bm{p}+\frac{1}{2}\Delta \bm{p}^\top\Gamma\Delta \bm{p}
+\theta\Delta t,
\label{value-change-appr}
\end{eqnarray}
where $\bm\delta = \sum\limits_{i=1}^nx_i\bm{\delta}^i+\bm y,~ \Gamma = \sum\limits_{i=1}^n x_i\Gamma^i \mbox{~ and~ }
\theta = \sum\limits_{i=1}^n x_i\theta^i$.

By Assumption 1, the conditional mean and variance of the value change of the portfolio are provided as follows.
\begin{proposition}\label{moment}
Given the condition $\mathscr{G}=\{\Delta p_i=-k_i, i\in\mathcal{I}\}$, the conditional mean and variance of $\Delta v(\bm{x},\bm{y})$ are given by
\begin{eqnarray*}
\mathbb{E}\left(\Delta v(\bm{x},\bm{y})|\mathscr{G}\right)&=&\bm{g}^{\top}\bm x+\bm{h}^{\top}\bm{y},
\\
\mathbb{V}(\Delta v(\bm{x},\bm{y})|\mathscr{G})&=&\left(\bm{x}^{\top},\bm{y}_{\mathcal{J}}^{\top}\right)R\left(\begin{array}{c}\bm{x}\\\bm{y}_{\mathcal{J}}\end{array}\right)+\frac{1}{2}{\bm x}^{\top}S\bm{x},
\end{eqnarray*}
where
\begin{eqnarray*}
\bm{g}&=&(g_i)_{n}=\left(\frac{1}{2}\bm{h}^{\top}\Gamma^i\bm{h}+\frac{1}{2}{\rm tr}(\Gamma^{i}_{\mathcal{JJ}}E)+(\bm{\delta}^i)^{\top}\bm{h}+\theta^i\Delta t\right)_{n},\\
\bm{h}&=&\left\{\begin{array}{lll}
\bm{h}_{\mathcal{I}}&=&-\bm{k},\\
\bm{h}_{\mathcal{J}}&=&\bm{c},\end{array}\right.\\
R&=&\left(\Gamma^1_{\mathcal{J}\cdot}\bm{h}+\bm{\delta}^1_{\mathcal{J}},\cdots,\Gamma^n_{\mathcal{J}\cdot}\bm{h}+\bm{\delta}^n_{\mathcal{J}},{\rm I}\right)^{\top}
E\left(\Gamma^1_{\mathcal{J}\cdot}\bm{h}+\bm{\delta}^1_{\mathcal{J}},\cdots,\Gamma^n_{\mathcal{J}\cdot}\bm{h}+\bm{\delta}^n_{\mathcal{J}},{\rm I}\right),\\
S&=&(s_{ij})_{n\times n}=\left({\rm tr}\left(E\Gamma^i_{\mathcal{J}\mathcal{J}}E\Gamma^j_{\mathcal{J}\mathcal{J}}\right)\right)_{n\times n}.
\end{eqnarray*}
Here, ${\rm tr}(\cdot)$ denotes the trace of a matrix.
\end{proposition}
\begin{proof}
See Appendix A.
\end{proof}

Even though the conditional mean and variance associated with the distress event are known, it remains difficult to calculate the systemic risk without further information about the distribution of the portfolio. In particular, to facilitate the computation in portfolio optimization, we need to derive an analytical form of the systemic risk measure. According to (\ref{value-change-appr}), given distressed event $\mathscr{G}$, the conditional value change of the portfolio can be expressed as
\begin{eqnarray}
&&\Delta v^{\mathscr{G}}(\bm{x},\bm{y})=
(\bm{\delta}_{\mathcal{J}}-\Gamma_{\mathcal{J}\mathcal{I}}\bm{k})^\top\Delta\bm{p}_{\mathcal{J}}+\frac{1}{2}\Delta\bm{p}_{\mathcal{J}}^\top\Gamma_{\mathcal{J}\mathcal{J}}\Delta \bm{p}_{\mathcal{J}}+c_0,\label{con-value-change-appr}
\end{eqnarray}
where $c_0=\theta\Delta t-\bm{\delta}_{\mathcal{I}}^{\top}\bm{k}+\frac{1}{2}\bm{k}^\top\Gamma_{\mathcal{I}\mathcal{I}}\bm{k}$ is a constant and the conditional distribution of $\Delta p_{\mathcal{J}}$ is still a normal distribution. By (\ref{con-value-change-appr}) we can see that the conditional value change of the portfolio is equal to the sum of a linear term and a quadratic term of a normally distributed vector and a constant.

Notice that the conditional covariance matrix $E$ is a semidefinite matrix and can be decomposed as $E=E^{\frac{1}{2}}E^{\frac{1}{2}}$, where $E^{\frac{1}{2}}$ is a symmetric matrix. Thus $E^{\frac{1}{2}}\Gamma_{\mathcal{J}\mathcal{J}}E^{\frac{1}{2}}$ is also a symmetric matrix and it can be further decomposed as $E^{\frac{1}{2}}\Gamma_{\mathcal{J}\mathcal{J}}E^{\frac{1}{2}}=D\Lambda D^{\top}$, where $D$ is an orthogonal matrix and $\Lambda$ is a diagonal matrix consisting of the eigenvalues $(\lambda_1,\cdots,\lambda_{m'})$ of $E^{\frac{1}{2}}\Gamma_{\mathcal{J}\mathcal{J}}E^{\frac{1}{2}}$. Assume $\lambda_1,\cdots,\lambda_h$ $(h\leq m')$ are the nonzero eigenvalues without loss of generality. Then the conditional value change of the portfolio can be reformulated as
\begin{eqnarray}\label{reform_quad}
&&\Delta v^{\mathscr{G}}(\bm{x},\bm{y})=\frac{1}{2}\sum_{i=1}^h\lambda_i\left (q_i+\frac{\iota_i}{\lambda_i}\right)^2+\sum_{i=h+1}^{m'}\iota_iq_i+\tau,\label{chi}
\end{eqnarray}
where $q_i$, $i=1,\cdots,m'$, are independent normal random variables with unit variance, $\iota_i$ and $\tau$ are constants. The details of derivation of (\ref{reform_quad}) are provided in Appendix D. The following lemma from \cite{zhu2020} states the asymptotic normality of $\Delta v^{\mathscr{G}}(\bm{x},\bm{y})$.

\begin{lemma}\label{asy_normal}
Suppose $\zeta_i$, $i=1,\cdots,h$, are independent normal random variables with mean $\nu_i$ and unit variance, i.e., $\zeta_i\sim\mathcal{N}(\nu_i,1)$. Denote $\xi_h=\sum_{i=1}^h\omega_i\zeta^2_i$. If $\lim\limits_{h\rightarrow+\infty}\frac{\left(\sum_{i=1}^{h}\omega_i^k\right)^{\frac{1}{k}}}{\sqrt{\sum_{i=1}^{h}\omega_i^2\left(\frac{1}{2}+\nu_i^2\right)}}=0$ for each integer $k\geq3$, then $\frac{\xi_h-\mathbb{E}(\xi_h)}{\sqrt{\mathbb{V}(\xi_h)}} \rightsquigarrow \mathcal{N}(0,1)$  as $h\rightarrow+\infty$, where $\rightsquigarrow$ means convergence in distribution.
\end{lemma}

In plain language, Lemma \ref{asy_normal} says that a linear combination of a sufficient number of independent noncentral
Chi-square random variables is nearly normally distributed if there is no dominant term within these random variables. Notice that $\Delta v^{\mathscr{G}}(\bm{x},\bm{y})$ is  the sum of a linear combination of independent noncentral Chi-square random variables, a linear combination of independent normal random variables and a constant, where the first term is also independent to the second term. According to Lemma \ref{asy_normal} and the fact that the sum of two independent normal variables remains a normal variable, we have that $\Delta v^{\mathscr{G}}(\bm{x},\bm{y})$ is approximately normally distributed.

As remarked in \cite{zhu2020}, in general, a diversified optioned portfolio satisfies the conditions of asymptotic normal distribution. Thus we can reasonably approximate the conditional distribution of the value change of a relatively diversified optioned portfolio with a normal distribution and calculate its CoVaR as
\begin{eqnarray*}
&&CoVaR_q^{\Delta v(\bm{x},\bm{y})|\mathscr{G}}=\alpha_q \mathbb{V}^{\frac{1}{2}}\left(\Delta v(\bm{x},\bm{y})|\mathscr{G}\right)-\mathbb{E}(\Delta v(\bm{x},\bm{y})|\mathscr{G}).
\end{eqnarray*}
In the sequel, most of the analyses are based on the above formulation of systemic risk measure.

\subsection{Controllability of systemic risk for optioned portfolio}

In Section 2, we demonstrate that the systemic risk of pure stock portfolios is usually uncontrollable due to the contagion effect and the seesaw effect. Actually, the contagion effect is caused by the strong correlation among assets, while the seesaw effect is caused by the large loss of the distressed event. Now we show that the introduction of options into the portfolio can reduce both of these two effects, thus achieves the goal of systemic risk control.  To facilitate the comparison between portfolios with and without options, we first introduce the concept of optioned assets.

Suppose there are $n_i$ options written on stock $i$ ($n = \sum_{i=1}^m n_i$), and reorder these options from 1 to $n_i$. Denote the price of the $j$th option written on stock $i$ by $d_{ij}$ and the corresponding ``Greeks'' as $\delta_{ij}$, $\gamma_{ij}$ and $\theta_{ij}$, respectively.

\begin{definition}\label{option}
A portfolio of a given stock and the options written on it is called an optioned asset if its value is equal to the price of the stock. More specifically, the price of optioned asset $i\in\{1,\cdots,m\}$ is defined as
\begin{eqnarray*}
&&\varphi_i\triangleq b_ip_i+\sum_{j=1}^{n_i}a_{ij}d_{ij},
\end{eqnarray*}
where $b_i,a_{i1},\cdots,a_{in_i}$
are the amounts of stock and options such that
\begin{eqnarray}\label{price_equal}
&&b_ip_i+\sum_{j=1}^{n_i}a_{ij}d_{ij}=p_i.
\end{eqnarray}
\end{definition}

Restriction (\ref{price_equal}) ensures that the comparison of the stock portfolio and the optioned portfolio can be translated to the comparison of the stock portfolio and the portfolio of optioned assets.
Let $\bm{z}=(z_{i})_{m}$ be a vector of holding amounts of optioned assets in the portfolio. Notice that by condition (\ref{price_equal}), $\bm{z}$ can also denote a feasible stock portfolio with the same initial wealth, mathematically $\sum_{i=1}^{m}z_{i}\varphi_i=\sum_{i=1}^{m}z_{i}p_i$.  For given $b_i$ and $a_{ij}$, the portfolio of optioned assets $\Delta v(\bm{z})=\sum_{i=1}^mz_i\Delta\varphi_i$ is also an optioned portfolio with $x_{ij}=z_ia_{ij}$ and $y_i=z_ib_i$.

By Definition \ref{option}, the price change of optioned asset $i\in\{1,\cdots,m\}$ is given by
\begin{eqnarray*}
&&\Delta\varphi_i=\delta_i\Delta p_i+\frac{1}{2}\gamma_{i}\Delta p_{i}^2+\theta_{i}\Delta t,
\end{eqnarray*}
where
\begin{eqnarray}\label{greeks-1}
&&\delta_i=b_i+\sum_{j=1}^{n_i}a_{ij}\delta_{ij},~\gamma_i=\sum_{j=1}^{n_i}a_{ij}\gamma_{ij}~\mbox{ and }~\theta_i=\sum_{j=1}^{n_i}a_{ij}\theta_{ij}.
\end{eqnarray}

In the sequel, we illustrate the controllability of systemic risk for optioned portfolios from two aspects, i.e., the risk contagion among optioned assets and the extreme loss of each optioned asset. To this end, denote $\rho_{ij}$ and $\rho_{ij}^\mathscr{G}$ as the correlation coefficient and conditional correlation coefficient between stock $i$ and $j$, $\varrho_{ij}$ and $\varrho_{ij}^\mathscr{G}$ as the correlation coefficient and conditional correlation coefficient between optioned asset $i$ and $j$, respectively.

\begin{proposition}\label{corr}
For any given $b_i,  a_{i1},\cdots,a_{in_i}$ and $b_j,  a_{j1},\cdots,a_{jn_j} \in \mathbb{R}$, the following inequalities
\begin{eqnarray*}
&& |\varrho_{ij}|\leq|\rho_{ij}|
~\mbox{and}~
|\varrho_{ij}^\mathscr{G}|\leq|\rho_{ij}^\mathscr{G}|
\end{eqnarray*}
always hold. Furthermore, the strict inequalities hold under very mild conditions.
\end{proposition}

\begin{proof}
Since the proof of the unconditional and conditional case are similar, we just provide the proof of the unconditional case. To make the proof simpler, we consider the correlation between assets indexed by 1 and 2 without loss of generality. Let $\Delta\widetilde{p}_1=\Delta p_1-\mu_1$ and $\Delta\widetilde{p}_2=\Delta p_2-\mu_2$. By the properties of multivariate normal distribution, we have ${\rm Cov}(\Delta\widetilde{p}_1,\Delta\widetilde{p}^2_2)={\rm Cov}(\Delta\widetilde{p}^2_1,\Delta\widetilde{p}_2)=0$ and ${\rm Cov}(\Delta\widetilde{p}^2_1,\Delta\widetilde{p}^2_2)=2\sigma_{12}^2$. Then the covariance of optioned assets 1 and 2 can be reformulated as
\begin{eqnarray*}
\text{Cov}(\Delta\varphi_1,\Delta\varphi_2)
&=&\delta_1\delta_2\text{Cov}\left(\Delta p_1,\Delta p_2\right)+\frac{1}{2}\delta_1\gamma_2\text{Cov}(\Delta p_1,\Delta p_2^2)+\frac{1}{2}\delta_2\gamma_1\text{Cov}(\Delta p_1^2,\Delta p_2)\\
&&+\frac{1}{4}\gamma_1\gamma_2\text{Cov}(\Delta p_1^2,\Delta p_2^2)\\
&=&\delta_1\delta_2\sigma_{12}+\delta_1\gamma_2\mu_2\sigma_{12}+\delta_2\gamma_1\mu_1\sigma_{12}+\frac{1}{4}(4\gamma_1\gamma_2\mu_1\mu_2\sigma_{12}+2\gamma_1\gamma_2\sigma_{12}^2)\\
&=&(\delta_1+\gamma_1\mu_1)(\delta_2+\gamma_2\mu_2)\sigma_{12}+\frac{1}{2}\gamma_1\gamma_2\sigma_{12}^2,
\end{eqnarray*}
where the second equality holds since
\begin{eqnarray*}
~{\rm Cov}(\Delta p_1,\Delta p^2_2)={\rm Cov}(\Delta\widetilde{p}_1,\Delta\widetilde{p}^2_2)+2\mu_2\sigma_{12} ~\mbox{and}~ {\rm Cov}(\Delta{p}^2_1,\Delta{p}^2_2)={\rm Cov}(\Delta\widetilde{p}^2_1,\Delta\widetilde{p}^2_2)+4\mu_1\mu_2\sigma_{12}.
\end{eqnarray*}
Notice that replacing the subscript 2 or 1 with 1 or 2 in the above equalities yields the variance of $\Delta\varphi_1$ or $\Delta\varphi_2$.

Let $\tau_i = \delta_i + \gamma_i\mu_i$, $i=1,2$. Then
\begin{eqnarray*}
\varrho_{12}&=&\frac{\text{Cov}(\Delta\varphi_1,\Delta\varphi_2)}{\sqrt{\mathbb{V}(\Delta\varphi_1)\mathbb{V}(\Delta\varphi_2)}}\\
&=&\frac{\frac{1}{2}\gamma_1\gamma_2\sigma_{12}+\tau_1\tau_2}
{\sqrt{(\frac{1}{2}\gamma_1^2\sigma_{11}+\tau_1^2)(\frac{1}{2}\gamma_2^2\sigma_{22}+\tau_2^2)}}\frac{\sigma_{12}}{\sqrt{\sigma_{11}\sigma_{22}}}\\
&=&\frac{\frac{1}{2}\gamma_1\gamma_2\sigma_{12}+\tau_1\tau_2}
{\sqrt{(\frac{1}{2}\gamma_1^2\sigma_{11}+\tau_1^2)(\frac{1}{2}\gamma_2^2\sigma_{22}+\tau_2^2)}}\rho_{12}.
\end{eqnarray*}
According to Cauchy inequality,
\begin{eqnarray*}
\sqrt{(\frac{1}{2}\gamma_1^2\sigma_{11}+\tau_1^2)(\frac{1}{2}\gamma_2^2\sigma_{22}+\tau_2^2)} &\geq& |\frac{1}{2}\gamma_1\gamma_2\sqrt{\sigma_{11}\sigma_{22}}|+|\tau_1\tau_2|\\
&\geq&|\frac{1}{2}\gamma_1\gamma_2\sqrt{\sigma_{11}\sigma_{22}}\rho_{12}|+|\tau_1\tau_2|\\
&\geq&|\frac{1}{2}\gamma_1\gamma_2\sigma_{12}+\tau_1\tau_2|.
\end{eqnarray*}
Thus we have $|\varrho_{12}|\leq|\rho_{12}|$. The first, second and third equalities hold if and only if $\left|\gamma_1\tau_2\sqrt{\sigma_{11}}\right|=\left|\gamma_2\tau_1\sqrt{\sigma_{22}}\right|$, $|\rho_{12}|=1$ and $\rho_{12}\tau_1\tau_2\gamma_1\gamma_2\geq0$, respectively. The proof is completed.
\end{proof}

\textbf{Remark 1.} By introducing options, we derive an optioned portfolio whose constituent assets (optioned assets) are less correlated. Therefore, the effect of risk contagion in the optioned portfolio can be lower than that in the stock portfolio.  We call this effect as {\it correlation hedging}.  Furthermore, based on Proposition \ref{corr}, we can pinpoint the extent to which correlations can be hedged under some conditions.

\begin{corollary}\label{perfect_hedge} If $\delta_i=0$ and $\gamma_i=0$, then ${\rm Cov}(\Delta\varphi_i,\Delta\varphi_j)=0$ for any $j\in\{1,\cdots,m\}$.
\end{corollary}
\begin{proof}
If $\delta_i=0$ and $\gamma_i=0$, by Proposition \ref{corr}, the covariance of optioned assets $i$ and $j\in\{1,\cdots,m\}$ turns to
\begin{eqnarray*}
&&\text{Cov}(\Delta\varphi_i,\Delta\varphi_j)=\frac{1}{2}\gamma_i\gamma_j\sigma_{ij}^2+(\gamma_i\mu_i+\delta_i)(\gamma_j\mu_j+\delta_j)\sigma_{ij}=0.
\end{eqnarray*}
The proof is completed.
\end{proof}
Corollary \ref{perfect_hedge} implies that the correlations of optioned assets can reach zeros
under some conditions, which are determined by the vector ($b_i,a_{i1},\cdots,a_{in_i}$) that defines the optioned asset.
By (\ref{price_equal}) and (\ref{greeks-1}), the sufficient condition for $\delta_i=0$ and $\gamma_i=0$ is that there exits a solution to linear equations $A{\bm x}=\bm{a}$ with respect to $\bm{x}$, where
\begin{eqnarray*}
&&A = \left(\begin{array}{cccc}p_i&d_{i1}&\cdots&d_{in_i}\\
1&\delta_{i1}&\cdots&\delta_{in_i}\\
0&\gamma_{i1}&\cdots&\gamma_{in_i}
\end{array}\right)
~\mbox{and}~ \bm{a}=
\left(\begin{array}{c}
p_i\\
0\\
0
\end{array}\right).
\end{eqnarray*}
Obviously, ${\rm Rank}(A)={\rm Rank}(A,\bm{a})$, or more strictly, ${\rm Rank}(A)=3$ is a sufficient condition to guarantee  $\delta_i=0$ and $\gamma_i=0$. In other words, if there exists three options written on stock $i$ with $(d_{ij},\delta_{ij},\gamma_{ij})$ $(j=1,2,3)$ being linearly independent, there exists $(b_i,a_{i1},\cdots,a_{in_i})$ such that $\delta_i=0$ and $\gamma_i=0$.

We call the condition in Corollary \ref{perfect_hedge}, $\delta_i=0$ and $\gamma_i=0$, as {\it Delta-Gamma neutrality}, which is a sufficient but not necessary condition for the zero correlation. However, as we show in the following, when all of the correlations are hedged to zeros, we can also derive Delta-Gamma neutrality for at least $m-2$ assets with some additional conditions.

\begin{corollary}\label{delta-gamma-hedge}
Suppose $m\geq3$ and $\rho_{ij}>0$ for any $i,j\in\{1,\cdots,m\}$. If ${\rm Cov}(\Delta\varphi_i,\Delta\varphi_j)=0$ for any $i,j\in\{1,\cdots,m\}$, then there exists a subset $\mathcal{D}\subseteq\{1,\cdots,m\}$ and $|\mathcal{D}|\geq m-2$ such that $\delta_i=0$ and $\gamma_i=0$ for any $i\in\mathcal{D}$.
\end{corollary}
\begin{proof}
See Appendix B.
\end{proof}
\textbf{Remark 2.} Corollaries \ref{perfect_hedge} and \ref{delta-gamma-hedge} illustrate the possibility to make all correlations of optioned assets to zeros, which can sufficiently reduce the risk contagion among assets. Notice that the portfolio of the optioned assets achieving Delta-Gamma neutrality can only earns a risk-free return. Therefore, compared with the lower bound of systemic risk of stock portfolios in Proposition \ref{controlability}, the systemic risk of the optioned portfolio can be controlled to a much lower level, even to zero, if all correlations of optioned assets are hedged to zeros.  In addition, Corollaries \ref{perfect_hedge} and \ref{delta-gamma-hedge} just depict an extreme case, which also indicates that, with options, we can hedge the correlation to any degree. Notice that the assumption $\rho_{ij}>0$ is the true for most case in practice, i.e., most of the stocks are positively correlated.

It is worth noting that the correlation hedging has no other restrictions on the holding amount of options. However, this does not mean that systemic risk can be reduced by adding any amount of options to the portfolio. The leverage effect of options on the other hand increases the volatility of the portfolio, and finally might increase the systemic risk. Therefore, we also need to verify the return and risk profile of optioned assets compared with stock assets.

\begin{proposition}\label{dominate}
Under some mild conditions, there exist $(\delta_i,\gamma_i,\theta_i)\in\mathbb{R}^{3}$ for $i\in\{1,\cdots,m\}$ such that
\begin{enumerate}[(i)]
    \item
$\mathbb{E}(\Delta\varphi_i)>\mathbb{E}(\Delta p_i)$  and
$\mathbb{V}(\Delta\varphi_i)<\mathbb{V}(\Delta p_i)$;
    \item
$\mathbb{E}(\Delta\varphi_i|\mathscr{G})>\mathbb{E}(\Delta p_i|\mathscr{G})$ and
$\mathbb{V}(\Delta\varphi_i|\mathscr{G})<\mathbb{V}(\Delta p_i|\mathscr{G})$ for $i\in\mathcal{J}$, and $\mathbb{E}(\Delta\varphi_i|\mathscr{G})>\mathbb{E}(\Delta p_i|\mathscr{G})$ for $i\in\mathcal{I}$, given the condition $\mathscr{G}=\{\Delta p_i=-k_i,i\in\mathcal{I}\}$.
\end{enumerate}
\end{proposition}

\begin{proof}
Notice that $\text{Cov}(\Delta p_i,\Delta p_i^2)=2\sigma_{ii}\mu_i$ and $\mathbb{V}(\Delta p_i^2)=2\sigma_{ii}^2+4\sigma_{ii}\mu_i^2$. For any $i\in\{1,\cdots,m\}$, the unconditional mean and variance of $\Delta\varphi_i$ are given by
\begin{eqnarray*}
&&\mathbb{E}(\Delta\varphi_i)=\mu_i\delta_i+\frac{\sigma_{ii}+\mu_i^2}{2}\gamma_i+\theta_i\Delta t~\mbox{and}~
\mathbb{V}(\Delta\varphi_i)=\frac{\sigma_{ii}^2}{2}\gamma_i^2+\sigma_{ii}\left(\mu_i\gamma_i+\delta_i\right)^2.
\end{eqnarray*}
Thus $\mathbb{E}(\Delta\varphi_i)>\mathbb{E}(\Delta p_i)$  and
$\mathbb{V}(\Delta\varphi_i)<\mathbb{V}(\Delta p_i)$ can be reformulated as
\begin{eqnarray}\label{P4.1}
&&\mu_i\delta_i+\frac{\sigma_{ii}+\mu_i^2}{2}\gamma_i+\theta_i\Delta t-\mu_i>0~\mbox{and}~\frac{\sigma_{ii}}{2}\gamma_i^2+\left(\mu_i\gamma_i+\delta_i\right)^2-1<0.
\end{eqnarray}

First, we prove that both $(i)$ and $(ii)$ hold for $i\in\mathcal{J}$. Conditioning on $\mathscr{G}$, the conditional mean and conditional variance of $\Delta\varphi_i$ are given by
\begin{eqnarray*}
&&\mathbb{E}(\Delta\varphi_i|\mathscr{G})= c_i\delta_i+\frac{e_{ii}+c_i^2}{2}\gamma_i+\theta_i\Delta t~\mbox{and}~
\mathbb{V}(\Delta\varphi_i|\mathscr{G})=\frac{e_{ii}^2}{2}\gamma_i^2+e_{ii}\left(c_i\gamma_i+\delta_i\right)^2.
\end{eqnarray*}
Then $\mathbb{E}(\Delta\varphi_i|\mathscr{G})>\mathbb{E}(\Delta p_i|\mathscr{G})$ and
$\mathbb{V}(\Delta\varphi_i|\mathscr{G})<\mathbb{V}(\Delta p_i|\mathscr{G})$ can be reformulated as
\begin{eqnarray}\label{P4.2}
&&c_i\delta_i+\frac{e_{ii}+c_i^2}{2}\gamma_i+\theta_i\Delta t-c_i>0~\mbox{and}~\frac{e_{ii}}{2}\gamma_i^2+\left(c_i\gamma_i+\delta_i\right)^2-1<0.
\end{eqnarray}
We prove that there exists $(\delta_i,\gamma_i,\theta_i)$ such that both inequalities (\ref{P4.1}) and (\ref{P4.2}) hold. To facilitate the proof, we reorder them as the following inequality system
\begin{eqnarray}
\mu_i\delta_i+\frac{\sigma_{ii}+\mu_i^2}{2}\gamma_i+\theta_i\Delta t-\mu_i>0,&&c_i\delta_i+\frac{e_{ii}+c_i^2}{2}\gamma_i+\theta_i\Delta t-c_i>0,\label{P4.3}\\
\frac{\sigma_{ii}}{2}\gamma_i^2+\left(\mu_i\gamma_i+\delta_i\right)^2-1<0,&&\frac{e_{ii}}{2}\gamma_i^2+\left(c_i\gamma_i+\delta_i\right)^2-1<0.\label{P4.4}
\end{eqnarray}
We firstly prove that inequalities in (\ref{P4.4}) are true for some $\delta_i$ and $\gamma_i$. Based on it, then we prove that inequalities in (\ref{P4.3}) are true.

Let $\bm{x}=(\delta_i,\gamma_i)^{\top}$. Then inequalities (\ref{P4.4}) can be reformulated as
\begin{eqnarray}\label{P4.5}
&&\bm{x}^{\top}A_1\bm{x}-1<0~\mbox{and}~\bm{x}^{\top}A_2\bm{x}-1<0,
\end{eqnarray}
where $A_1 = \left(\begin{array}{cc}1&\mu_i\\\mu_i&\frac{\sigma_{ii}+2\mu_i^2}{2}\end{array}\right)$ and $A_2 = \left(\begin{array}{cc}1&c_i\\c_i&\frac{e_{ii}+2c_i^2}{2}\end{array}\right)$ are both definite matrices. Obviously, there exists at least one feasible solution to inequality system (\ref{P4.5}). For instance, $\bm{x}=\bm{0}$ or other vectors that are sufficiently close to it.

Suppose $(\overline{\delta}_i,\overline{\gamma}_i)$ is a feasible solution to (\ref{P4.5}). Replacing it into to (\ref{P4.3}) derives
\begin{eqnarray*}
&&\mu_i\overline{\delta}_i+\frac{\sigma_{ii}+\mu_i^2}{2}\overline{\gamma}_i+\theta_i\Delta t-\mu_i>0~\mbox{and}~
c_i\overline{\delta}_i+\frac{e_{ii}+c_i^2}{2}\overline{\gamma}_i+\theta_i\Delta t-c_i>0
\end{eqnarray*}
with respect to $\theta_i$, which holds  for $\theta_i=\overline{\theta}_i$, where
\begin{eqnarray*}
&&\overline{\theta}_i>\max\left\{(1-\overline{\delta}_i)\mu_i-\frac{\sigma_{ii}+\mu_i^2}{2}\overline{\gamma}_i,(1-\overline{\delta}_i)c_i-\frac{e_{ii}+c_i^2}{2}\overline{\gamma}_i\right\}/\Delta t.
\end{eqnarray*}
Thus the conclusion for $i\in\mathcal{J}$ holds.

Next, we prove that both $(i)$ and $(ii)$  hold for $i\in\mathcal{I}$. Given $\mathscr{G}$, $\Delta\varphi_i$ and $\Delta p_i$ are both constants and $\mathbb{E}(\Delta\varphi_i|\mathscr{G})>\mathbb{E}(\Delta p_i|\mathscr{G})$ can be reformulated as
\begin{eqnarray}\label{P4.7}
&&-k_i\delta_i+\frac{k_i^2}{2}\gamma_i+\theta_i\Delta t+k_i>0.
\end{eqnarray}
Combining with $\mathbb{E}(\Delta\varphi_i)>\mathbb{E}(\Delta p_i)$ and $\mathbb{V}(\Delta\varphi_i)<\mathbb{V}(\Delta p_i)$, we need to find a feasible solution $(\delta_i,\gamma_i,\theta_i)$ to make both of (\ref{P4.7}) and (\ref{P4.1}) hold, which are equivalent to
\begin{eqnarray}\label{P4.8}
&&-k_i\delta_i+\frac{k_i^2}{2}\gamma_i+\theta_i\Delta t+k_i>0,~\mu_i\delta_i+\frac{\sigma_{ii}+\mu_i^2}{2}\gamma_i+\theta_i\Delta t-\mu_i>0
\end{eqnarray}
and
\begin{eqnarray}\label{P4.9}
&&\frac{\sigma_{ii}}{2}\gamma_i^2+\left(\mu_i\gamma_i+\delta_i\right)^2-1<0.
\end{eqnarray}
Similar to (\ref{P4.4}), there must exist a solution $(\overline{\delta}_i,\overline{\gamma}_i)$ to make (\ref{P4.9}) hold. Given $(\overline{\delta}_i,\overline{\gamma}_i)$, inequalities (\ref{P4.8}) hold with $\theta_i=\overline{\theta}_i$, where
\begin{eqnarray*}
&&\overline{\theta}_i>\max\left\{(\overline{\delta}_i-1)k_i-\frac{k_i^2}{2}\overline{\gamma}_i,(1-\overline{\delta}_i)\mu_i-\frac{\sigma_{ii}+\mu_i^2}{2}\overline{\gamma}_i\right\}/\Delta t.
\end{eqnarray*}
The proof is completed.
\end{proof}

\textbf{Remark 3.} Proposition \ref{dominate} shows that the introduction of options can alleviate the impact of systemic risk on each single stock. For the stocks suffered initial losses, i.e., stocks indexed by $\mathcal{I}$, options can hedge the extreme losses. For stocks affected by the initial losses, or stocks indexed by $\mathcal{J}$, options can alleviate the impact of systemic risk by increasing their conditional mean. In addition, options can reduce the uncertainty by decreasing the conditional variance, and finally hedge the extreme losses. We call this effect as {\it extreme loss hedging}. Besides the merits in terms of systemic risk control, the introduction of options also enhances the return-risk performance.

The sufficient condition for the existence of ($\delta_i,\gamma_i,\theta_i$) in Proposition \ref{dominate} is that there exists a solution $(b_i,a_{i1},\cdots,a_{in_i})$ to the linear system $A\bm{x}=\bm{a}$, where $A$ and $\bm{a}$ are given by
\begin{eqnarray*}
&&A = \left(\begin{array}{cccc}p_i&d_{i1}&\cdots&d_{in_i}\\
1&\delta_{i1}&\cdots&\delta_{in_i}\\
0&\gamma_{i1}&\cdots&\gamma_{in_i}\\
0&\theta_{i1}&\cdots&\theta_{in_i}
\end{array}\right)
~\mbox{and}~ \bm{a}=
\left(\begin{array}{c}
p_i\\
{\delta}_i\\
{\gamma}_i\\
{\theta}_i
\end{array}\right).
\end{eqnarray*}
Similar to the discussion of Corollary \ref{perfect_hedge}, a sufficient condition of the feasibility of $({\delta}_i,{\gamma}_i,{\theta}_i)$ is that ${\rm Rank}(A)={\rm Rank}(A,\bm{a})$ or more strictly ${\rm Rank}(A)=4$.

In Propositions \ref{corr} and \ref{dominate}, we derive the correlation hedging effect and extreme loss hedging effect. Based on these two effects, now we can turn the comparison of optioned assets and stock assets to the comparison of optioned portfolios and stock portfolios.

\begin{corollary}\label{frontier} Suppose the conditions for Proposition \ref{dominate} are satisfied and further
assume that $\bm{y}>\bm{0}$, $\rho_{ij}>0$ and $\rho^{\mathscr{G}}_{ij}>0$ for any $i,j\in\{1,\cdots,m\}$. Then, for any given stock portfolio $\bm{y}$, there exists a corresponding optioned portfolio $\bm{z}$ such that
\begin{eqnarray*}
&&\mathbb{E}(\Delta v(\bm{y})) < \mathbb{E}(\Delta v(\bm{z})),~
\mathbb{V}(\Delta v(\bm{y})) > \mathbb{V}(\Delta v(\bm{z}))~\mbox{and}~CoVaR_q^{\Delta v(\bm{y})|\mathscr{G}} > CoVaR_q^{\Delta v(\bm{z})|\mathscr{G}}.
\end{eqnarray*}
\end{corollary}
\begin{proof}
See Appendix C.
\end{proof}
Corollary \ref{frontier} indicates that the introduction of options not only prevents the portfolio from extreme losses, but also ensures a better return-risk performance. By expanding the feasibility of the portfolio, the optioned portfolio has more investment opportunity to alleviate the contagion effect and seesaw effect.

In practice, there might be some additional restrictions on ($b_i,a_{i1},\cdots,a_{in_i}$). For instance, when short selling is prohibited, which is a very common trading rule in practice, we need to impose the non-negative constraint on the position of stock, i.e., $b_i\geq 0$ to ensure that the optioned assets are tradable.

\subsection{Optimization of optioned portfolio}
The previous analysis illustrates that options can improve the performance of portfolios with systemic risk control from several aspects. In this section, we derive the tractable reformulation of optimization of optioned portfolios with systemic risk control. Again we mention that, without loss of generality, we discuss the case when there is only one systemic risk constraint.

For practical application, we specify $\Omega$, the feasible set of portfolios, as
\begin{eqnarray*}
&&\Omega=\left\{(\bm{x},\bm{y}):(\bm{d}^{ask})^\top(\bm{x}-\bm{x}^0)^++(\bm{d}^{bid})^\top(\bm{x}-\bm{x}^0)^-+\bm{p}^\top(\bm{y}-\bm{y}^0)\leq k^0,\right.\\
&&~~~~~~~~~~~~~~~~~\left.d^{bid}_{i}x_i\geq l_{d_i}, d^{ask}_{i}x_i\leq u_{d_i},l_{p_i}\leq p_{i}y_i\leq u_{p_i}\right\}.
\end{eqnarray*}
where $\bm{x}^0$ and $\bm{y}^0$ are the holding amount of options and stocks at the beginning of investment period, $k^0$ is the cash on hand at the beginning of investment period, $l_{d_i}, l_{p_i}$ and $u_{d_i}, u_{p_i}$ are the lower bound and upper bound of the investment amount of options and underlying stocks, respectively, $\bm{d}^{ask}$, $\bm{d}^{bid}$ and $\bm{p}$ are the ask price, bid price vector of options and price vector of stock, and $\bm{x}^+$ and $\bm{x}^-$ are the vectors with $\max\{x_i,0\}$ and $\min\{x_i,0\}$ being the $i$th element. It is worth mentioning that the above $\Omega$ is set in a very general case for portfolio rebalance. To construct a totally new portfolio without an existing portfolio, we just need to set $\bm{x}^0=\bm 0$ and $\bm{y}^0=\bm 0$ in $\Omega$.

We show that the first and the second constraints in problem ($P$) can be both reformulated to second-order cone constraints. According to Proposition 1 of \cite{zhu2020}, the mean and variance of $\Delta v(\bm{x},\bm{y})$ are
\begin{eqnarray*}
\mathbb{E}(\Delta v(\bm{x},\bm{y}))&=&\bm{\eta}^{\top}\bm x+\bm{\mu}^{\top}\bm{y},~~\label{fir-moment}
\\
\mathbb{V}(\Delta v(\bm{x},\bm{y}))&=&(\bm{x}^{\top},\bm{y}^{\top})\Psi\left(\begin{array}{c}\bm{x}\\\bm{y}\end{array}\right)+\frac{1}{2}{\bm x}^{\top}\Phi
\bm{x},
\end{eqnarray*}
respectively, where
\begin{eqnarray*}
&&\bm\eta= (\eta_i)_{n}= \left(\frac{1}{2}\bm{\mu}^\top \Gamma^i\bm{\mu}
+\frac{1}{2}{\rm tr}(\Gamma^i\Sigma)+\left(\bm{\delta}^i\right)^\top
\bm{\mu}+\theta^i\Delta
t\right)_{n},\\
&&\Psi=\left(\Gamma^1\bm\mu+\bm\delta^1,\cdots,\Gamma^n\bm\mu+\bm\delta^n,\mathrm{I}\right)^{\top}\Sigma\left(\Gamma^1\bm\mu+\bm \delta^1,\cdots,\Gamma^n\bm\mu+\bm\delta^n, {\mathrm{I}}\right),\\
&&\Phi=\left(\phi_{ij}\right)_{n\times n}=\left({\rm tr}\left(\Gamma^i\Sigma\Gamma^j\Sigma\right)\right)_{n\times n}.
\end{eqnarray*}

We decompose $\Phi$ and $\Sigma$ as
\begin{eqnarray*}
&&\Phi=M^{\top}M ~\mbox{and}~ \Sigma=L^{\top}L.
\end{eqnarray*}
Furthermore, denote
\begin{eqnarray*}
&&H=L(\Gamma^1\bm{\mu}+\bm{\delta}^1,\cdots,\Gamma^n\bm{\mu}+\bm{\delta}^n,{\rm I}).
\end{eqnarray*}
The variance constraint $\sigma(\bm{x},\bm{y})\leq\bar\sigma$ can be reformulated as
\begin{eqnarray*}
&&\left(\begin{array}{c}\bar\sigma\\
                      H\left(\begin{array}{c}\bm{x}\\
                                             \bm{y}\end{array}\right)\\
                      \frac{1}{\sqrt{2}}M\bm{x}\end{array}\right)\succeq_{SOCP}\bm{0},
\end{eqnarray*}
which means $\left|\left|\left(H(\bm{x};\bm{y});\frac{1}{\sqrt{2}}M\bm{x}\right)\right|\right|\leq\bar\sigma$. Decompose $S$ and $E$ as
\begin{eqnarray*}
&&S=N^{\top}N~\mbox{and}~E = F^{\top}F
\end{eqnarray*}
and denote
\begin{eqnarray*}
&&Q=F\left(\Gamma^1_{\mathcal{J}\cdot}\bm{h}+\bm{\delta}^1_{\mathcal{J}},\cdots,\Gamma^n_{\mathcal{J}\cdot}\bm{h}+\bm{\delta}^n_{\mathcal{J}},{\rm I}\right).
\end{eqnarray*}
Then the constraint on CoVaR can be reformulated as
\begin{eqnarray*}
&&\left(\begin{array}{c}(\bar\rho+\bm{g}^{\top}\bm{x}+\bm{h}^{\top}\bm{y})/\alpha_q\\
                      Q\left(\begin{array}{c}\bm{x}\\
                                             \bm{y}_{\mathcal{J}}\end{array}\right)\\
                      \frac{1}{\sqrt{2}}N\bm{x}
                      \end{array}\right)\succeq_{SOCP}\bm{0}.
\end{eqnarray*}

Following \cite{zhu2020}, by introducing some additional variables to reformulate $\Omega$ as a polyhedron, the problem ($P$) can be equivalently transformed to the following second-order cone program
\begin{eqnarray*}
(P_1)&\max\limits_{\bm{x},\bm{y},\bm{\vartheta}}&\bm{\eta}^{\top}\bm{x}+\bm{\mu}^{\top}\bm{y}\\
&{\rm s.t.}&\left(\begin{array}{c}\bar\sigma\\
                      H\left(\begin{array}{c}\bm{x}\\
                                             \bm{y}\end{array}\right)\\
                      \frac{1}{\sqrt{2}}M\bm{x}\end{array}\right)\succeq_{SOCP}\bm{0},\\
&&\left(\begin{array}{c}(\bar\rho+\bm{g}^{\top}\bm{x}+\bm{h}^{\top}\bm{y})/{\alpha_q}\\
                      Q\left(\begin{array}{c}\bm{x}\\
                                             \bm{y}_{\mathcal{J}}\end{array}\right)\\
                      \frac{1}{\sqrt{2}}N\bm{x}
                      \end{array}\right)\succeq_{SOCP}\bm{0},\\
&&\sum_{i=1}^{n}\vartheta_i+\bm{p}^{\top}\bm{y}\leq k^0,\\
&&d^{ask}_i(x_i-x_i^0)\leq\vartheta_i,~~d^{bid}_i(x_i-x_i^0)\leq\vartheta_i,~~~~i=1,\cdots,n,\\
&&l_{d_i}\leq d^{bid}_ix_i,~~d^{ask}_ix_i\leq u_{d_i},~~l_{p_i}\leq p_iy_i\leq u_{p_i},~~~~i=1,\cdots,n.
\end{eqnarray*}
If $l_{d_i}=u_{d_i}=0$ and $x_i^0=0$ for $i=1,\cdots,n$, the model turns to the systemic risk constrained model of stock portfolios and is still an SOCP. Readers who are interested in the details of SOCP can refer to \cite{Alizadeh2003}.

Notice that the above reformulation is suitable for relatively diversified optioned portfolios due to the property of asymptotic normality depicted by Lemma 1. However, for a small-scale optioned portfolio, the normal distribution may not be a good approximation for the distribution of the value change. Considering the severity of systemic risk, a reasonable choice is to use the worst-case risk to measure the systemic risk. According to the definition of worst-case VaR and the Theorem 1 of \cite{El2003}, given that the conditional mean and variance of the portfolio are the only information available, the worst-case CoVaR has the closed form as
\begin{eqnarray*}
&&WCoVaR_q^{\Delta v(\bm{x},\bm{y})|\mathscr{G}}=\kappa_q \mathbb{V}^{\frac{1}{2}}(\Delta v(\bm{x},\bm{y})|\mathscr{G})-\mathbb{E}(\Delta v(\bm{x},\bm{y})|\mathscr{G}),
\end{eqnarray*}
where $\kappa_q=\sqrt{\frac{q}{1-q}}$. Notice that the difference between CoVaR and worst-case CoVaR is $\alpha_q$ and $\kappa_q$. Thus we just need to replace $\alpha_q$ with $\kappa_q$ in problem ($P_1$) to construct a model with worst-case CoVaR as the systemic risk measure.

\section{Simulation analysis and empirical study}
In this section, we carry out simulations and empirical tests to examine the performance of the proposed approach. First, in-sample test is adopted to illustrate the uncontrollability/controllability of systemic risk of stock portfolios/optioned portfolios. And then the comparisons of the optioned portfolio strategy with systemic risk control with other portfolio strategies are provided in the out-of-sample test.

We use the constituent stocks of index of Dow Jones Industrial Average (DJIA) and the corresponding options in simulation and empirical study. Two data sets deriving from the Bloomberg database are used in this section. Data set 1 is from January 2000 to December 2018 for stocks and from January 2018 to December 2018 for options. Two stocks from the index are removed due to lacking data, and thus it consists of 28 stocks. We select four types of American options, i.e., at-the-money (ATM) call option, ATM put option, 20\% out-of-the-money (OTM) call option and 20\% OTM put option. The expiration of options is selected as one year. There are 53 options left after removing those options with missing data and the bid price smaller than 0.3. Following the same data collection and processing rules, data set 2 consists of 27 constituent stocks from January 2000 to January 2021 and 80 options from January 2020 to January 2021, which consist of six types of American option, ATM call option, ATM put option, 20\% OTM call option, 20\% OTM put option, 20\% in-the-money (ITM) call option and 20\% ITM put option.

The time horizon for portfolio selection is one week, i.e., weekly data is used for the analysis. In the following study, the period of the option data set is defined as the out-of-sample period, and the period of the stock data set removing the out-of-sample period is the in-sample period.

In this section, four basic strategies/portfolios are used:
\begin{enumerate}[(1)]
  \item {\it Stock}: generated by model ($P$) with only stocks and without systemic risk constraints.
  \item {\it Stock control}: generated by model ($P$) with only stocks.
  \item {\it Optioned}: generated by model ($P$) with stocks and options and without systemic risk constraints.
  \item {\it Optioned control}: generated by model ($P$) with stocks and options.
\end{enumerate}

Throughout the following analysis, the confidence levels of VaR and CoVaR are set at $p=q=0.95$. The risk-free rate is set at 0.05. The total amount of options is restricted within thirty percent of the total wealth.  The stocks index by $\mathcal{I}$ are called {\it systemically important assets}.

All the computations are completed with a personal computer using Matlab R2020b, the SOCP is solved via CVX, a Matlab-based modeling package for convex optimization problems.

\subsection{In-sample analysis}
In this subsection, we compare the efficiency of systemic risk control of optioned portfolios with stock portfolios. The comparison is made up of three parts. We firstly show the uncontrollability of stock portfolios and the controllability of optioned portfolios. Then we compare the efficient frontiers and finally the strategy performances under different scenarios. To simplify the analysis in this subsection, we assume the initial wealth is one. The feasible set $\Omega$ is specified as
\begin{eqnarray*}
&&\Omega=\left\{(\bm{x},\bm{y}):\bm{d}^\top\bm{x}+\bm{p}^\top\bm{y}\leq 1, l_{d_i}\leq d_ix_i\leq u_{d_i}, l_{p_i}\leq p_iy_i\leq u_{p_i}\right\}
\end{eqnarray*}
without considering the bid-ask spread. For the {\it stock} strategy and {\it stock control} strategy, we set $l_{d_i}=u_{d_i}=0$, and thus $\bm{x}=\bm{0}$.

\subsubsection{Uncontrollability and controllability of systemic risk}
In this subsection, the parameters of stocks and options are generated by simulation (e.g., mean, variance, ``Greeks'' and price). The process of parameter generation is provided in Appendix E.

We simulate a 10-stock portfolio to show the seesaw effect. We first execute a {\it stock} strategy with $l_{p_i}=0$, $u_{p_i}=0.15$ and $\bar\sigma=0.03$ and specify the stock with the highest CoVaR as the systemically important asset. Then, with the same parameter settings as {\it stock} strategy, we execute the {\it stock control} strategy that restricts $\bar\rho$ to its minimal level. Finally, we add 25 options into the portfolio and execute the {\it optioned control} strategy with $l_{p_i}=l_{d_i}=0$, $u_{p_i}=u_{d_i}=0.15$, $\bar\sigma=0.03$ and $\bar\rho$ being its minimal level. Furthermore, we impose a constraint that the expected return of the {\it stock control} strategy and the {\it optioned control} strategy is no less than that of the {\it stock} strategy. The systemically important asset is stock 1 in this example. As can be seen in Figure \ref{fig3}, for the {\it stock control} strategy, the CoVaR of stock 1 decreases slightly, while CoVaRs of some other stocks increase. For the {\it optioned control} strategy, the CoVaRs of all stocks decrease dramatically. This implies that the systemic risk of the pure stock portfolio is uncontrollable while that of the optioned portfolio is controllable.

For the optioned portfolio consisting of 10 stocks and 25 options, we fix parameters associated with weights and adjust $\bar\sigma$ and $\bar\rho$ to construct an optioned portfolio that has a higher return, lower variance and lower systemic risk than the stock portfolio, as shown in Figure \ref{fig4}, which directly verifies the result of Corollary \ref{frontier}.

\begin{figure}[htp]
  \centering
  \includegraphics[width=8cm]{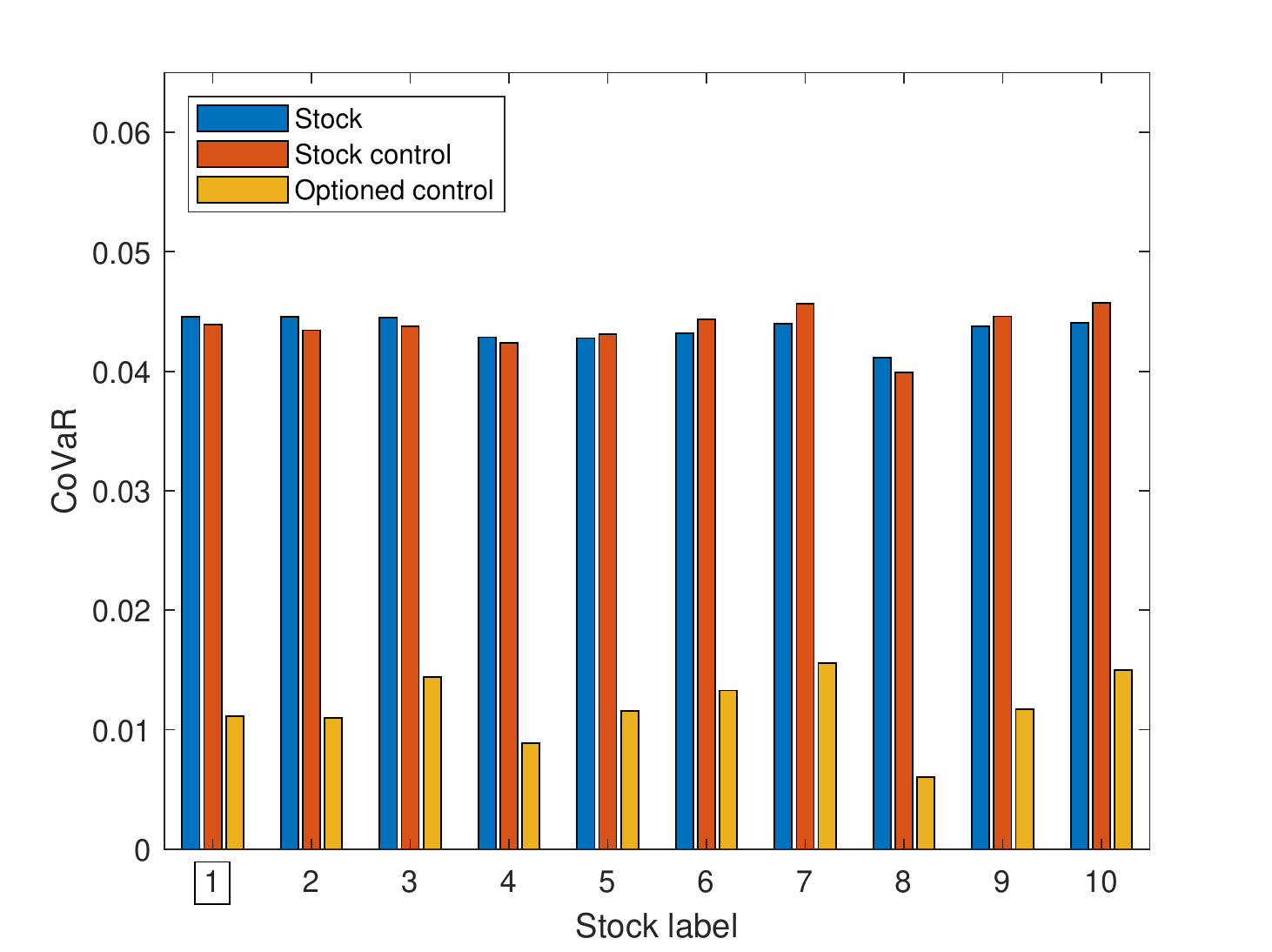}
  \caption{CoVaR of different strategies}
  \label{fig3}
\end{figure}

\begin{figure}[H]
  \centering
  \includegraphics[width=8cm]{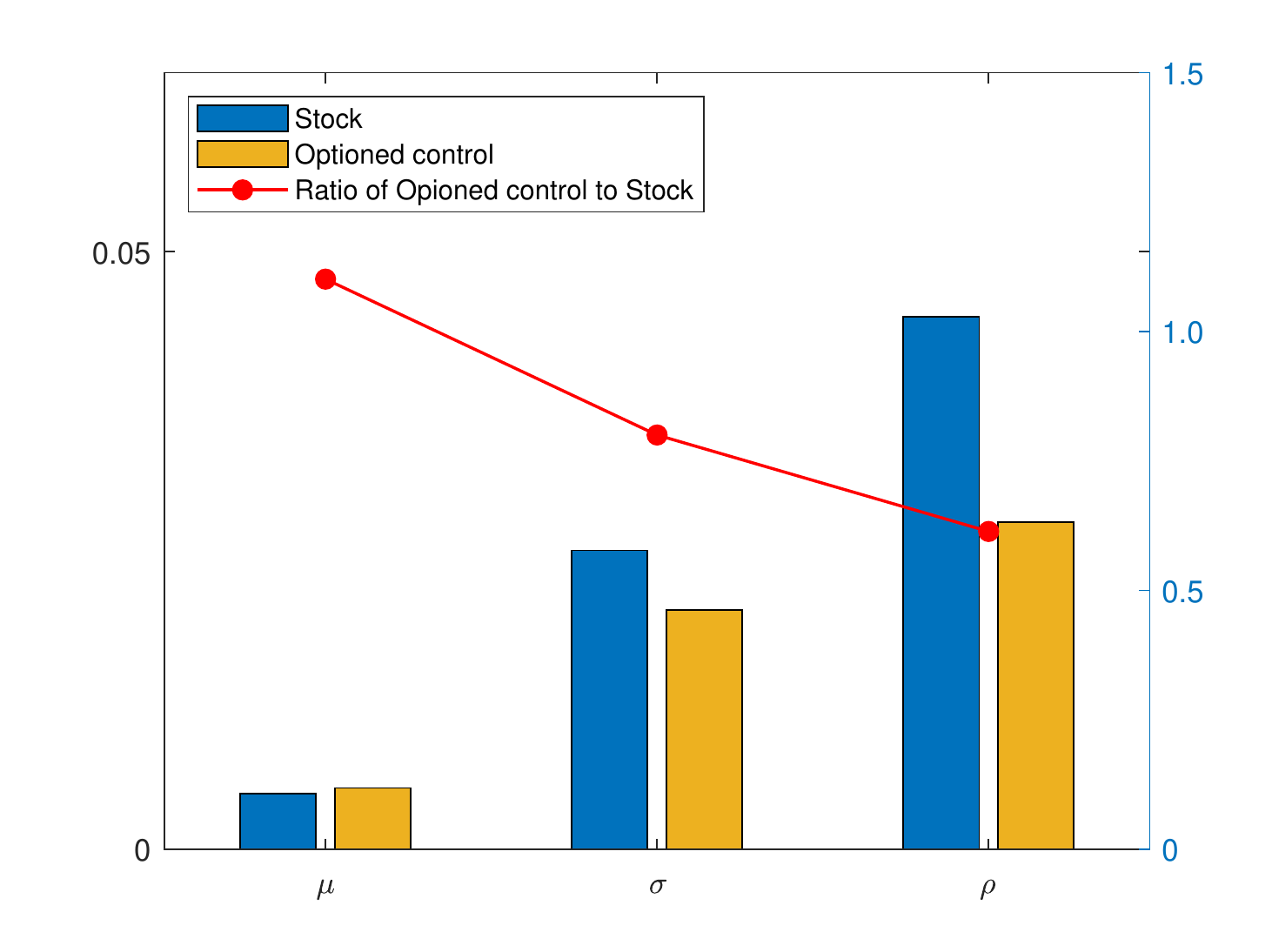}
  \caption{Performances of stock portfolio and optioned portfolio}
  \label{fig4}
\end{figure}

We further explore the effect of systemic risk control of optioned portfolios. We cumulatively add options into the portfolio and in each time execute the {\it optioned control} strategy with $\bar\rho$ being its minimal level. Other parameters are the same as the above. Once the optimal portfolio $(\bm{x},\bm{y})$ is derived, the optioned assets are naturally constructed as $b_i=\alpha_i y_i$ and $a_{ij}=\alpha_ix_{ij}$, where $\alpha_i=p_i/(y_ip_i+\sum_{j=1}^{n_i}x_{ij}d_{ij})$. Figures \ref{fig5} and \ref{fig6} show that as the number of options in the portfolio increases, the CoVaR of the optioned portfolio and the correlation of optioned assets decrease simultaneously, where in Figure \ref{fig6}, ``Max correlation", ``Ave correlation" and ``Min correlation" denote the maximal correlation, average correlation (the average of nondiagonal elements of the correlation matrix) and minimal correlation of optioned assets, respectively. This implies that options can reduce the systemic risk by the effect of correlation hedging.

\begin{figure}[H]
  \centering
  \includegraphics[width=8cm]{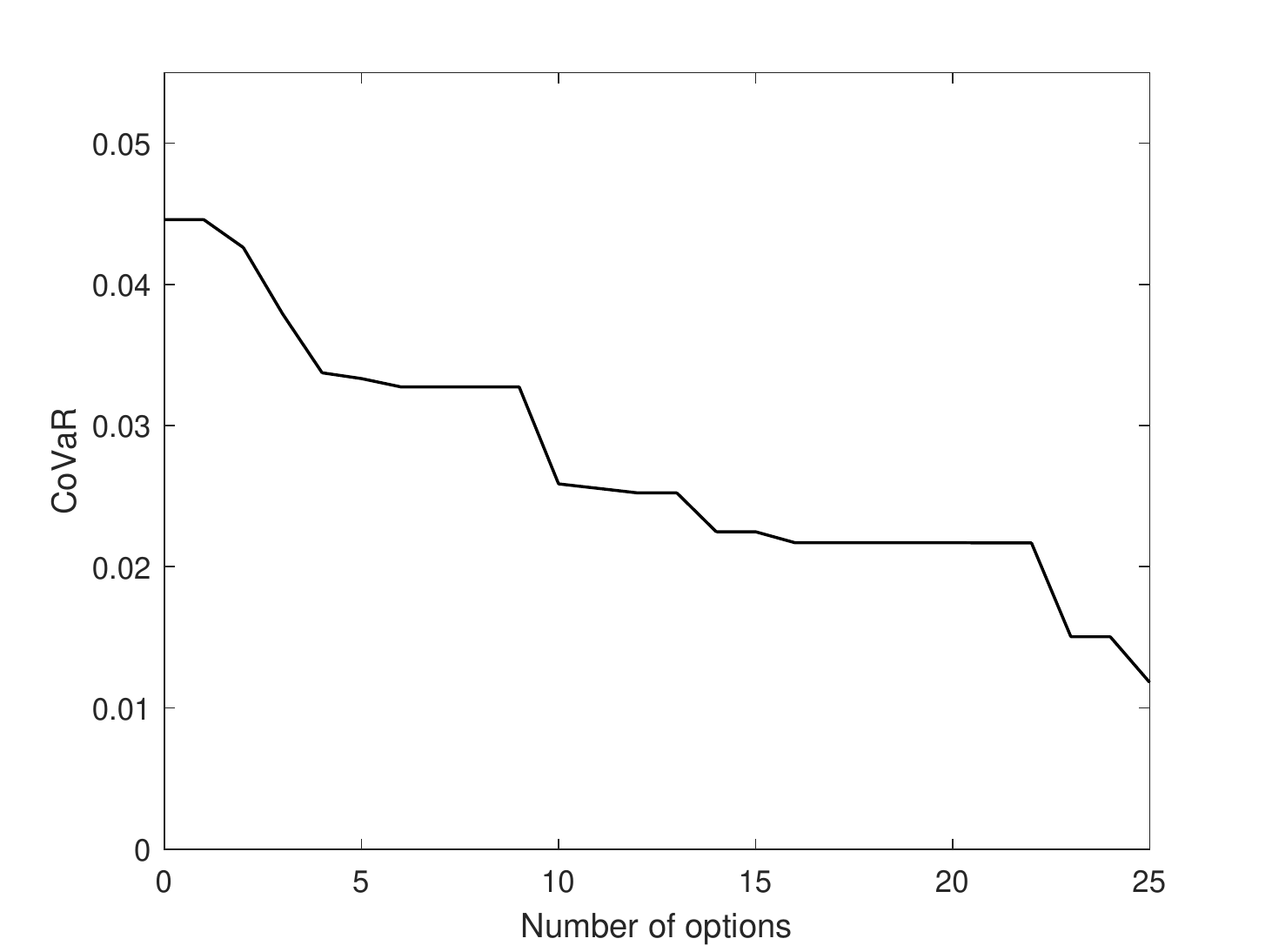}
  \caption{CoVaR of optioned portfolios with different number of options}
  \label{fig5}
\end{figure}

\begin{figure}[H]
  \centering
  \includegraphics[width=8cm]{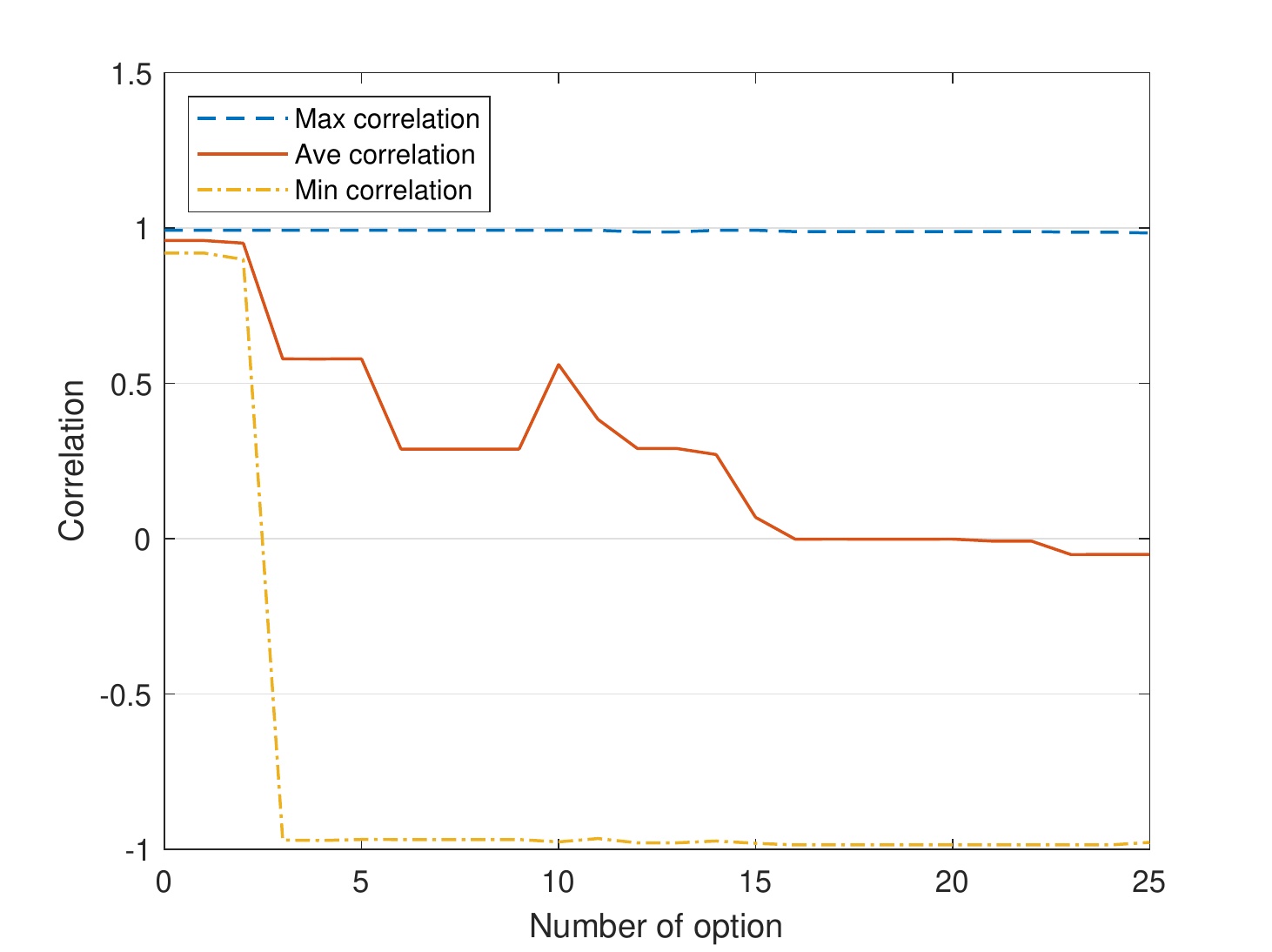}
  \caption{Correlations of optioned assets with different number of options}
  \label{fig6}
\end{figure}

As shown in Proposition \ref{corr}, there exists a correlation hedging effect in optioned portfolios. In practice, we are also interested in how many options are exactly needed to hedge the correlations to a specific level. To investigate the relationship between the number of options in the portfolio and correlations between optioned assets, we use simulation to generate portfolios with different numbers of stocks and options. We cumulatively add options into the portfolio and in each step use the Matlab function ``fsolve" to calculate the weight of optioned assets $(b_i,a_{i1},\cdots,a_{in_i})$, $i=1,\cdots,m,$ with the objective to make all of the correlations between optioned assets to zeros. Figure \ref{fig7} illustrates that using the same number of options as that of stocks has a satisfactory correlation hedging effect. When the number of options is approximately twice as much as that of stocks, the average correlation between optioned assets approaches zero.

\begin{figure}[H]
  \begin{minipage}{0.5\linewidth}
  \centerline{\includegraphics[width=6.5cm]{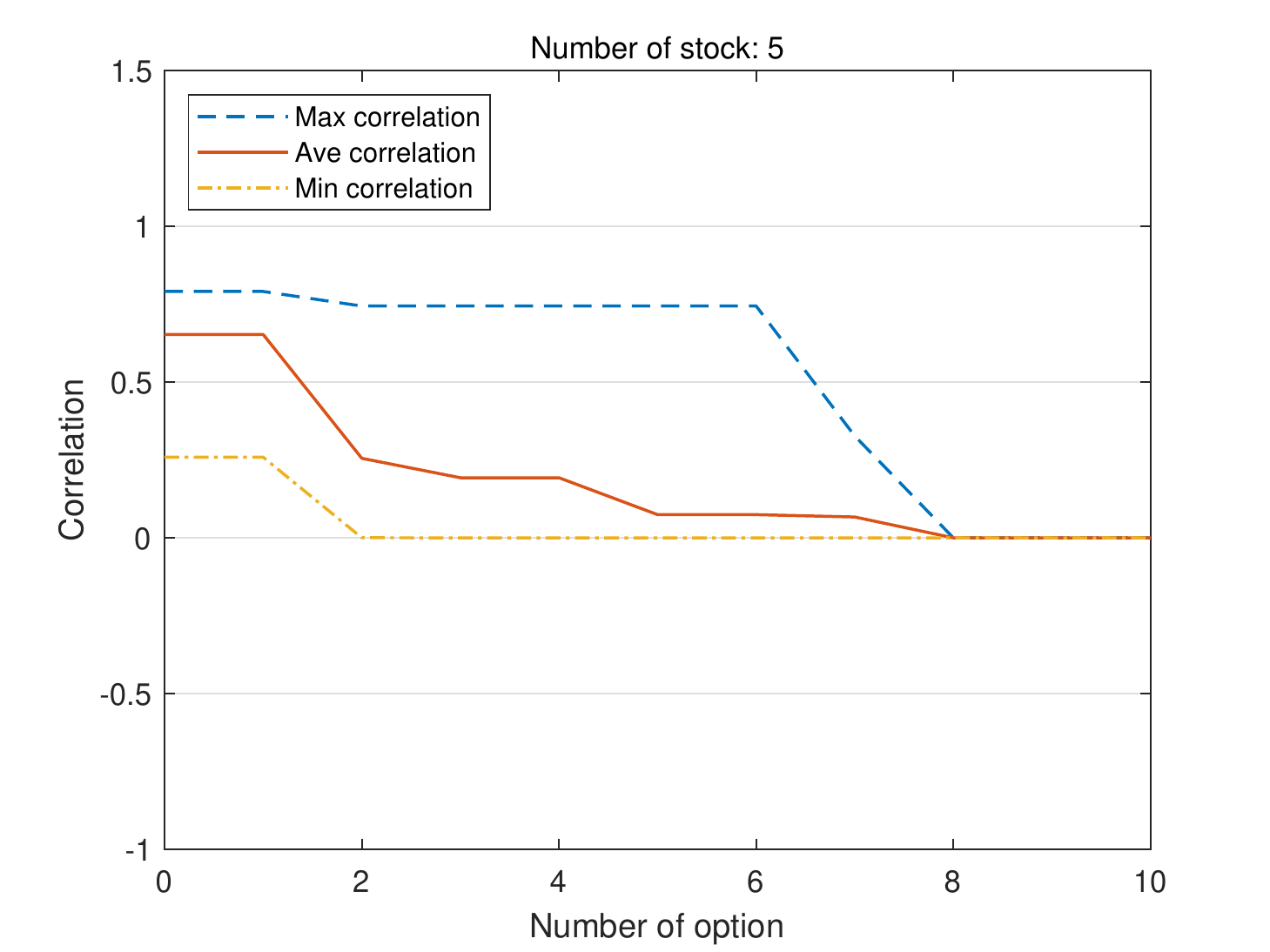}}
  \end{minipage}
  \hfill
  \begin{minipage}{0.5\linewidth}
  \centerline{\includegraphics[width=6.5cm]{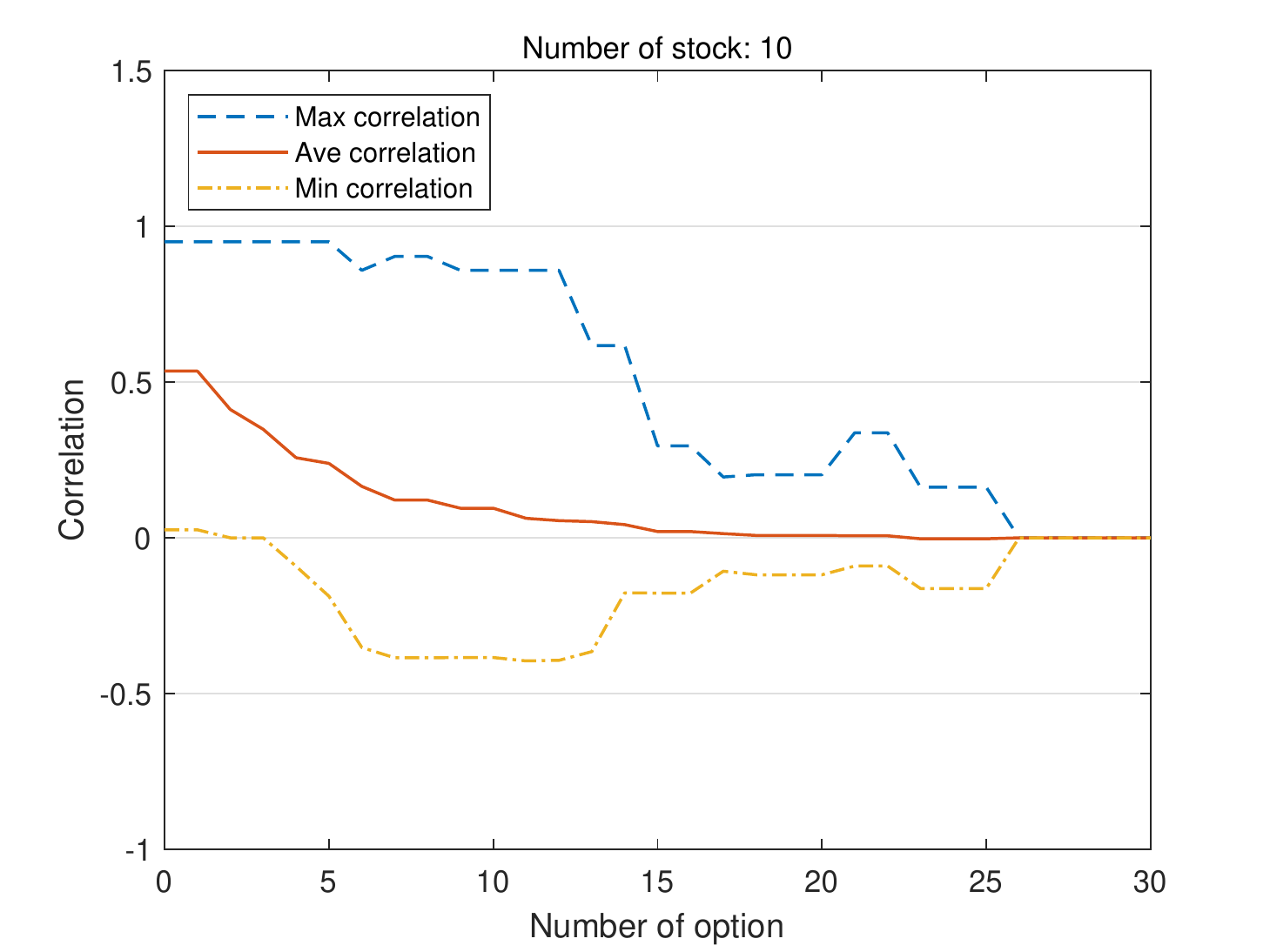}}
  \end{minipage}
  \vfill
  \begin{minipage}{0.5\linewidth}
  \centerline{\includegraphics[width=6.5cm]{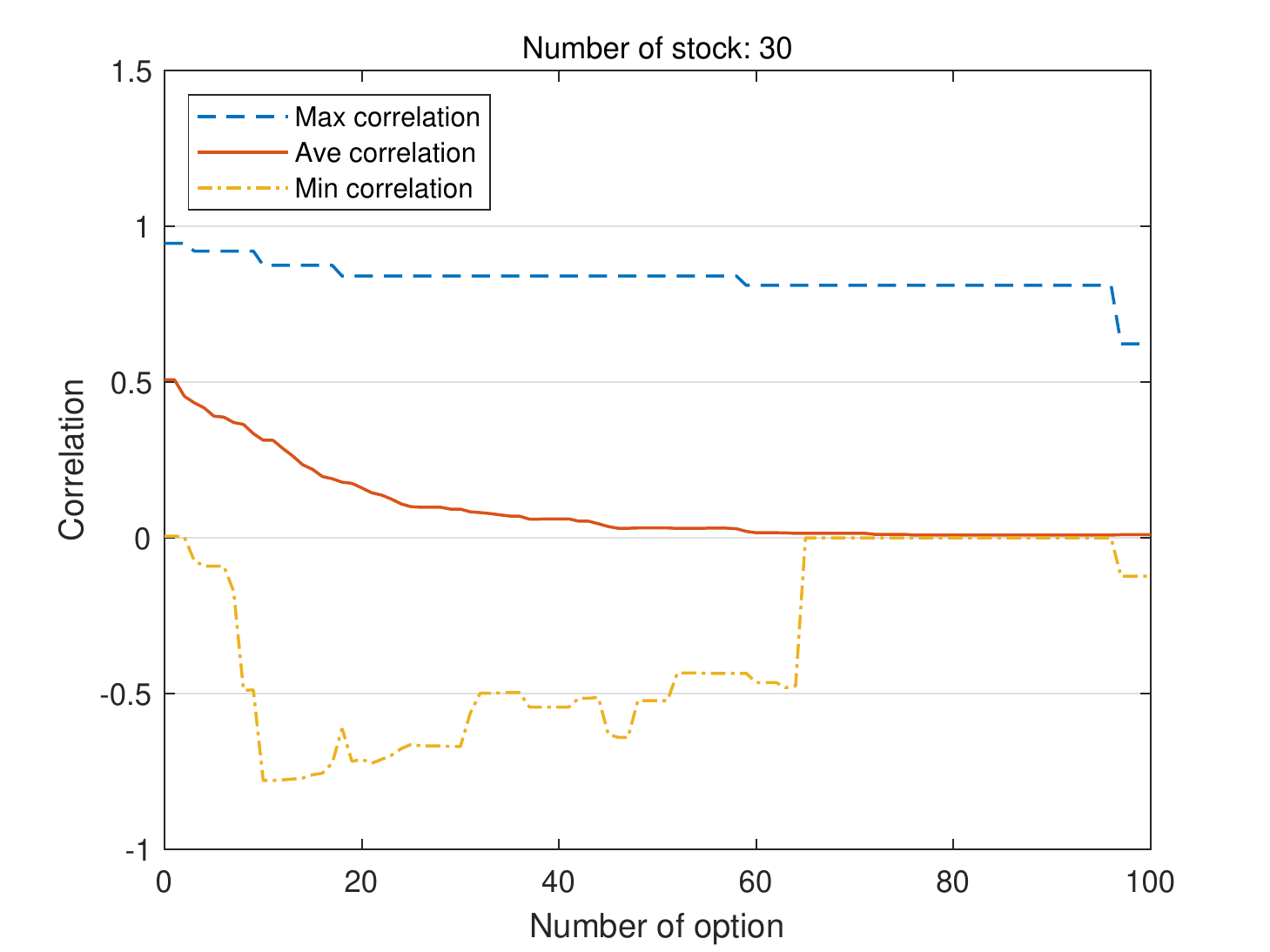}}
  \end{minipage}
  \hfill
  \begin{minipage}{0.5\linewidth}
  \centerline{\includegraphics[width=6.5cm]{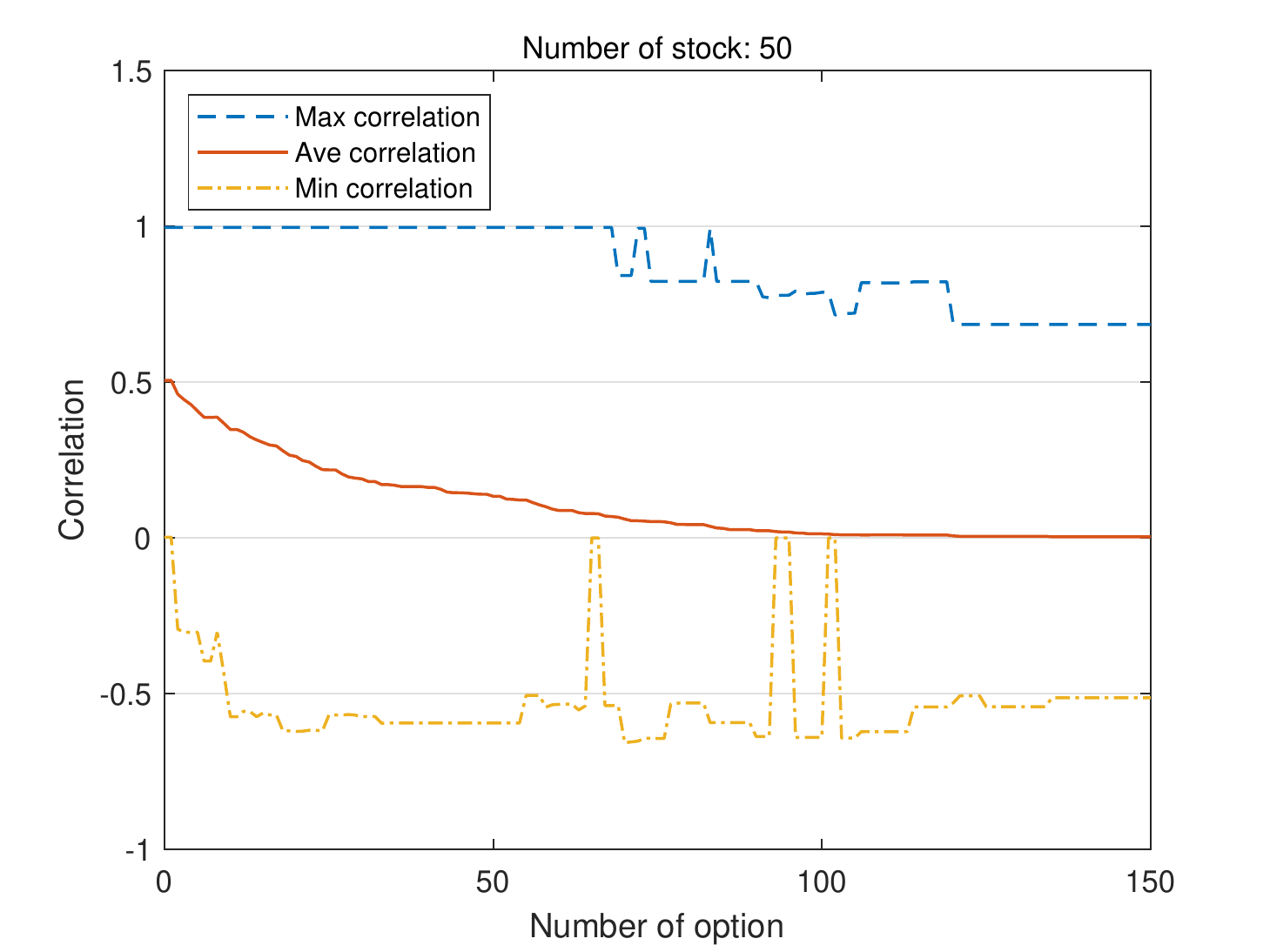}}
  \end{minipage}
  \caption{Correlation hedging effect of different portfolio sizes}
  \label{fig7}
\end{figure}

\subsubsection{Frontiers of stock portfolios and optioned portfolios}
We compare the efficient frontiers of stock portfolios and optioned portfolios to examine the effectiveness of systemic risk control with real data. Cluster analysis is used to identify systemically important assets. Specifically, we regard correlation of returns as the distance between stocks and use cluster analysis to group stocks with the highest correlation into the same cohort. In each cohort, we calculate for each stock the sum of correlations with other stocks, which measures the importance to systemic risk and select the most important stock in each cohort as one of the systemically important asset. Finally, we identify 5 stocks in subset $\mathcal{I}$ as the systemically important assets.

Since the options used in the data set are American options, we use the Least Square Monte Carlo (LSMC) method  proposed by \cite{longstaff2001} and the finite difference to estimate the ``Greeks" of the options. Similar to \cite{longstaff2001}, in the LSMC, the number of simulation paths is set to 20000 and the periods of each path is set to 5.

Throughout the process of drawing the frontier, we set $l_{p_i}=l_{d_i}=0$ and $u_{p_i}=u_{d_i}=0.1$ for the optioned portfolio and $l_{p_i}=0$ and $u_{p_i}=0.1$ for the stock portfolio. We first calculate the minimal $\bar\rho$ by solving a convex optimization problem that minimizes the CoVaR subject to the feasible set $\Omega$, where  and maximal $\bar\rho$ of the stock portfolio by maximizing the expected return subject to the feasible set $\Omega$. Then we select 20 points with equal intervals within the range between the minimal $\bar\rho$ and maximal $\bar\rho$. For each point, we calculate the minimal $\bar\sigma$ by minimizing the variance of the portfolio subject to the systemic risk constraint and feasible set and maximal $\bar\sigma$ by maximizing the expected return subject to the same constraints. Similarly, we select 20 points of $\bar\sigma$ with equal intervals. Then a $20 \times 20$ grid is obtained. For each point in the grid, we solve problem ($P$) for the stock portfolio and the optioned portfolio, respectively, to generate the efficient frontiers.

Figure \ref{fig8} shows the frontiers of the optioned portfolios and the stock portfolios with parameters estimated by different data sets. Overall, the results of the two data sets are similar. We can see that the efficient frontier of optioned portfolios dominates that of stock portfolios, which is consistent with Corollary \ref{frontier}. Furthermore, the the efficient frontier of optioned portfolios is larger than that of  stock portfolios. In addition, the systemic risk of the optioned portfolios can reach zero, while the stock portfolio cannot. For a given parameter $\bar\rho$, the frontier turns to the mean-variance frontier, which is characterized by high returns with high variance. As the parameter $\bar\rho$ decreases, the mean-variance frontier moves down, which implies that the systemic risk control sacrifices the conventional risk-return performance.
\begin{figure}[H]
  \begin{minipage}{0.5\linewidth}
  \centerline{\includegraphics[width=8cm]{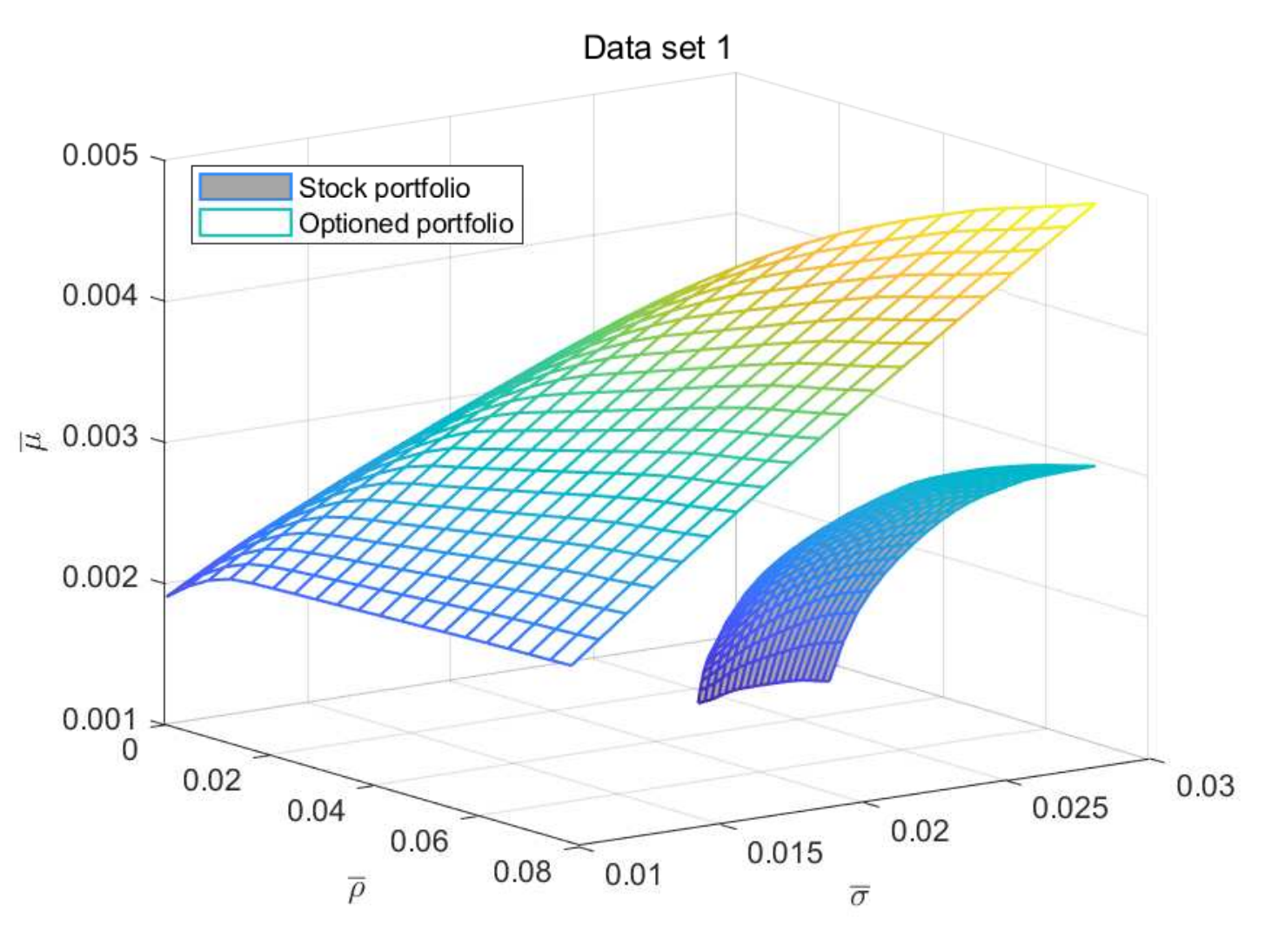}}
  \end{minipage}
  \hfill
  \begin{minipage}{0.5\linewidth}
  \centerline{\includegraphics[width=8cm]{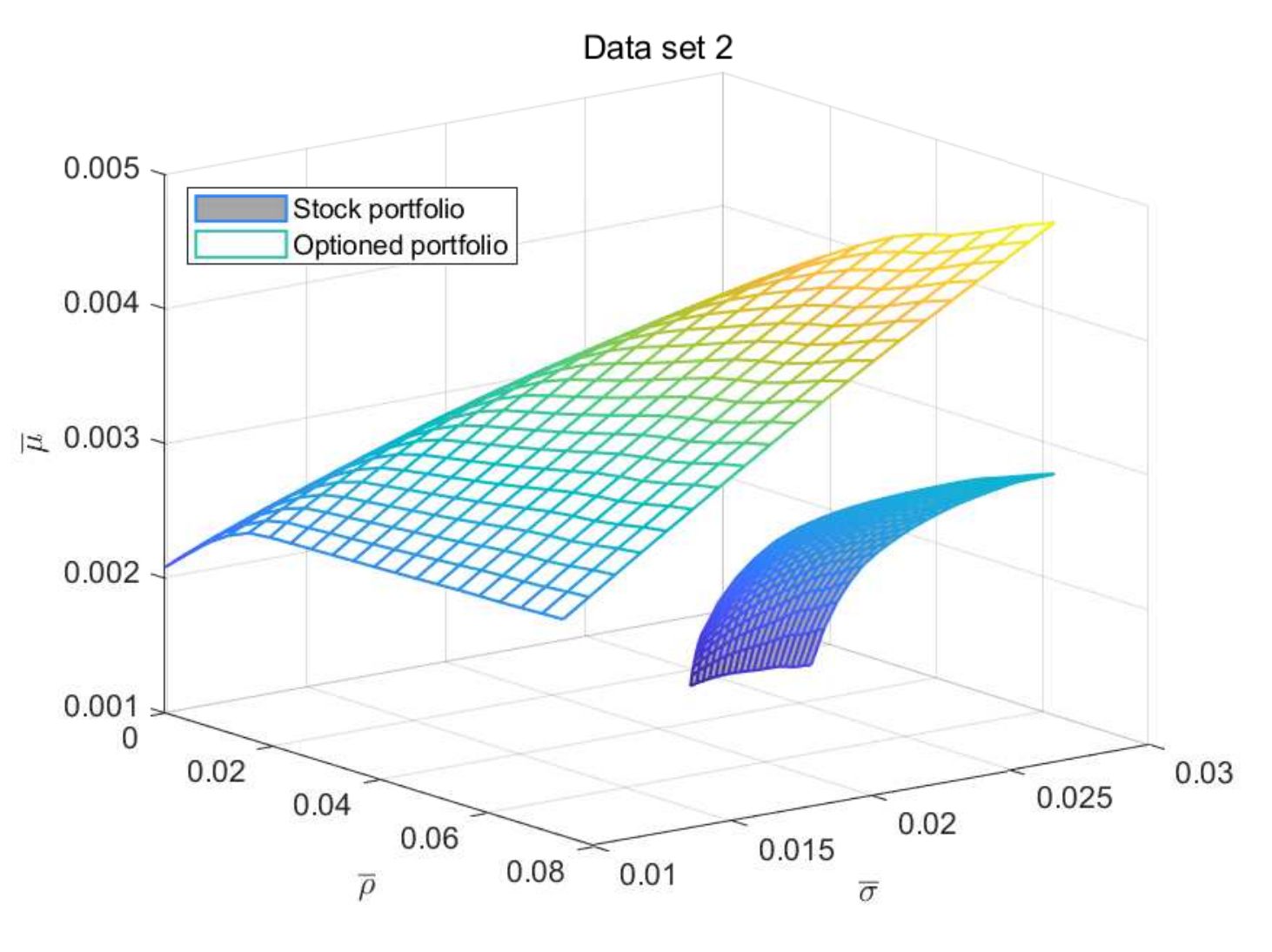}}
  \end{minipage}
  \caption{Frontiers of stock portfolios and optioned portfolios}
  \label{fig8}
\end{figure}

\subsubsection{Portfolio performance under different scenarios}

We simulate different scenarios, including depressions, normalities and booms, to compare the performance of different strategies. The scenarios are generated as follows. Let $\Delta \bm{p}_{\mathcal{I}}=(2\lambda-1)\bm{k}$, where $\lambda$ increases from 0 to 1 with increment of 0.02. The value changes of other stocks are equal to their conditional mean, i.e., $\Delta\bm{ p}_{\mathcal{J}}=\bm{\mu}_{\mathcal{J}}+\Sigma_{\mathcal{J}\mathcal{I}}\Sigma_{\mathcal{I}\mathcal{I}}^{-1}(\Delta\bm{ p}_{\mathcal{I}}-\bm{\mu}_{\mathcal{I}})$. Thus, we generate 51 scenarios and as parameter $\lambda$ increases, these scenarios change from depressions to booms. Here, the parameters of stocks and options that are used for generating scenarios are estimated with the two data sets. Furthermore, we set $\bar\sigma=0.1$ and $u_{d_i}=u_{p_i}=0.1$. When short selling is allowed, set $l_{d_i}=l_{p_i}=-0.1$. Otherwise set $l_{d_i}=l_{p_i}=0$. To ensure the systemic risk constraint is active, we set $\bar\rho$ to the minimum that the {\it stock control} portfolio is feasible. For the {\it optioned control}  portfolio, we set $\bar\rho$ to 0.1 times of the CoVaR of the {\it stock} strategy. We consider stocks that are heavily weighted and highly correlated with others as systemically important stocks. Specifically, we use cluster analysis to identify stocks with high correlation. In addition,  we execute {\it stock}  strategy and then select stocks with the highest weight as the objective of  systemic risk control. We start with an initial wealth of one and calculate the portfolio value under each scenario.

As we can see in Figure \ref{fig9}, in the depressive scenarios, the {\it optioned control} strategy outperforms the {\it optioned strategy} and the {\it stock control} strategy slightly outperforms the {\it stock strategy} but performs poorer than the {\it optioned control} strategy. This result reflects the uncontrollability of systemic risk of the stock portfolio. The {\it optioned strategy} performs worst in the depressions, since the leverage effect of options reinforces the risk without management. The systemic risk control strategies ({\it stock control} and {\it optioned control}) have a worse performance in the booming scenarios. Therefore, the systemic risk control strategies prevent portfolios from large losses in depressions, but loss opportunities in booms.
\begin{figure}[H]
  \begin{minipage}{0.5\linewidth}
  \centerline{\includegraphics[width=8cm]{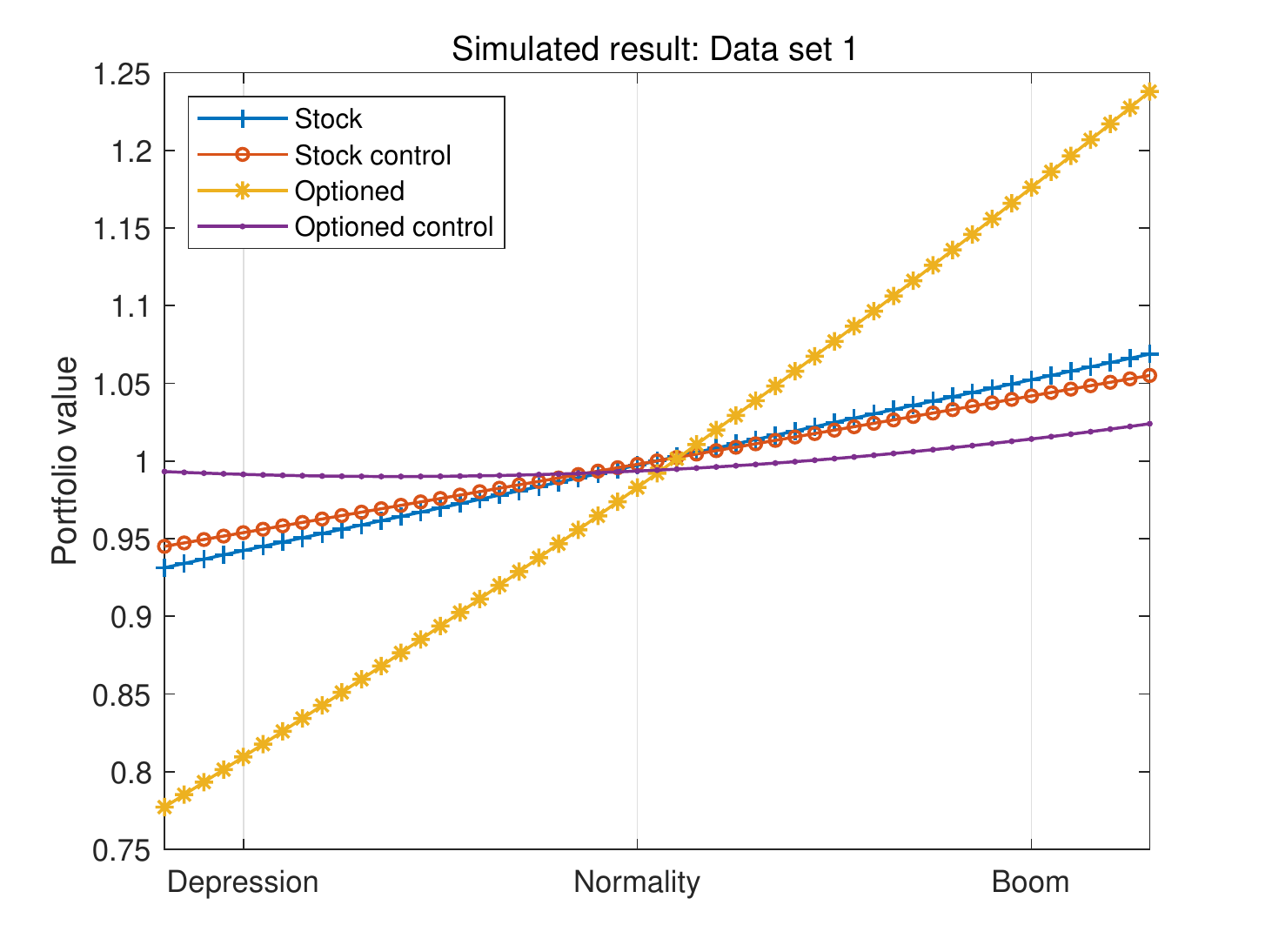}}
  \end{minipage}
  \hfill
  \begin{minipage}{0.5\linewidth}
  \centerline{\includegraphics[width=8cm]{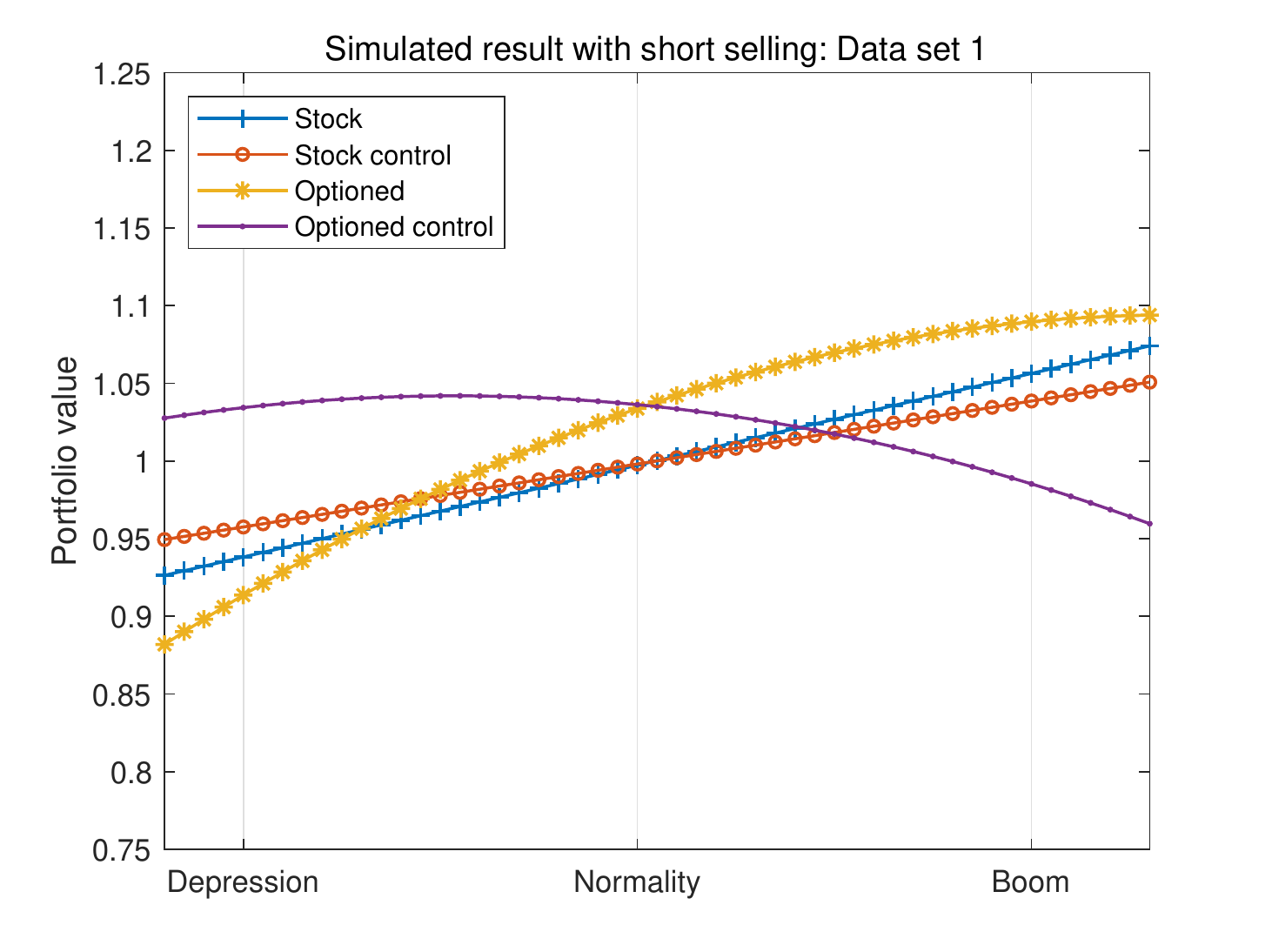}}
  \end{minipage}
  \vfill
  \begin{minipage}{0.5\linewidth}
  \centerline{\includegraphics[width=8cm]{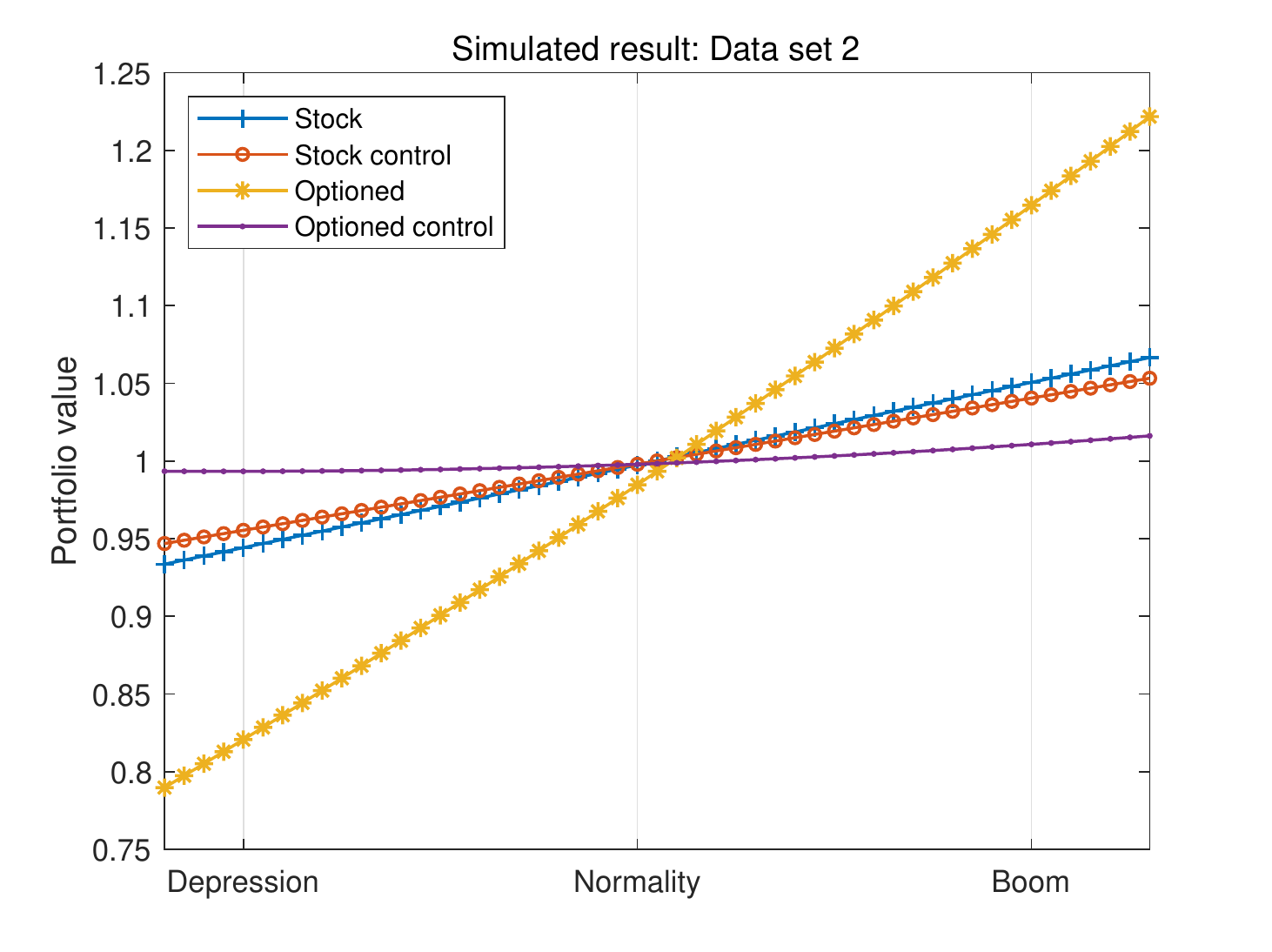}}
  \end{minipage}
  \hfill
  \begin{minipage}{0.5\linewidth}
  \centerline{\includegraphics[width=8cm]{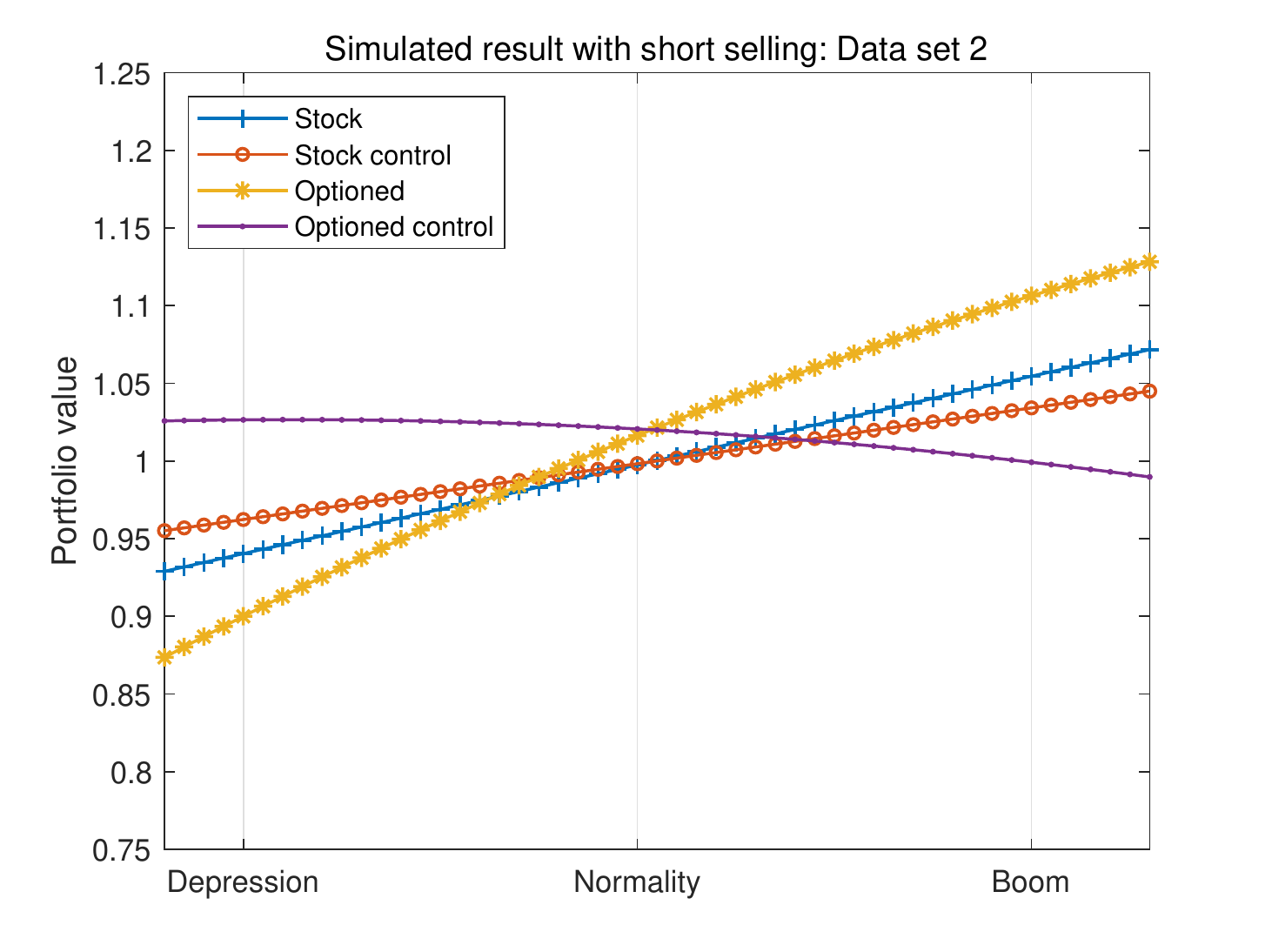}}
  \end{minipage}
  \caption{Portfolio performance under different scenarios}
  \label{fig9}
\end{figure}

\subsection{Out-of-sample test}
In this subsection, we compare the out-of-sample performance of different strategies. To be more practical, we model the stock returns as a multivariate normal distribution with heteroscedasticity, where the variances of stocks follow the GJR-GARCH model (\cite{glosten1993}) and the correlation matrix follows the DCC model (\cite{Engle2002}). Accordingly, the option pricing under GARCH process is referred to \cite{Duan1995}. Some details of the model specification are delegated to Appendix F. The criteria for determining the systemically important assets are the same as that of the scenario analysis of in-sample test (Section 4.1.3).

\subsubsection{Back testing during financial crisis via simulation}
As one of the most typical systemic risk event in recent years, the financial crisis in 2008 caused the stock market to suffer heavy losses. To verify our model in that period, we compare the performance of {\it stock}, {\it optioned} and two {\it optioned control} strategies during January 2008 to December 2009. As the option data in that period is unavailable for us, we simulate four types of European options, ATM call option, ATM put option, 20\% OTM call option and 20\% OTM put option. We use the Monte Carlo method to estimate the prices and ``Greeks". Specifically, we simulate 20000 different paths for each stock. In each path, we simulate stock returns with the GJR-GARCH model and discount the gain of options at the expiration date. The average of discounted gains of all paths is the estimated option prices. The ``Greeks'' are calculated by the finite difference method. When approaching the due date, the effect of estimation error on ``Greeks'' is large. Therefore, we set the expiration of options at 1.5 years, but all of the options are tradable just for one year and then replaced by new options.

We start with an initial portfolio value $v_0=1$ for each strategy and rebalance the strategies once a week. At the beginning of every week, the parameters of the optimization model are estimated with historical data from January 2000 and the systematically important stocks are identified. At the beginning of period $t$, we set $l_{d_i}=l_{p_i}=0$, $u_{d_i}=u_{p_i}=0.1v_t$ and $\bar\sigma=0.1$ for all strategies, where $v_t$ is the portfolio value at the beginning of period $t$. We set $\bar\rho_1=0.05$ and $\bar\rho_2=0.001$. The corresponding portfolios are defined as {\it high risk control} and {\it low risk control}. Assume that $(\bm{x}_t,\bm{y}_t)$ is the optimal portfolio. Then the cash $k_t$ is calculated by $k_t=v_t-\bm{d}_t^\top\bm{x}_t-\bm{p}_t^\top\bm{y}_t$, where $\bm{d}_{t}$ and $\bm{p}_{t}$ are the prices of options and stocks at the beginning of period $t$. If $k_t>0$, we invest the cash on the risk-free asset and earn a risk-free interest rate. The portfolio value at the beginning of the next period is
\begin{eqnarray*}
&&v_{t+1}=\bm{d}_{t+1}^\top\bm{x}_t+\bm{p}_{t+1}^\top\bm{y}_t+k_t(1+r\Delta t),
\end{eqnarray*}
where $r$ is the risk-free rate.
This process is repeated until the end of the out-of-sample period and the portfolio value of each period is collected.

Figure \ref{fig10} shows the portfolio values of different strategies from January 2008 to December 2009. As we can see, during the financial crisis, the optioned portfolios outperforms other strategies. When the market experiences a dramatic decline, the prices of the deep out-of-the-money put options increased sharply. Since the GARCH model predicted the market downturn, the optioned portfolios invest in those put options and derive a huge profit. Particularly, the systemic risk control strategies perform better than other strategies since they require a lower systemic risk.
\begin{figure}[H]
  \centering
  \includegraphics[width=12cm]{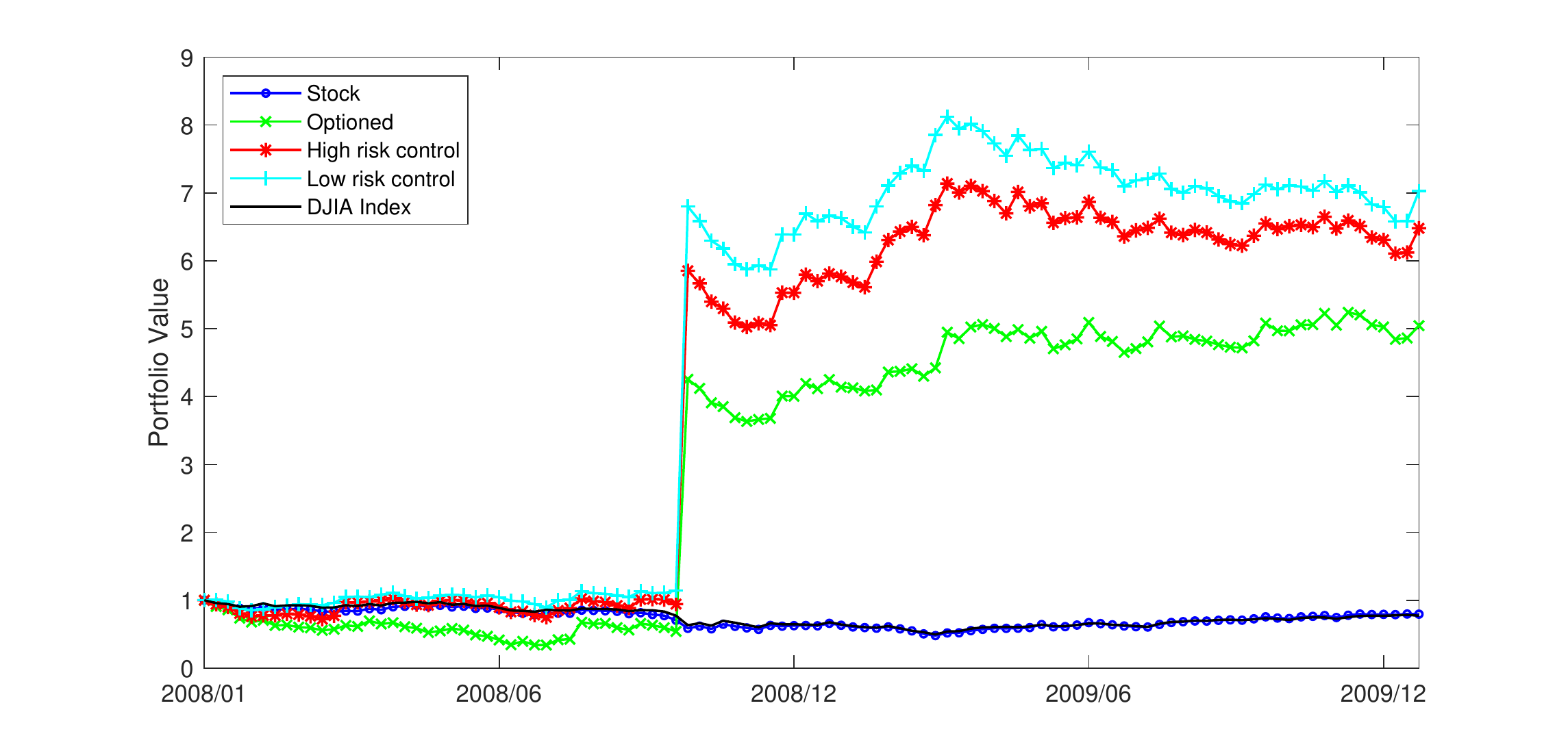}
  \caption{Out-of-sample performance during 2008-2009}
  \label{fig10}
\end{figure}
Furthermore, we calculate the weekly return of all strategies and summarize their performance: mean (Mean), standard deviation (Std), minimal return (Min), average drawdown (ADD, \cite{alexander2006}), upside potential ratio (UP ratio, \cite{sortino1991}) and downside Sharpe ratio (DS ratio, \cite{ziemba2005}). According to Table 1, the {\it low risk strategy} performs better than any other strategy, with four of six performance measures being the best. It has the lowest minimal return, lowest average drawdown, highest upside potential ratio and highest downside sharp ratio. In addition, it has a satisfactory mean of returns.

\begin{table}[H]
  \centering
  \caption{Performance of different strategies during 2008-2009}
  \begin{tabular}{lcccccc}
\hline
Strategy &  Mean &     Std &         Min &        ADD &         UP ratio &          DS ratio\\
\hline
Stock & -0.0014  & \textbf{0.0409} & -0.1681 & 0.4065 & 0.4660 & -0.0350 \\
Optioned & \textbf{0.0655} & 0.6819 & -0.1683 & 0.3346 &	1.9821 & 1.0919 \\
High risk control & 0.0521 & 0.5152 & -0.1065 &	0.0999 & 2.4582 & 1.4523 \\
Low risk control & 0.0503 &	0.4888 & \textbf{-0.1055} &	\textbf{0.0849} & \textbf{2.9732} &	\textbf{1.8397} \\
DJIA Index & -0.0017 & 0.0389 & -0.1815 & 0.3675 & 0.4341 & -0.0441 \\
\hline
  \end{tabular}
\end{table}

\subsubsection{Back testing with real data}
We use real data to examine the out-of-sample performance of the four strategies adopted in section 4.2.1 with the same parameter settings. Here, the ask-bid spread of options is considered as the transaction cost since it is non-negligible for the trade of options. In this case, the cash $k_t$ is calculated by $k_t=v_t-(\bm{d}_t^{ask})^\top\bm{x}_t^+-(\bm{d}_t^{bid})^\top\bm{x}_t^--\bm{p}_t^\top\bm{y}_t$, where $\bm{d}_{t}^{ask}$ and $\bm{d}_{t}^{bid}$ are the ask prices and bid prices of options at the beginning of period $t$. The portfolio value at the beginning of period $t+1$ is
\begin{eqnarray*}
&&v_{t+1}=(\bm{d}_{t+1}^{ask})^\top\bm{x}_t^++(\bm{d}_{t+1}^{bid})^\top\bm{x}_t^-+\bm{p}_{t+1}^\top\bm{y}_t+k_t(1+r\Delta t).
\end{eqnarray*}
We still start with an initial wealth of one for each strategy and rebalance the strategies once a week.
\begin{figure}[H]
  \begin{minipage}{0.5\linewidth}
  \centerline{\includegraphics[width=8cm]{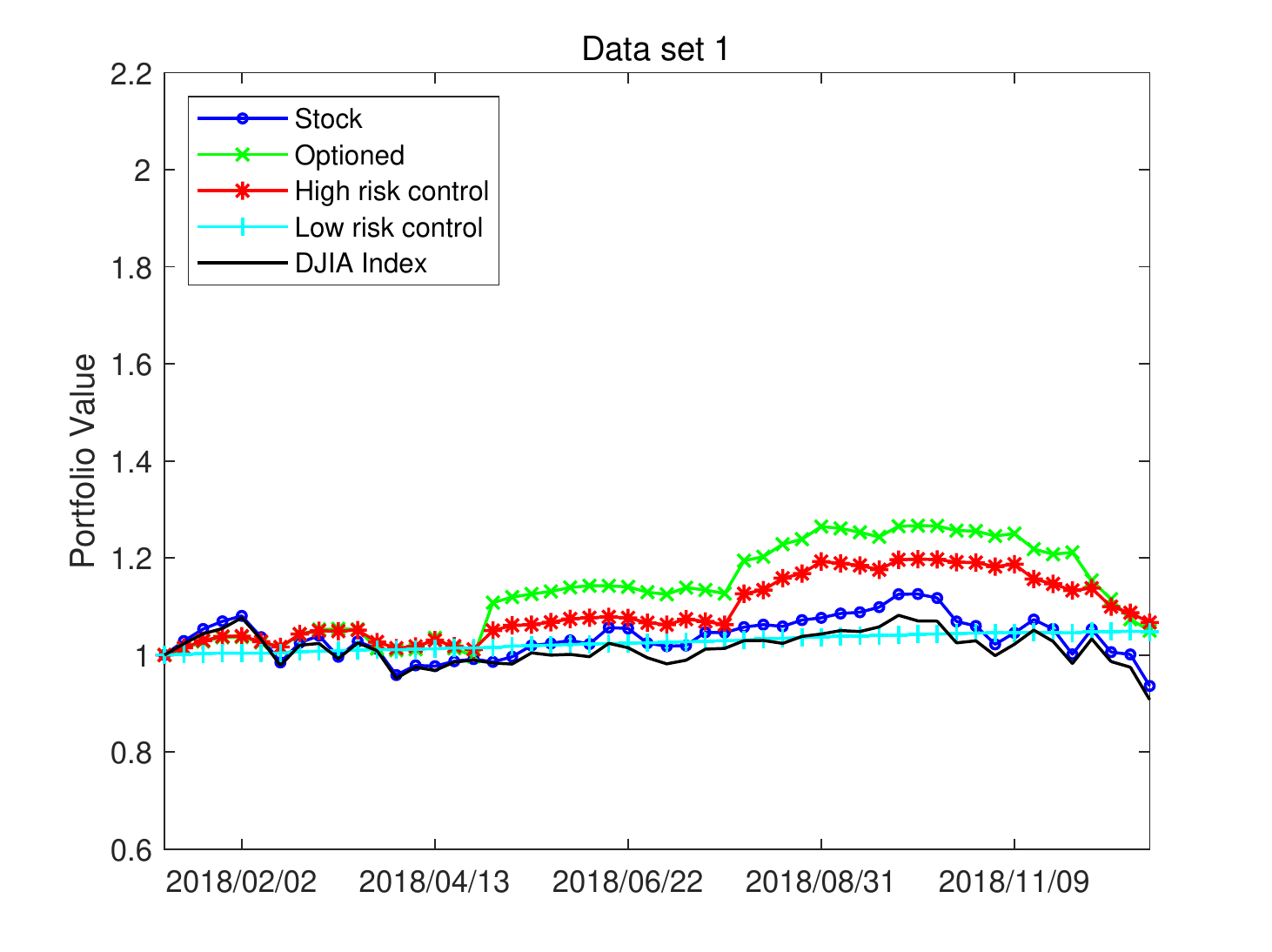}}
  \end{minipage}
  \hfill
  \begin{minipage}{0.5\linewidth}
  \centerline{\includegraphics[width=8cm]{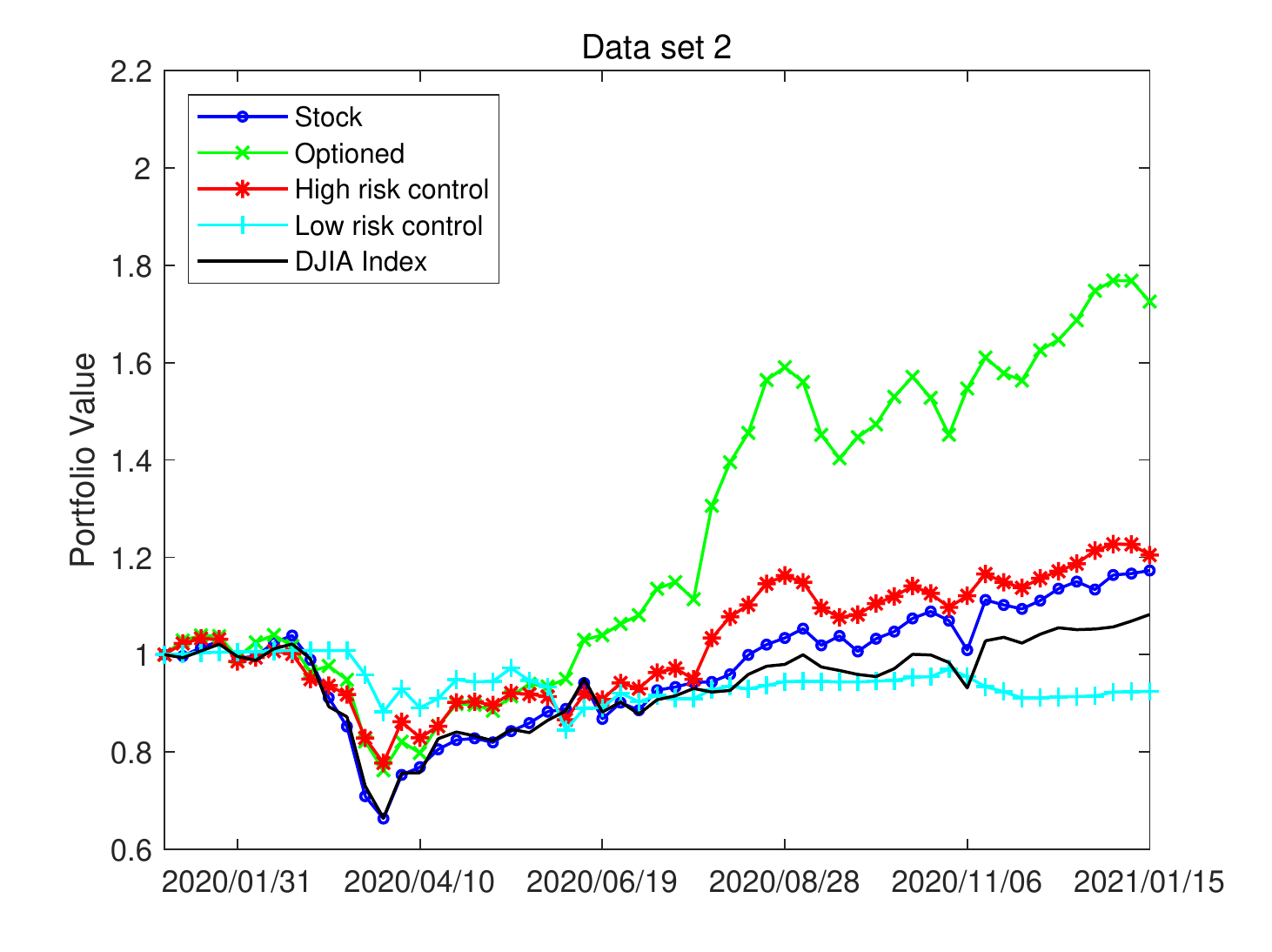}}
  \end{minipage}
  \caption{Out-of-sample performance of different strategies}
  \label{fig11}
\end{figure}

\begin{table}[H]
  \centering
  \caption{Performance measures of different strategies}
\begin{tabular}{clcccccc}
\hline
~&Strategy &  Mean &     Std &         Min &        ADD &         UP ratio &          DS ratio\\
\hline
\multirow{5}*{Data set 1}&Stock & -0.0009 & 0.0264 & -0.0648 & 0.0534 & 0.4198 & -0.0296 \\
~&Optioned & 0.0012 & 0.0237 & -0.0478 &	0.0222 & 0.5395 & 0.0669 \\
~&High risk control & \textbf{0.0014} & 0.0154 & -0.0335 & 0.0161 & \textbf{0.6036} &	0.1115 \\
~&Low risk control&  0.0009 & \textbf{0.0008} & \textbf{-0.0010} & \textbf{0.0001} &	0.4568 & \textbf{1.5556} \\
~&DJIA Index & -0.0016 &	0.0260 & -0.0687 & 0.0608 &	0.4084 & -0.0517 \\
\hline
\multirow{5}*{Data set 2}&Stock & 0.0040 & 0.0450 & -0.1682 & 0.1015 & 0.5073 & 0.0677 \\
~&Optioned & \textbf{0.0112} & 0.0473 & -0.1339 &	\textbf{0.0627} & \textbf{0.8292} & \textbf{0.2531} \\
~&High risk control & 0.0041 & 0.0356 &	-0.0973 & 0.0716 & 0.6397 &	0.1224 \\
~&Low risk control& -0.0011 & \textbf{0.0250} & \textbf{-0.0949} & 0.0705 &	0.2861 & -0.0327 \\
~&DJIA Index & 0.0025 & 0.0457 &	-0.1630 & 0.0999 &	0.4609 & 0.0504 \\
\hline
\end{tabular}
\end{table}
Figure \ref{fig11} shows the changes in the portfolio value of different strategies. During 2018, which is close to the normal market situation, the portfolio value of {\it low risk control} strategy changes little and the {\it optioned} strategy performs the best. During 2020, the market experiences a dramatic decline during the outbreak of COVID-19 and then a slow recovery. The {\it low risk control} strategy performs well during the downturn but poorly during the recovery. This result is consistent with our scenario analysis of the in-sample test. The {\it high risk control} strategy has moderate performance in both periods. Performance measures in Table 2 show that in general the {\it low risk control} strategy outperforms others during 2018 and performs more stable during 2020. The performance of the {\it high risk control} strategy is between the {\it low risk control} and {\it optioned} strategy. Figures \ref{fig12} and \ref{fig13} show the return distribution of different strategies during the out-of-sample period. The thinner tail of return distribution indicates that the {\it low risk control} strategy is more stable than other strategies.

\begin{figure}
  \begin{minipage}{0.5\linewidth}
  \centerline{\includegraphics[width=7.5cm]{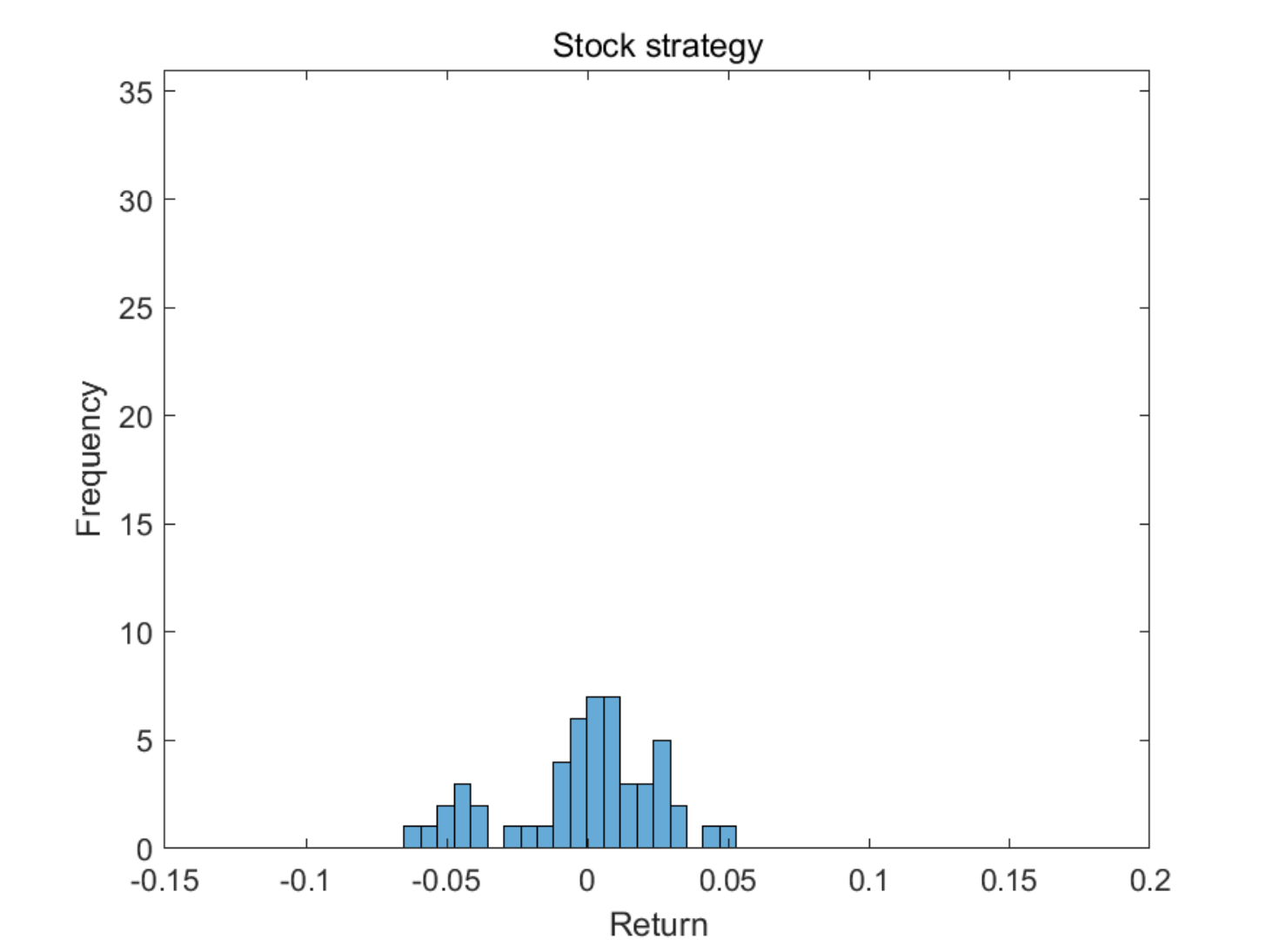}}
  \end{minipage}
  \hfill
  \begin{minipage}{0.5\linewidth}
  \centerline{\includegraphics[width=7.5cm]{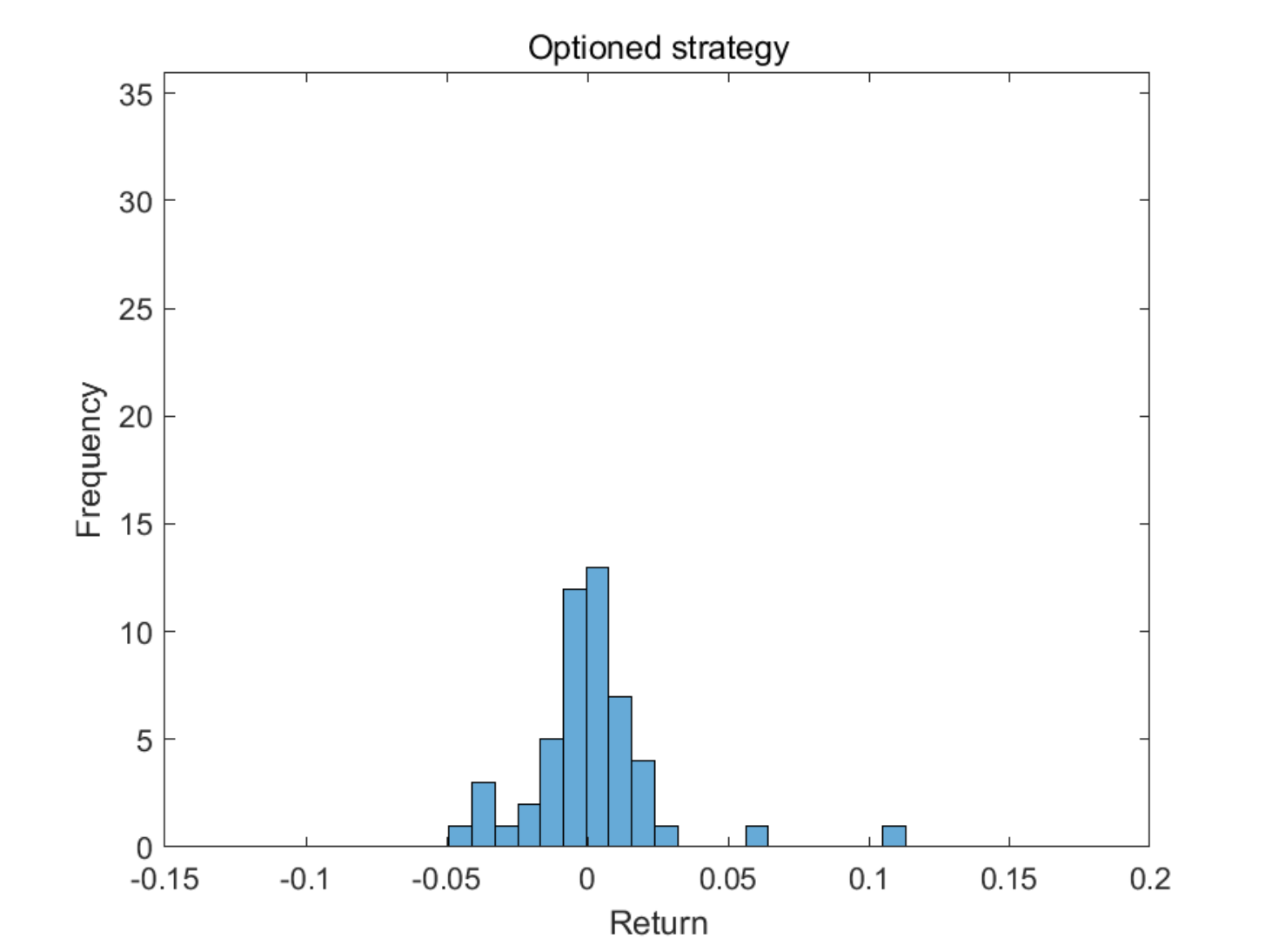}}
  \end{minipage}
  \vfill
  \begin{minipage}{0.5\linewidth}
  \centerline{\includegraphics[width=7.5cm]{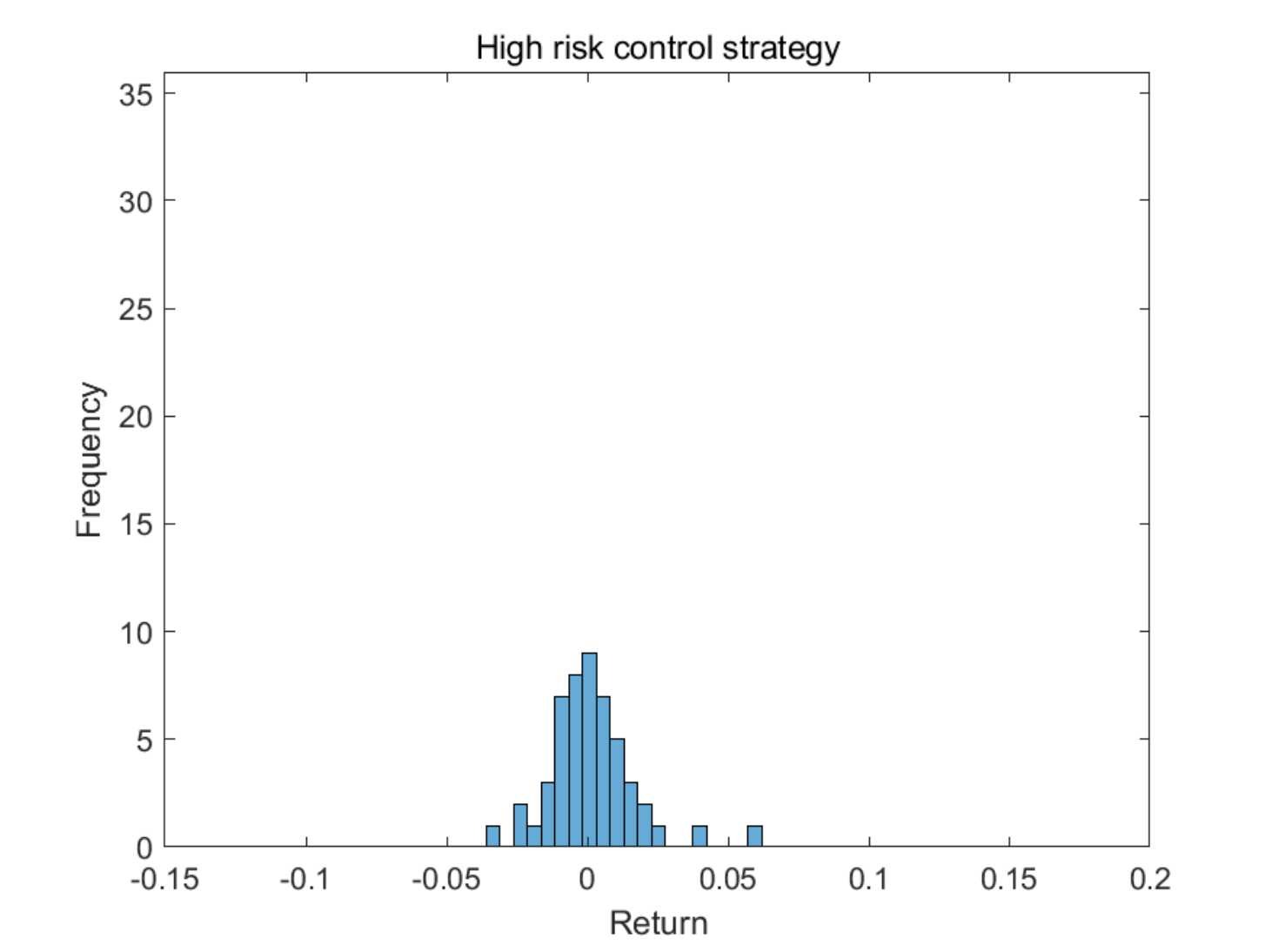}}
  \end{minipage}
  \hfill
  \begin{minipage}{0.5\linewidth}
  \centerline{\includegraphics[width=7.5cm]{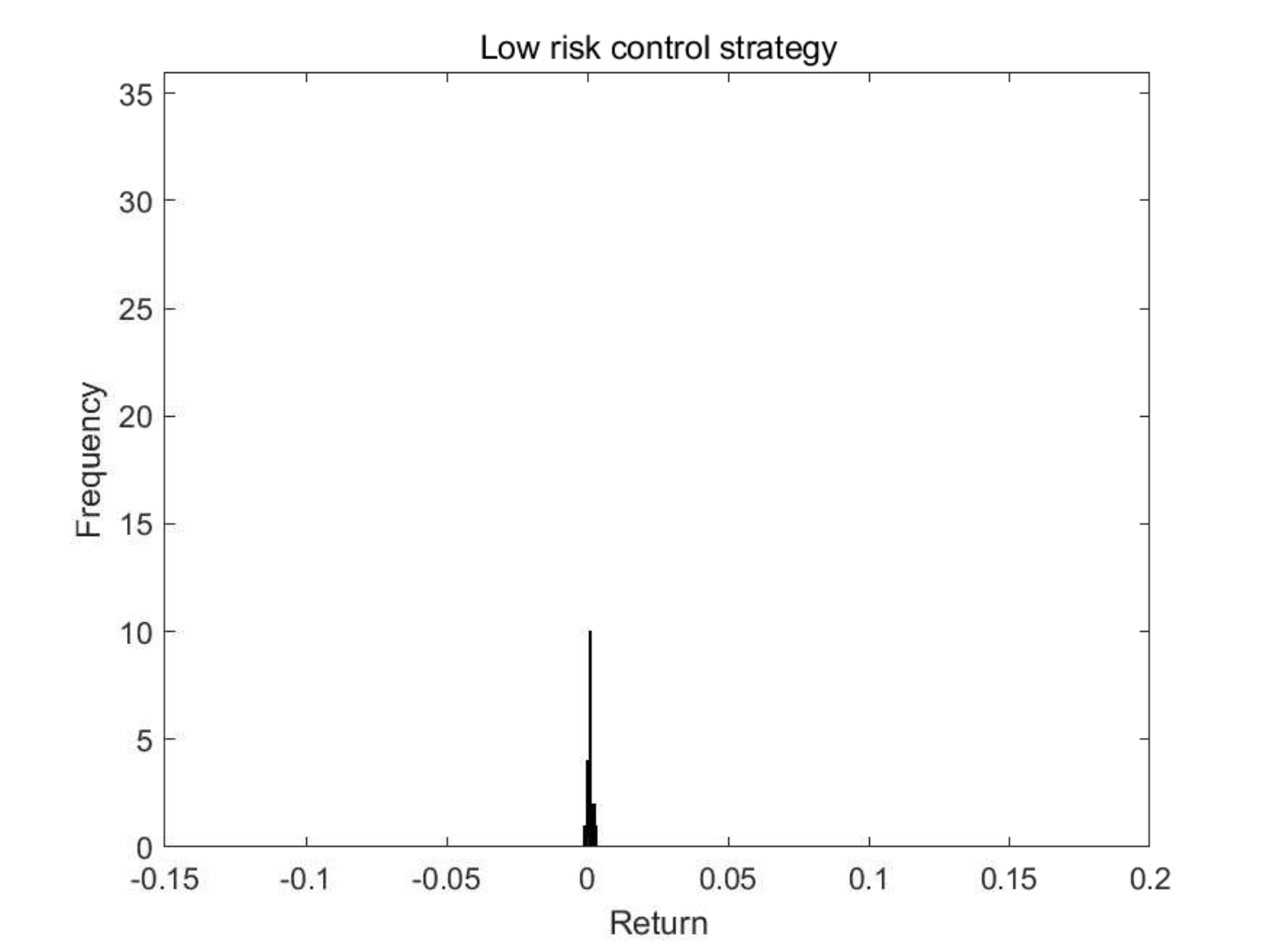}}
  \end{minipage}
  \caption{Out-of-sample return distribution: Data set 1}
  \label{fig12}
\end{figure}

\begin{figure}
  \begin{minipage}{0.5\linewidth}
  \centerline{\includegraphics[width=7.5cm]{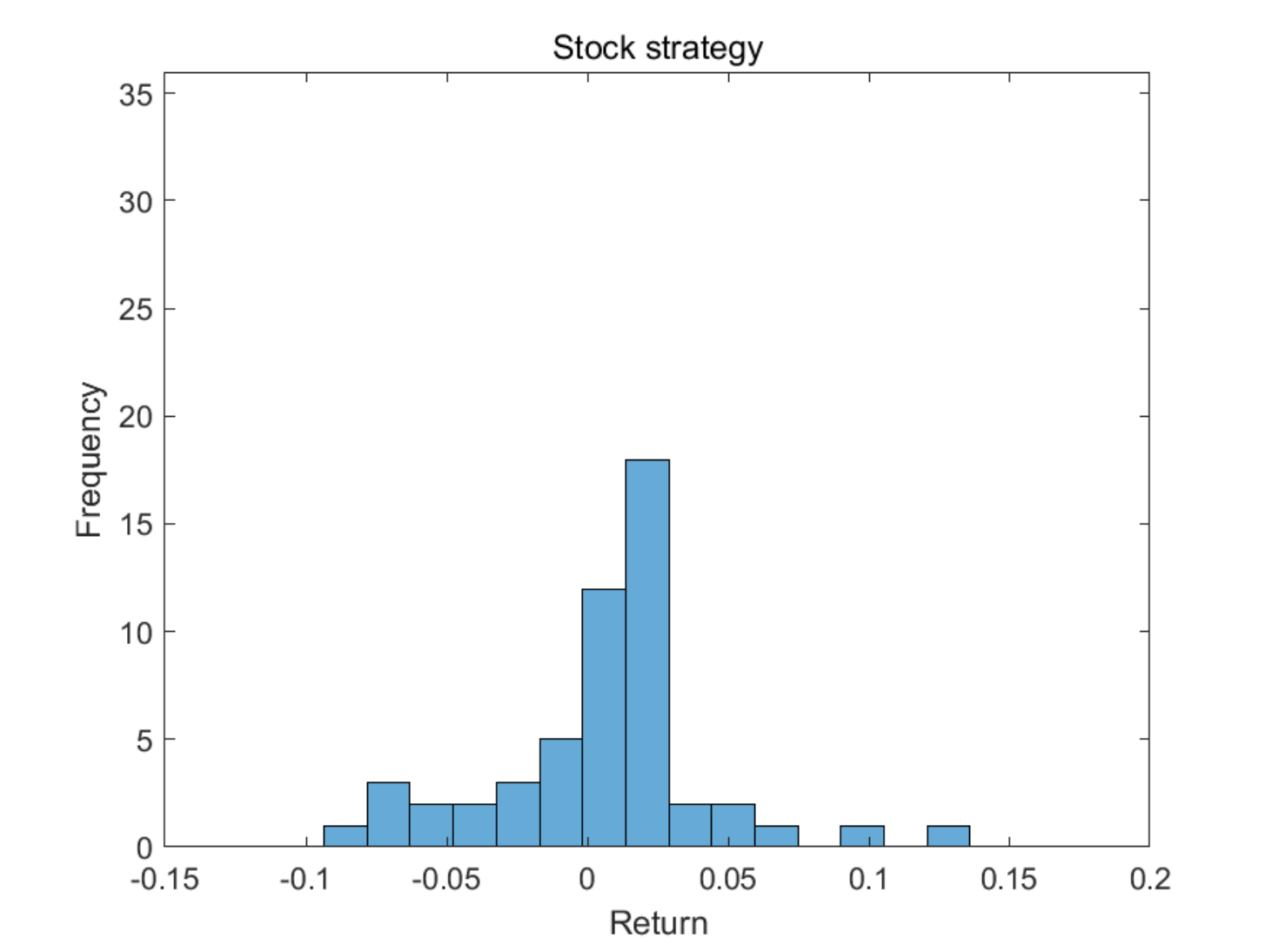}}
  \end{minipage}
  \hfill
  \begin{minipage}{0.5\linewidth}
  \centerline{\includegraphics[width=7.5cm]{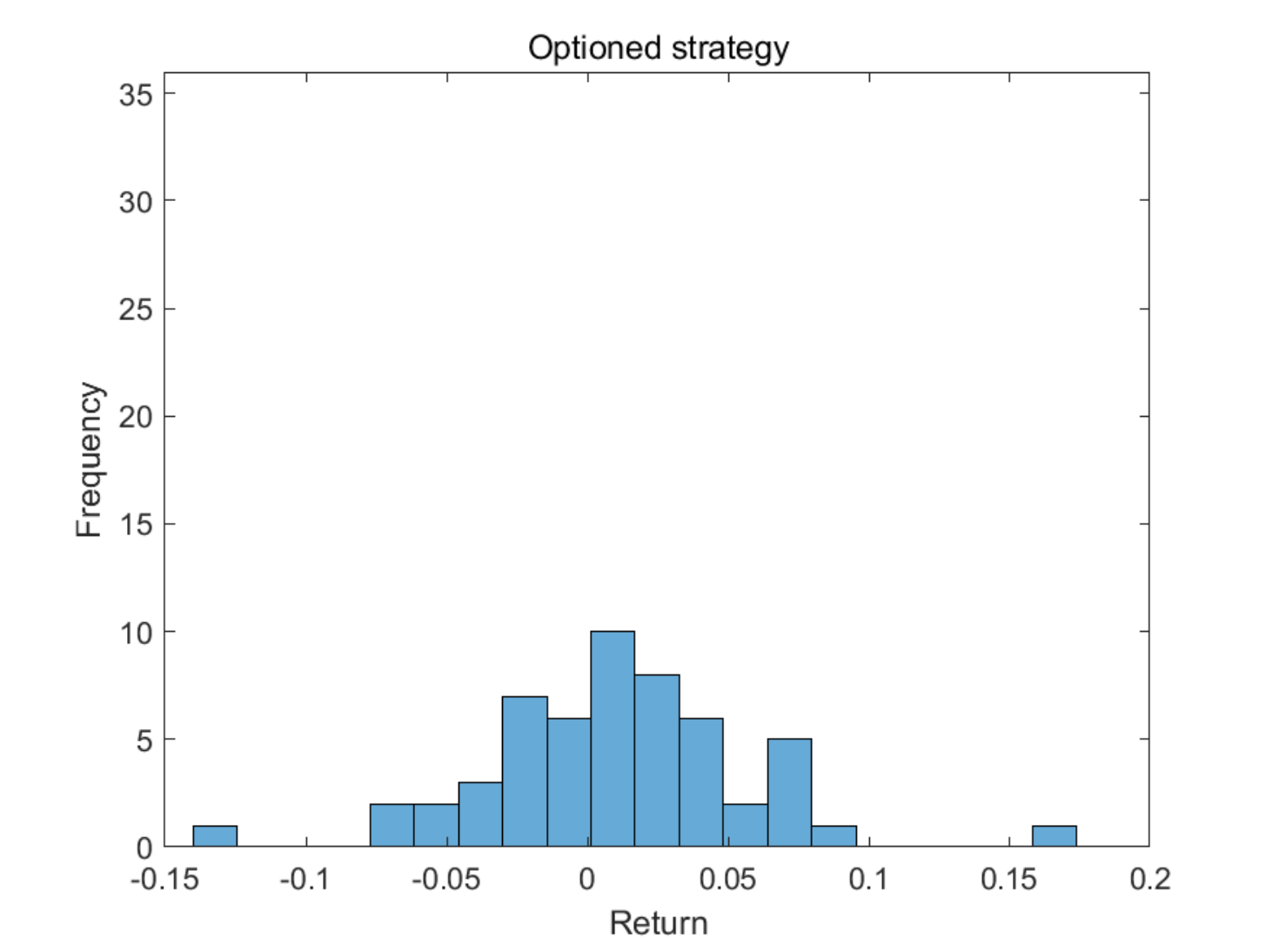}}
  \end{minipage}
  \vfill
  \begin{minipage}{0.5\linewidth}
  \centerline{\includegraphics[width=7.5cm]{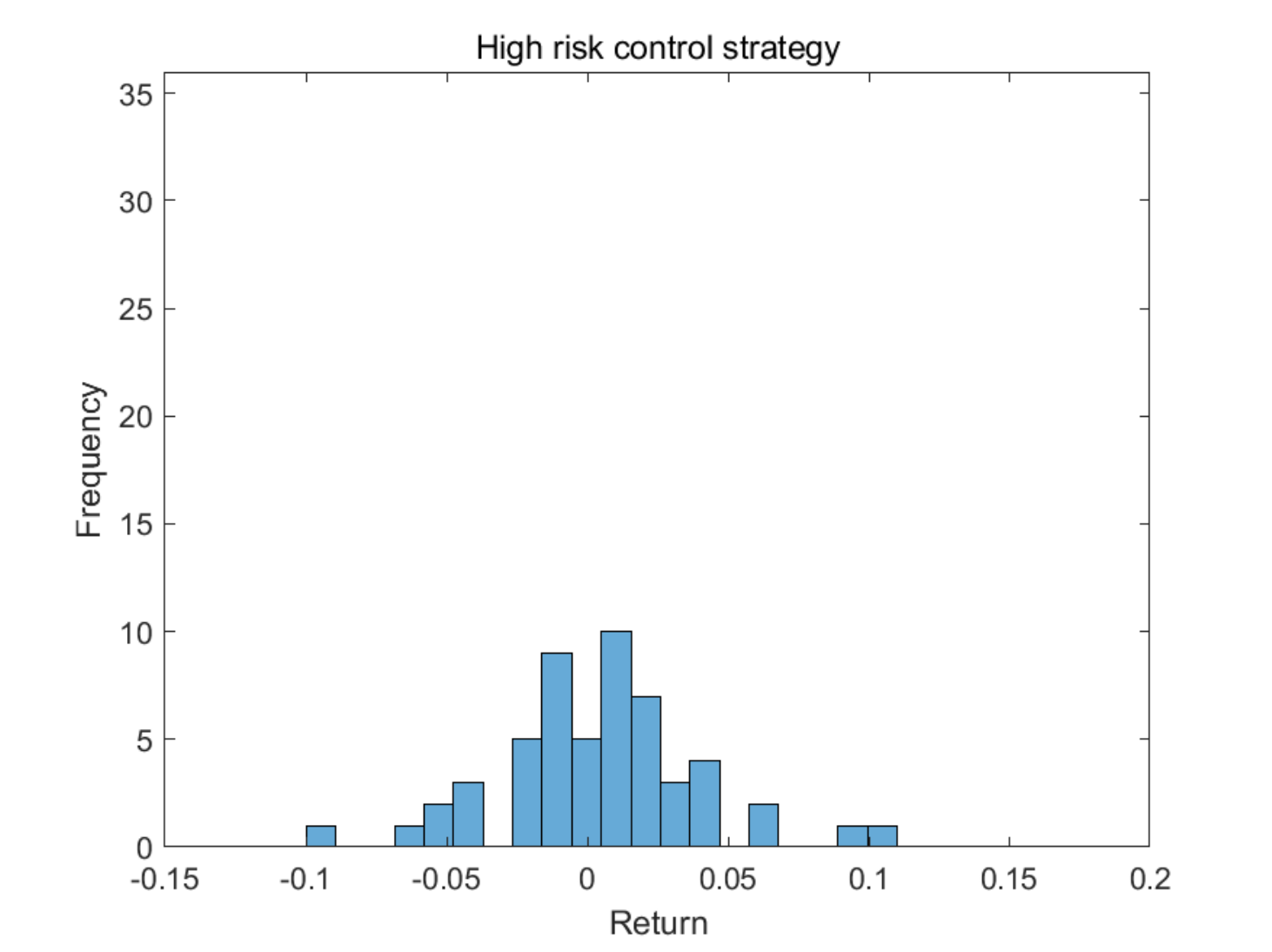}}
  \end{minipage}
  \hfill
  \begin{minipage}{0.5\linewidth}
  \centerline{\includegraphics[width=7.5cm]{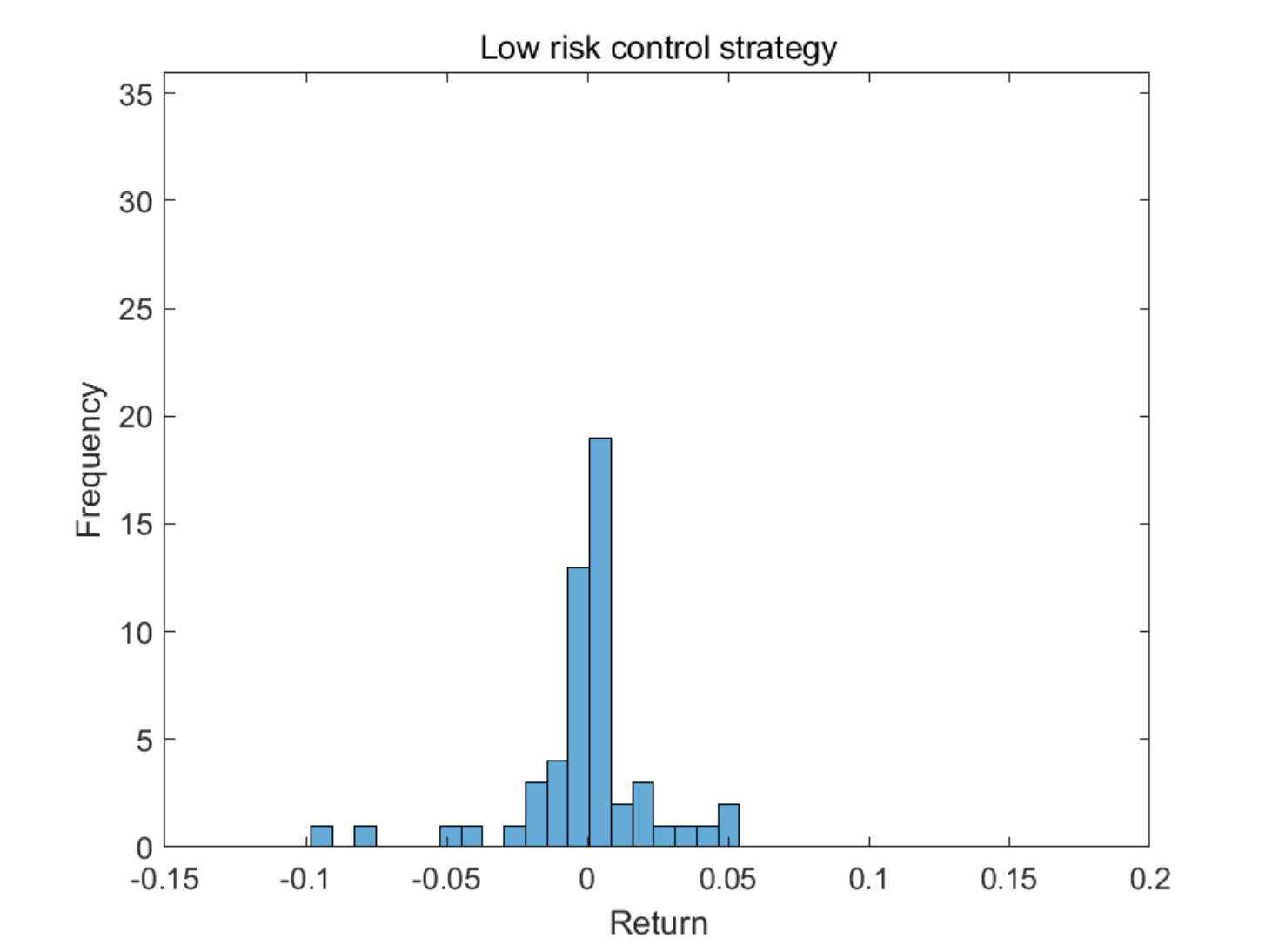}}
  \end{minipage}
  \caption{Out-of-sample return distribution: Data set 2}
  \label{fig13}
\end{figure}

We can see from the out-of-sample period of 2020 that the {\it low risk control} strategy performed poorly during the time of booming market, which is due to the strict systemic risk constraint that is unnecessarily required in such a situation. To take full advantage of our model, the systemic risk constraint should be used with discretion. In the following, we examine the performance of systemic risk control strategy with risk prediction. Specifically, at the beginning of each investment period, we use GJR-GARCH model to predict the variance of each stock. If more than half of stocks have a higher predictive variance than the currently estimated variance, we use the {\it low risk control} strategy, otherwise use the {\it optioned} strategy. This strategy is defined as the {\it discretion optioned} strategy. Similarly, for the pure stock portfolio, if more than half of stocks have a higher predictive variance, we execute stock strategy with $\bar\sigma=0.02$, otherwise with $\bar\sigma=0.1$. We define it as the {\it discretion stock} strategy.
\begin{figure}[H]
  \begin{minipage}{0.5\linewidth}
  \centerline{\includegraphics[width=8cm]{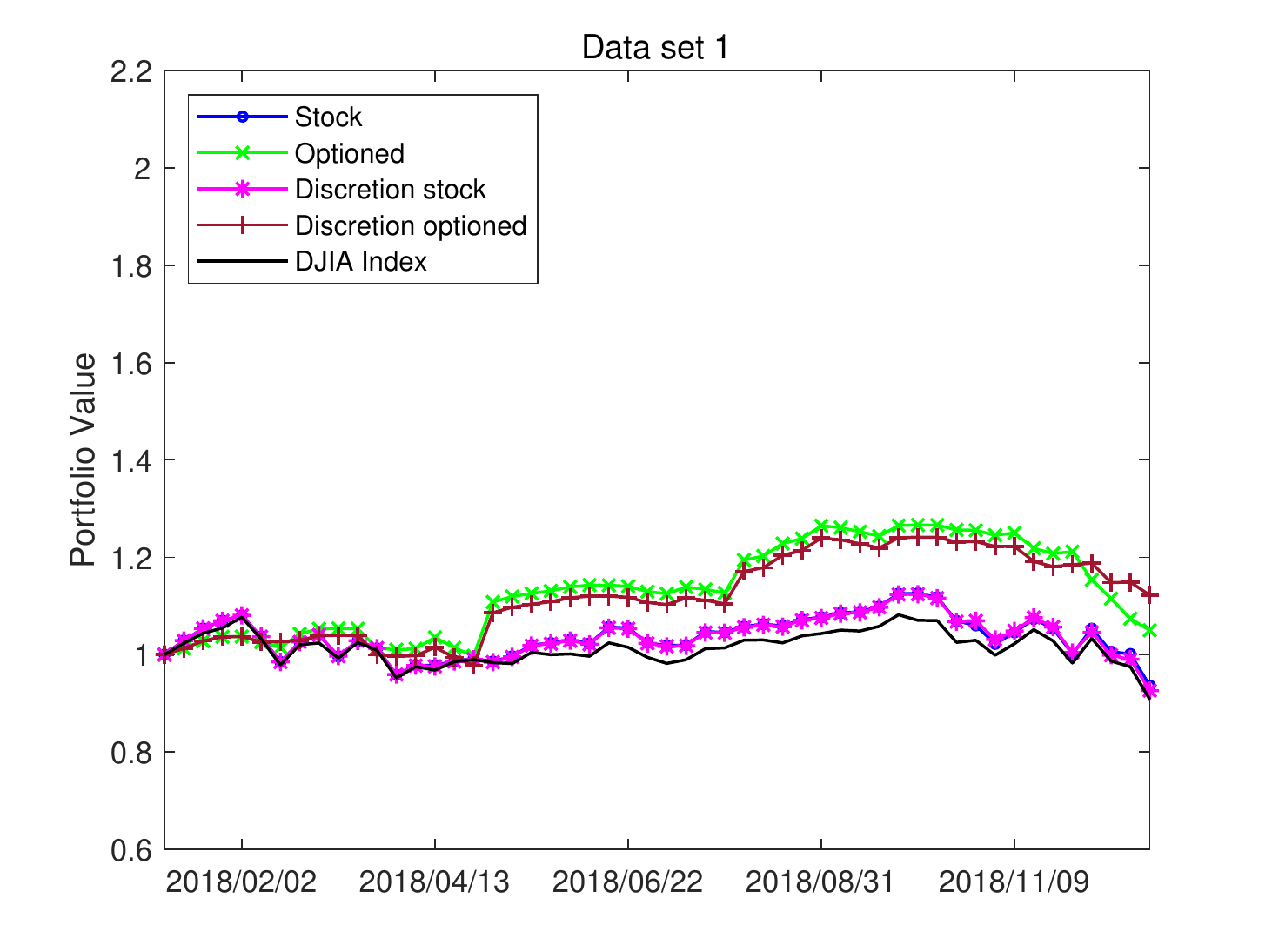}}
  \end{minipage}
  \hfill
  \begin{minipage}{0.5\linewidth}
  \centerline{\includegraphics[width=8cm]{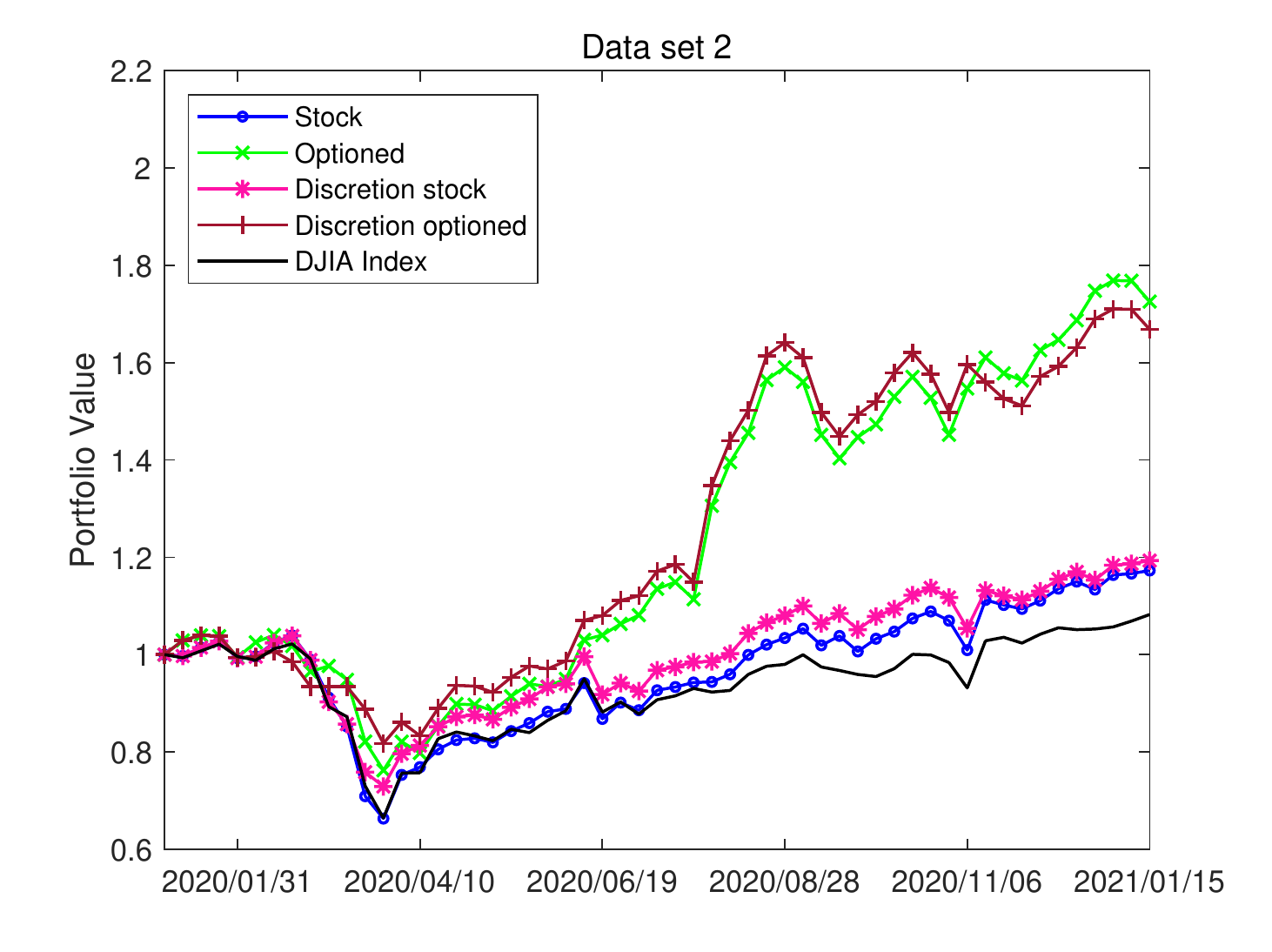}}
  \end{minipage}
  \caption{Out-of-sample performance of strategies with discretion}
  \label{fig14}
\end{figure}

\begin{table}[H]
  \centering
  \caption{Performance measures of different strategies}
\begin{tabular}{clcccccc}
\hline
~&Strategy &  Mean &     Std &         Min &        ADD &         UP ratio &          DS ratio\\
\hline
\multirow{5}*{Data set 1}&Stock & -0.0009 & 0.0264 & -0.0648 &	0.0534 & 0.4198 & -0.0296 \\
~&Optioned& 0.0012 & 0.0237 & -0.0478 & 0.0222 & 0.5395 & 0.0669 \\
~&Discretion stock & -0.0012 & 0.0259 & -0.0648 &	0.0538 & 0.4041 & -0.0360 \\
~&Discretion optioned &  \textbf{0.0025} & \textbf{0.0216} & \textbf{-0.0376} & \textbf{0.0168} &	\textbf{0.6431} & \textbf{0.1644} \\
~&DJIA Index & -0.0016 &	0.0260 & -0.0687 & 0.0608 &	0.4084 & -0.0517 \\
\hline
\multirow{5}*{Data set 2}&Stock & 0.0040 &	0.0450 & -0.1682 & 0.1015 &	0.5073 & 0.0677 \\
~&Optioned & \textbf{0.0112} & 0.0473 & -0.1339 & \textbf{0.0627} & 0.8292 & 0.2531 \\
~&Discretion stock & 0.0040 &	\textbf{0.0373} & -0.1151 & 0.0753 &	0.5507 & 0.0812 \\
~&Discretion optioned & 0.0104 & 0.0431 &	\textbf{-0.0794} & 0.0583 & \textbf{0.9227} &	\textbf{0.2966} \\
~&DJIA Index & 0.0025 & 0.0457 &	-0.1630 & 0.0999 &	0.4609 & 0.0504 \\
\hline
\end{tabular}
\end{table}
As we can see in Figure \ref{fig14}, during 2018, the {\it optioned} strategy still performs best. The {\it discretion optioned} strategy is very close to the {\it optioned} strategy and much better than other strategies. During 2020, no matter in market collapses or market booms, the {\it discretion optioned} strategy has the best performance. It combines the advantage of {\it optioned} strategy and {\it low risk control} strategy. Most of the time, the portfolio value of {\it discretion optioned} strategy comes first. The performance measures in Table 3 also suggest that the {\it discretion optioned} strategy outperforms others. It has the lowest risk measures and highest return measures. Figures \ref{fig15} and \ref{fig16} display the return distribution of different strategies. We can see that the {\it discretion optioned} strategy achieves more high positive returns and maintains relatively small losses.

\begin{figure}[H]
  \begin{minipage}{0.5\linewidth}
  \centerline{\includegraphics[width=7.5cm]{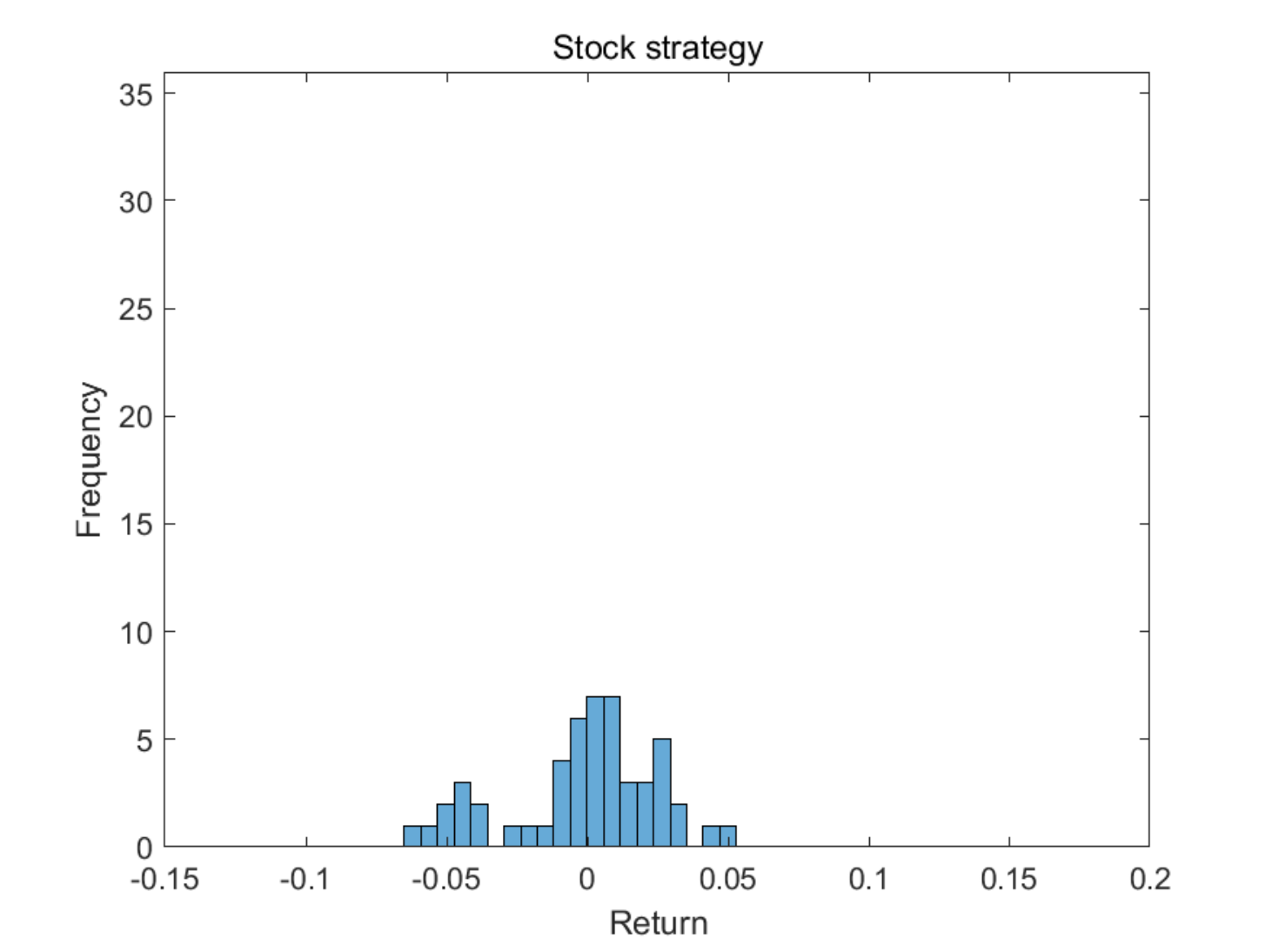}}
  \end{minipage}
  \hfill
  \begin{minipage}{0.5\linewidth}
  \centerline{\includegraphics[width=7.5cm]{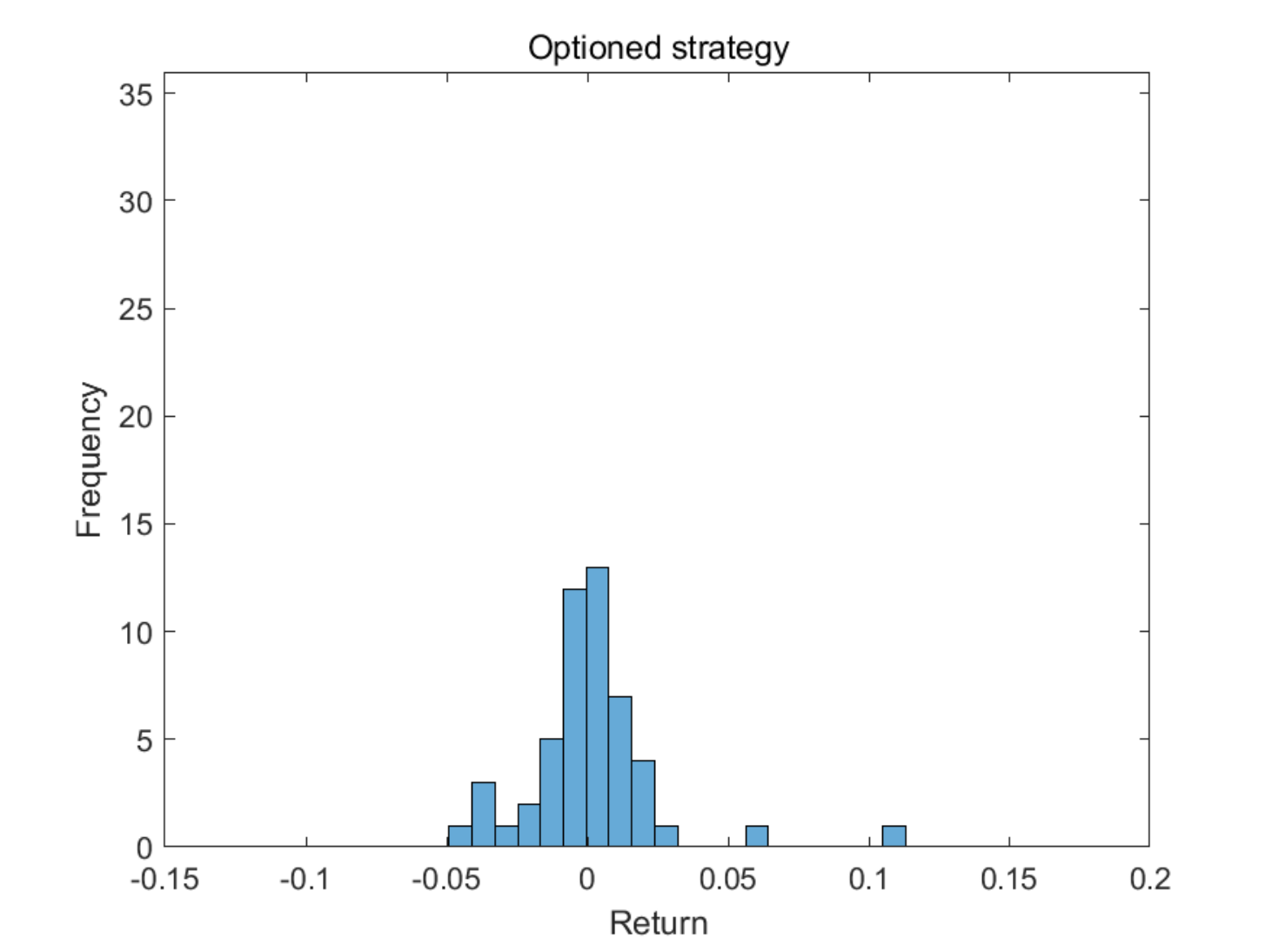}}
  \end{minipage}
  \vfill
  \begin{minipage}{0.5\linewidth}
  \centerline{\includegraphics[width=7.5cm]{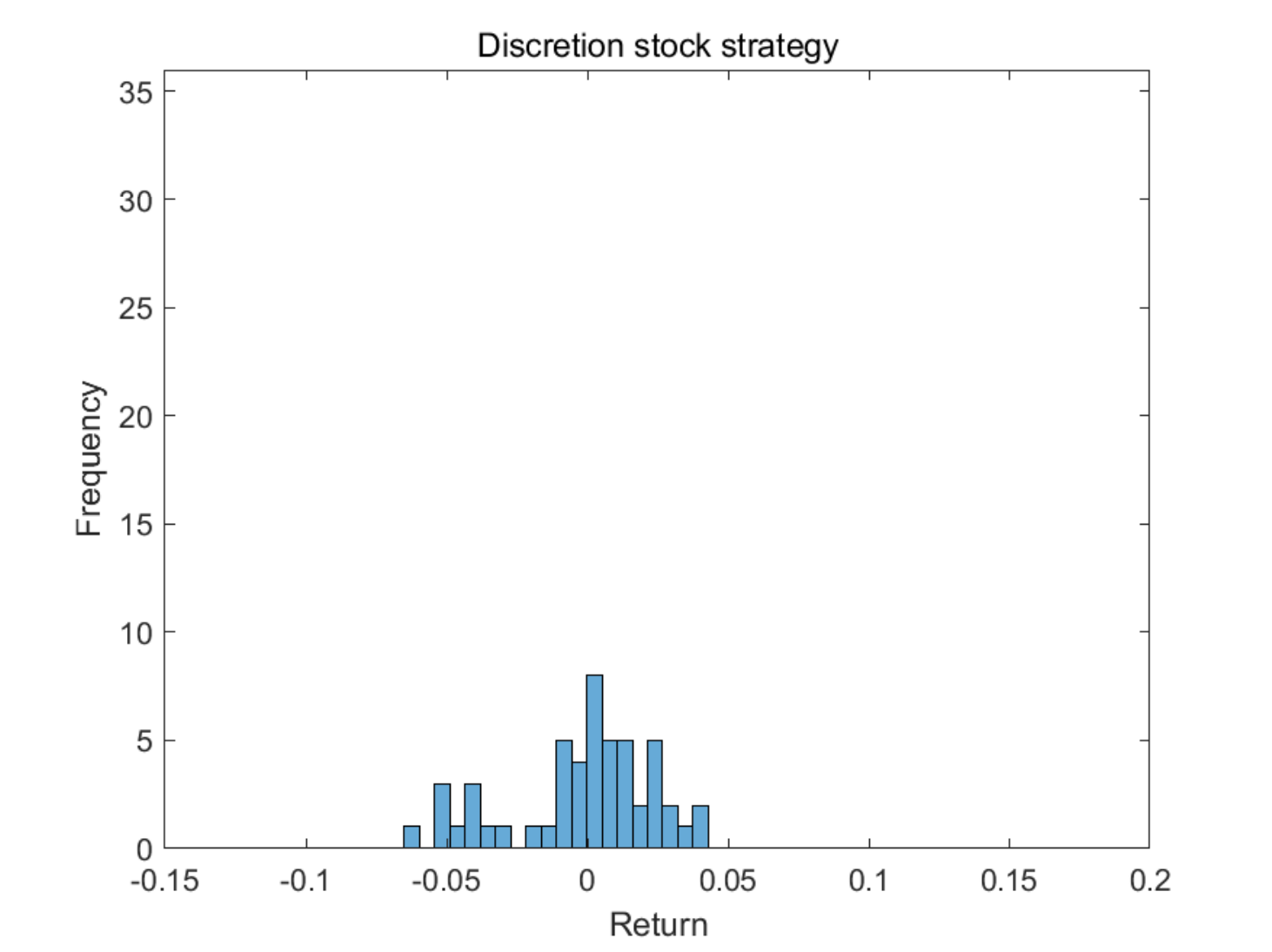}}
  \end{minipage}
  \hfill
  \begin{minipage}{0.5\linewidth}
  \centerline{\includegraphics[width=7.5cm]{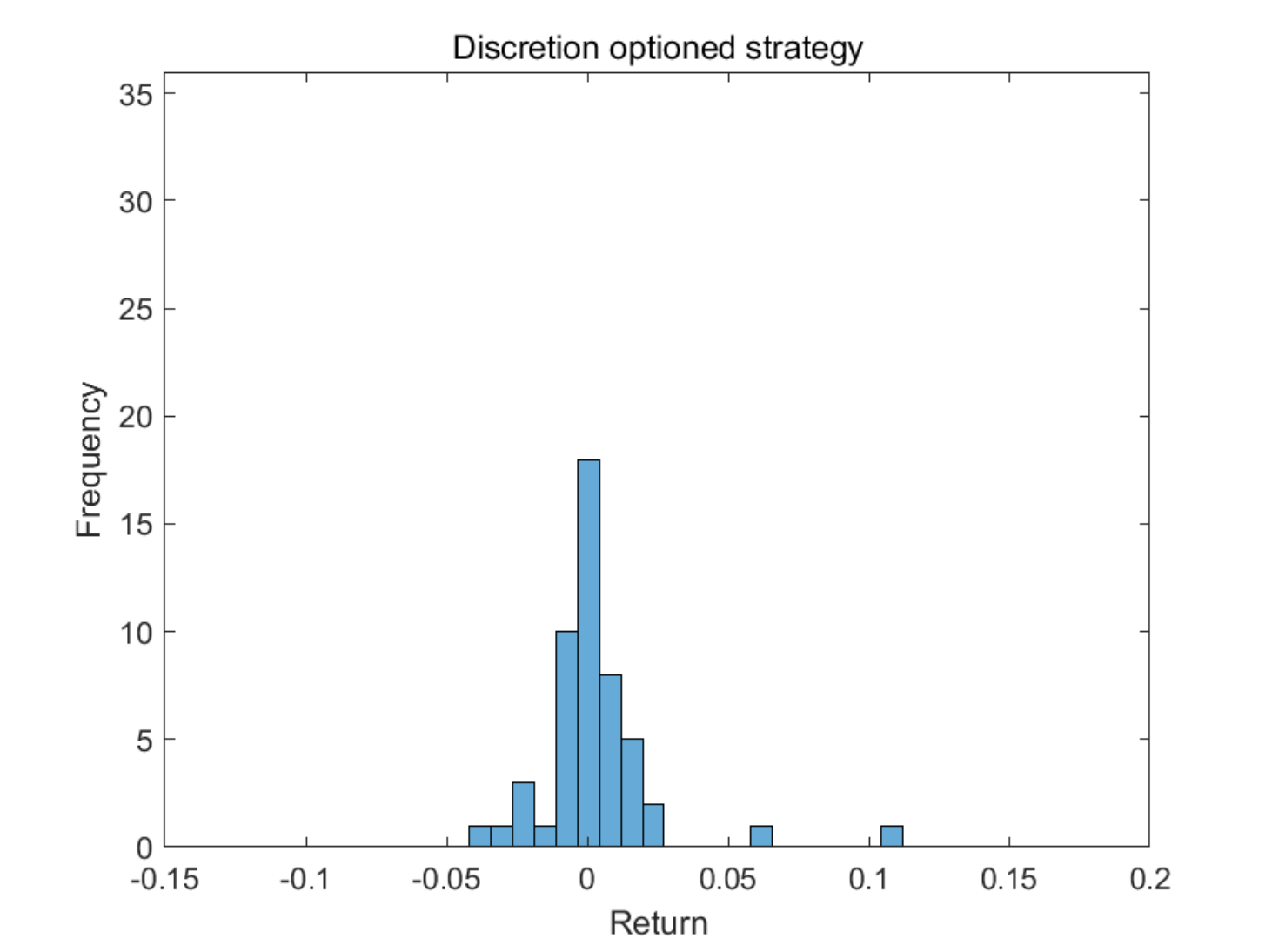}}
  \end{minipage}
  \caption{Out-of-sample return distribution: Data set 1}
  \label{fig15}
\end{figure}

\begin{figure}[H]
  \begin{minipage}{0.5\linewidth}
  \centerline{\includegraphics[width=7.5cm]{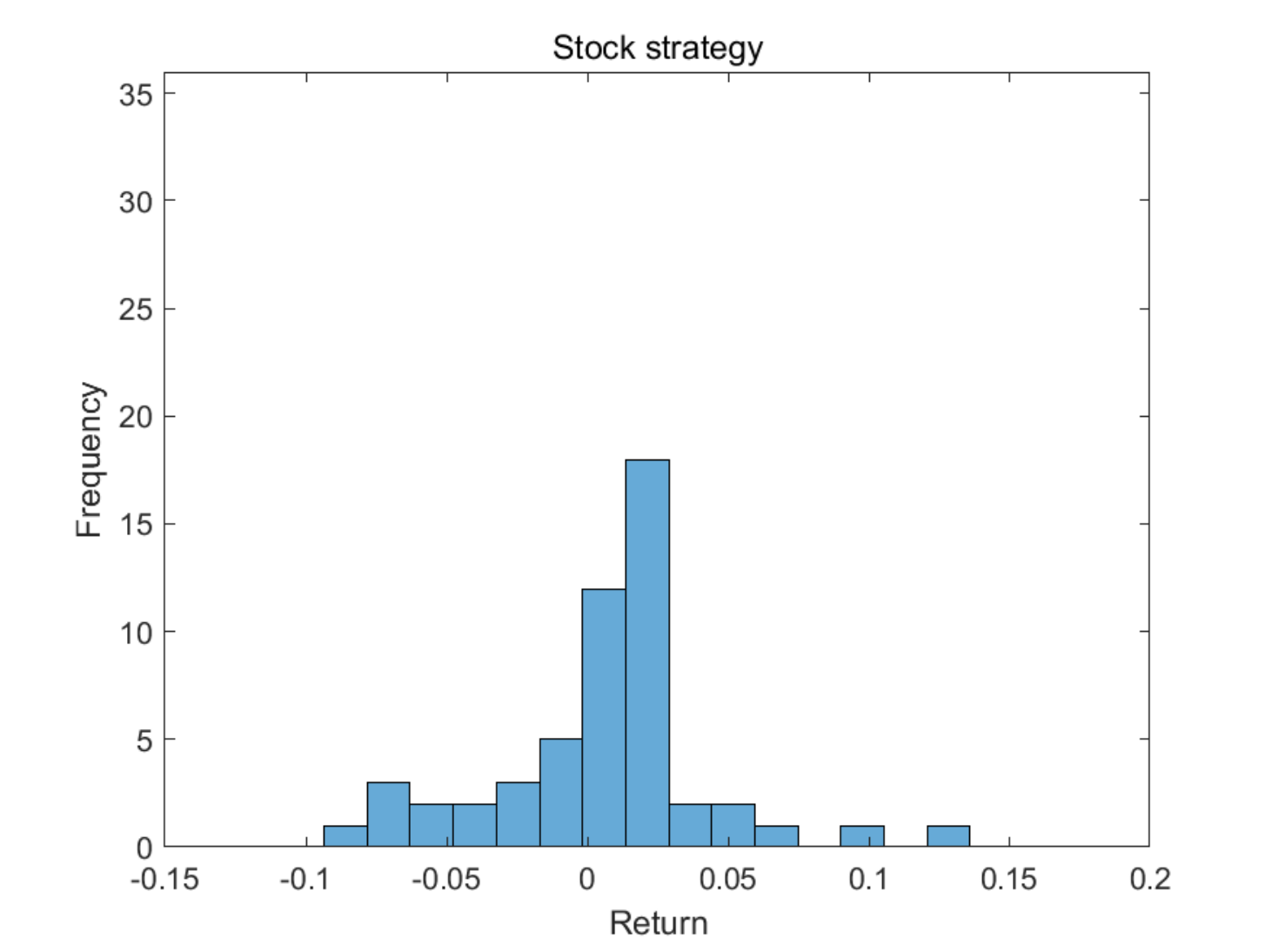}}
  \end{minipage}
  \hfill
  \begin{minipage}{0.5\linewidth}
  \centerline{\includegraphics[width=7.5cm]{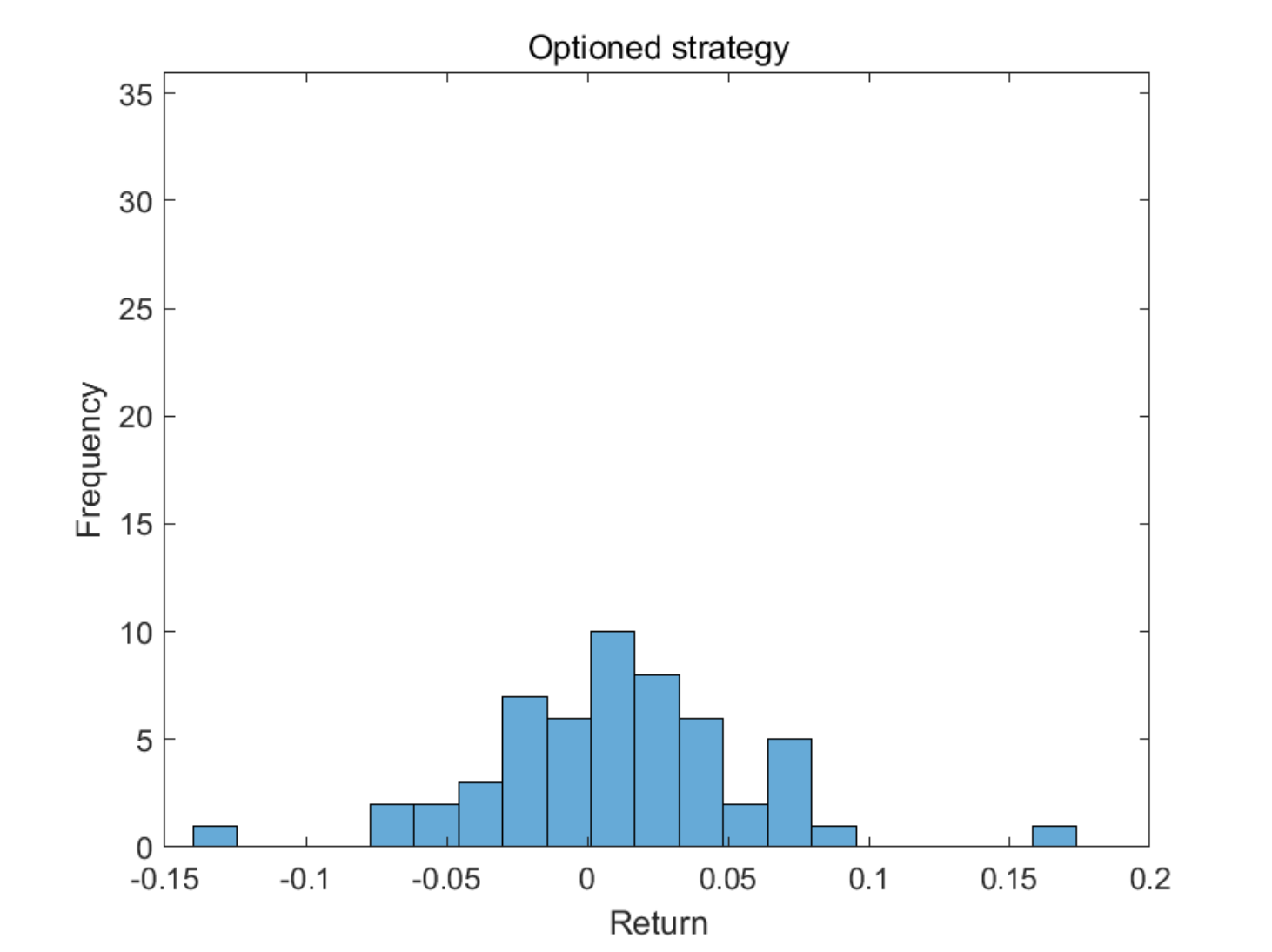}}
  \end{minipage}
  \vfill
  \begin{minipage}{0.5\linewidth}
  \centerline{\includegraphics[width=7.5cm]{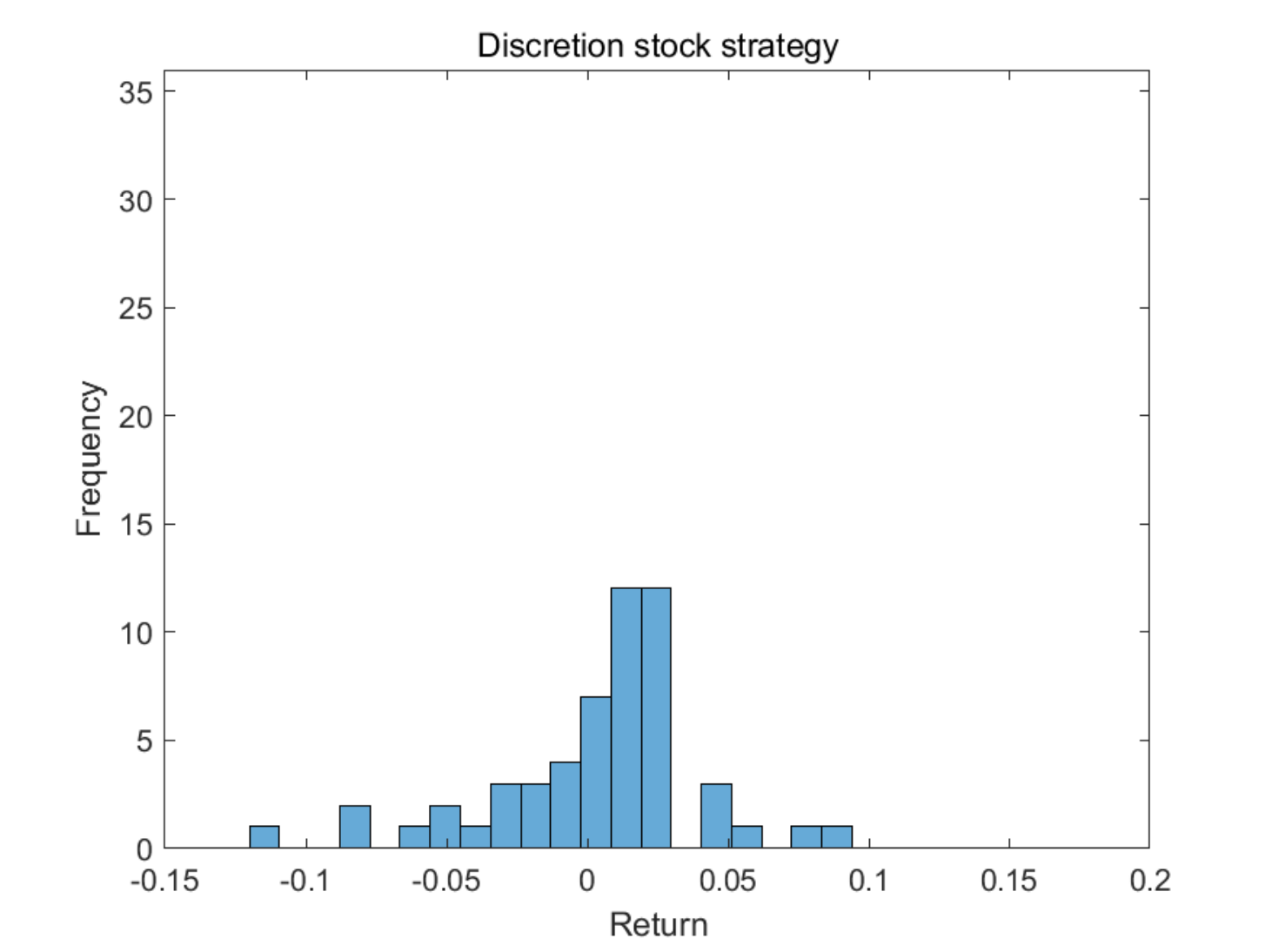}}
  \end{minipage}
  \hfill
  \begin{minipage}{0.5\linewidth}
  \centerline{\includegraphics[width=7.5cm]{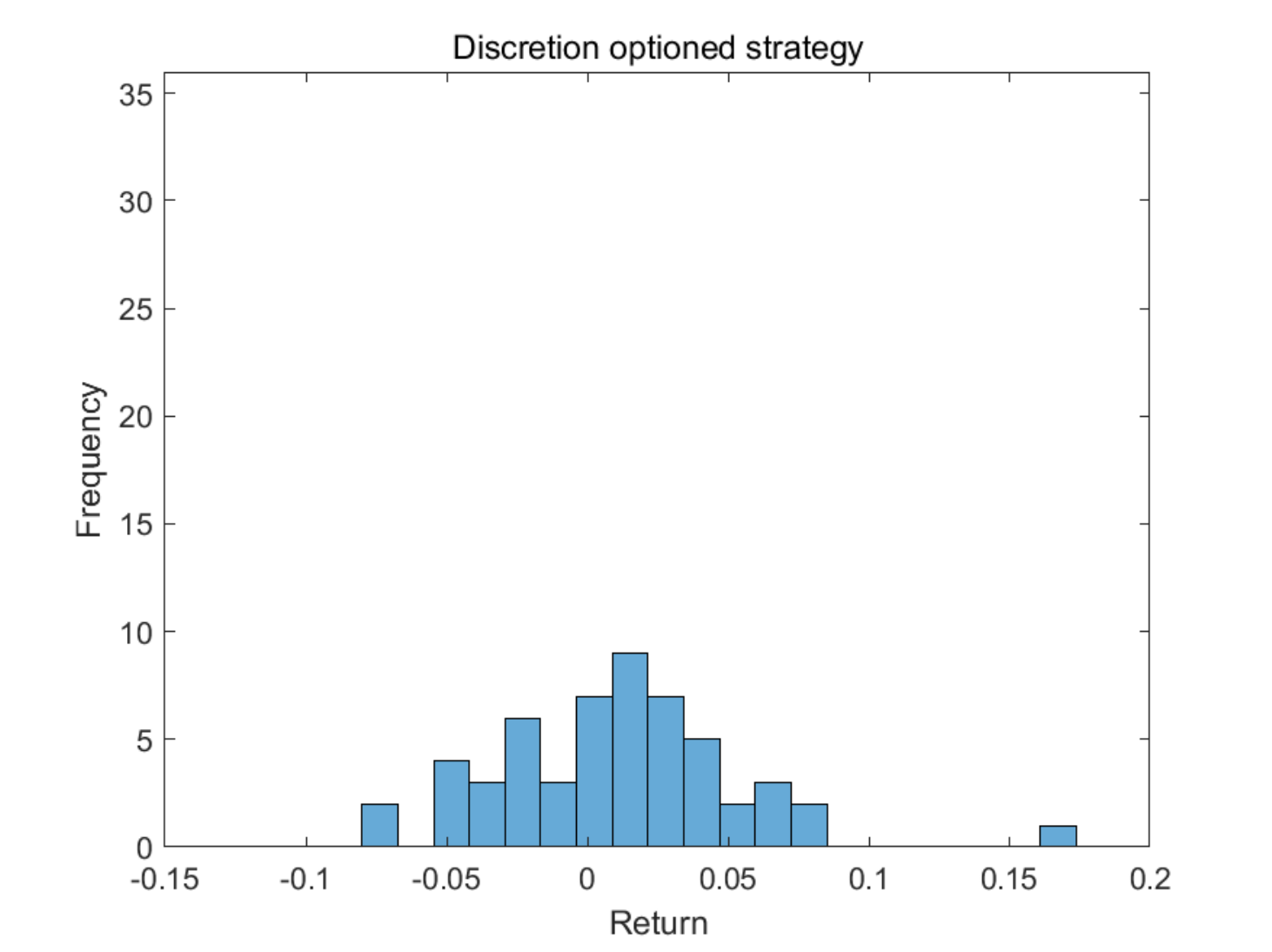}}
  \end{minipage}
  \caption{Out-of-sample return distribution: Data set 2}
  \label{fig16}
\end{figure}

\section{Conclusion}
This paper explores whether the systemic risk of portfolios can be controlled and how it can be controlled. It is illustrated that the systemic risk of pure stock portfolios is usually uncontrollable due to the contagion effect and the seesaw effect.  We demonstrate that these two effects forming systemic risk can be well treated via the correlation hedging and the extreme loss hedging by introducing options. We further derive a reasonable approximation of the distribution of optioned portfolios under some mild conditions and show that the optimal portfolio selection with systemic risk control can be formulated as a convex SOCP problem. which is critical to facilitate the portfolio optimization.

In addition, we examine the theoretical results via simulation and test the performance of the proposed model with empirical study. It is shown that during the time of market downturn, the optioned portfolio with systemic risk control outperforms other strategies, no matter in the in-sample test and the out-of-sample test. However, the systemic risk constraint makes portfolios perform poorly during the market booms, which is not surprising since everything has two sides and the model just focuses on risk control. Therefore, the systemic risk constrained model should be combined with some risk prediction methods and follow the rule of discretion to achieve the balance between risk control and return pursuit.

\subsection*{Appendix A: Proof of Proposition \ref{moment}}
\begin{proof}
We first introduce the following lemma on mean and variance of a quadratic form of normally distributed vector (Proposition 1 of \cite{zhu2020}).
\begin{lemma}\label{un-m-v}
Suppose
\begin{eqnarray*}
z=\frac{1}{2}\bm{\xi}^{\top}B\bm{\xi}+\bm{b}^{\top}\bm{\xi}+b_0,
\end{eqnarray*}
where $\bm{\xi}\sim\mathcal{N}(\bm{\mu},\Sigma)$, $B$ is a symmetric matrix, $\bm{b}$ and $b_0$ are parameters. Then the mean and variance of $z$ are given by
\begin{eqnarray*}
\mathbb{E}(z)&=&\frac{1}{2}\bm{\mu}^{\top}B\bm{\mu}+\frac{1}{2}{\rm tr}(B\Sigma)+\bm{b}^{\top}\bm{\mu}+b_0,\\
\mathbb{V}(z)&=&(B\bm{\mu}+\bm{b})^{\top}\Sigma(B\bm{\mu}+\bm{b})+\frac{1}{2}{\rm tr}((B\Sigma)^2).
\end{eqnarray*}
\end{lemma}
Now we turn to the value change of the portfolio. Conditioning on $\mathscr{G}=\{\Delta p_i=-k_i,i\in\mathcal{I}\}$, the conditional value change of the portfolio can be written as
\begin{eqnarray*}
&&\Delta v^{\mathscr{G}}=\frac{1}{2}\Delta\bm{p}_{\mathcal{J}}^\top\Gamma_{\mathcal{J}\mathcal{J}}\Delta \bm{p}_{\mathcal{J}}+
(\bm{\delta}_{\mathcal{J}}-\Gamma_{\mathcal{J}\mathcal{I}}\bm{k})^\top\Delta\bm{p}_{\mathcal{J}}+\theta\Delta t-\bm{\delta}_{\mathcal{I}}^\top\bm{k}+\frac{1}{2}\bm{k}^\top\Gamma_{\mathcal{I}\mathcal{I}}\bm{k},
\end{eqnarray*}
where the conditional distribution of $$\Delta p_{\mathcal{J}}|\mathscr{G}\sim\mathcal{N}(\bm{c},E)$$ is normally distributed. By the properties of multivariate normal distribution, since $\Delta\bm{p}\sim\mathcal{N}(\bm{\mu},\Sigma)$, we have $\bm{c} = (c_j)_{|\mathcal{J}|\times 1}=\bm{\mu}_{\mathcal{J}}-\Sigma_{\mathcal{J}\mathcal{I}}\Sigma_{\mathcal{I}\mathcal{I}}^{-1}(\bm{k}+\bm{\mu}_{\mathcal{I}})$ and $E = (e_{jj})_{|\mathcal{J}|\times |\mathcal{J}|}= \Sigma_{\mathcal{J}\mathcal{J}}-\Sigma_{\mathcal{J}\mathcal{I}}\Sigma_{\mathcal{I}\mathcal{I}}^{-1}\Sigma_{{\mathcal{I}}{\mathcal{J}}}$. Therefore, the conditional value change of the portfolio is also a quadratic form of normally distributed vector. By Lemma \ref{un-m-v}, the conditional mean  and variance are respectively given as
\begin{eqnarray*}
\mathbb{E}(\Delta v|\mathscr{G})&=&\frac{1}{2}\bm{c}^{\top}\Gamma_{\mathcal{JJ}}\bm{c}+\frac{1}{2}{\rm tr}(\Gamma_{\mathcal{JJ}}E)+(\bm{\delta}_{\mathcal{J}}-\Gamma_{\mathcal{JI}}\bm{k})^{\top}\bm{c}+\theta\Delta t-\bm{\delta}_{\mathcal{I}}^\top\bm{k}+\frac{1}{2}\bm{k}^\top\Gamma_{\mathcal{I}\mathcal{I}}\bm{k}\\
&=&\frac{1}{2}\bm{h}^{\top}\Gamma\bm{h}+\frac{1}{2}{\rm tr}(\Gamma_{\mathcal{JJ}}E)+\bm{\delta}^{\top}\bm{h}+\theta\Delta t,\\
\mathbb{V}(\Delta v|\mathscr{G})&=&(\Gamma_{\mathcal{JJ}}\bm{c}+\bm{\delta}_{\mathcal{J}}-\Gamma_{\mathcal{JI}}\bm{k})^{\top}E(\Gamma_{\mathcal{JJ}}\bm{c}+\bm{\delta}_{\mathcal{J}}-\Gamma_{\mathcal{JI}}\bm{k})+\frac{1}{2}{\rm tr}\left((\Gamma_{\mathcal{JJ}}E)^2\right)\\
&=&(\Gamma_{\mathcal{J}\cdot}\bm{h}+\bm{\delta}_{\mathcal{J}})^{\top}E(\Gamma_{\mathcal{J}\cdot}\bm{h}+\bm{\delta}_{\mathcal{J}})+\frac{1}{2}{\rm tr}\left((\Gamma_{\mathcal{JJ}}E)^2\right),
\end{eqnarray*}
where $\bm{h}=\left\{\begin{array}{ccc}\bm{h}_{\mathcal{I}}&=&-\bm{k}\\
\bm{h}_{\mathcal{J}}&=&\bm{c}\end{array}\right.$.

We can further rewrite the conditional mean and variance in the form of functions with respect to decision vector $\bm{x}$ and $\bm{y}$ as follows.
\begin{eqnarray*}
\mathbb{E}(\Delta v(\bm{x},\bm{y})|\mathscr{G})&=&\frac{1}{2}\bm{h}^{\top}\left(\sum_{i=1}^nx_i\Gamma^i\right)\bm{h}+\frac{1}{2}{\rm tr}\left(\left(\sum_{i=1}^nx_i\Gamma^i_{\mathcal{JJ}}\right)E\right)+\left(\sum_{i=1}^nx_i\bm{\delta}^i+\bm{y}\right)^{\top}\bm{h}\\
&&+\left(\sum_{i=1}^nx_i\theta^i\right)\Delta t\\
&=&\sum_{i=1}^n\left(\frac{1}{2}\bm{h}^{\top}\Gamma^i\bm{h}+\frac{1}{2}{\rm tr}(\Gamma^{i}_{\mathcal{JJ}}E)+(\bm{\delta}^i)^{\top}\bm{h}+\theta^i\Delta t\right)x_i+\bm{h}^{\top}\bm{y}\\
&=&\bm{g}^\top\bm{x}+\bm{h}^\top\bm{y},
\end{eqnarray*}
where $\bm{g}=(g_i)_{n}=\left(\frac{1}{2}\bm{h}^{\top}\Gamma^i\bm{h}+\frac{1}{2}{\rm tr}(\Gamma^{i}_{\mathcal{JJ}}E)+(\bm{\delta}^i)^{\top}\bm{h}+\theta^i\Delta t\right)_{n}$.
\begin{eqnarray*}
\mathbb{V}(\Delta v(\bm{x},\bm{y})|\mathscr{G})&=&
\left(\left(\sum_{i=1}^nx_i\Gamma^i_{\mathcal{J}\cdot}\right)\bm{h}+\sum_{i=1}^nx_i\bm{\delta}^i_{\mathcal{J}}+\bm{y}_{\mathcal{J}}\right)^{\top}E\left(\left(\sum_{i=1}^nx_i\Gamma^i_{\mathcal{J}\cdot}\right)\bm{h}+\sum_{i=1}^nx_i\bm{\delta}^i_{\mathcal{J}}+\bm{y}_{\mathcal{J}}\right)\\
&&+\frac{1}{2}{\rm tr}\left(\left(\left(\sum_{i=1}^nx_i\Gamma^i_{\mathcal{JJ}}\right)E\right)^2\right)\\
&=&\left(\sum_{i=1}^n\left(\Gamma^i_{\mathcal{J}\cdot}\bm{h}+\bm{\delta}^i_{\mathcal{J}}\right)x_i+\bm{y}_{\mathcal{J}}\right)^{\top}E\left(\sum_{i=1}^n\left(\Gamma^i_{\mathcal{J}\cdot}\bm{h}+\bm{\delta}^i_{\mathcal{J}}\right)x_i+\bm{y}_{\mathcal{J}}\right)\\
&&+\frac{1}{2}\sum_{i=1}^n\sum_{j=1}^nx_ix_j{\rm tr}\left(\Gamma^i_{\mathcal{JJ}}E\Gamma^j_{\mathcal{JJ}}E\right)\\
&=&(\bm{x}^{\top},\bm{y}_{\mathcal{J}}^{\top})R\left(\begin{array}{c}\bm{x}\\\bm{y}_{\mathcal{J}}\end{array}\right)+\frac{1}{2}{\bm x}^\top S\bm{x},
\end{eqnarray*}
where
\begin{eqnarray*}
R&=&\left(\Gamma^1_{\mathcal{J}\cdot}\bm{h}+\bm{\delta}^1_{\mathcal{J}},\cdots,\Gamma^n_{\mathcal{J}\cdot}\bm{h}+\bm{\delta}^n_{\mathcal{J}},{\rm I}\right)^{\top}
E\left(\Gamma^1_{\mathcal{J}\cdot}\bm{h}+\bm{\delta}^1_{\mathcal{J}},\cdots,\Gamma^n_{\mathcal{J}\cdot}\bm{h}+\bm{\delta}^n_{\mathcal{J}},{\rm I}\right),\\
S&=&(s_{ij})_{n\times n}=\left({\rm tr}\left(E\Gamma^i_{\mathcal{J}\mathcal{J}}E\Gamma^j_{\mathcal{J}\mathcal{J}}\right)\right)_{n\times n}.
\end{eqnarray*}
The proof is completed.
\end{proof}

\subsection*{Appendix B: Proof of Corollary \ref{delta-gamma-hedge}}
\begin{proof}
We use the contradiction. If $|\mathcal{D}|< m-2$, there must be at least three optioned assets satisfying $\delta_i\neq0$ or $\gamma_i\neq0$ (suppose $i=1,2,3$ without loss of generality). According to the assumption ${\rm Cov}(\Delta\varphi_i,\Delta\varphi_j)=0$, $i\neq j$, $i,j=1,2,3$, we have
\begin{eqnarray}
\frac{1}{2}\gamma_1\gamma_2\sigma_{12}&=&-(\gamma_1\mu_1+\delta_1)(\gamma_2\mu_2+\delta_2),\label{d1}\\
\frac{1}{2}\gamma_2\gamma_3\sigma_{23}&=&-(\gamma_2\mu_2+\delta_2)(\gamma_3\mu_3+\delta_3),\label{d2}\\
\frac{1}{2}\gamma_3\gamma_1\sigma_{31}&=&-(\gamma_3\mu_3+\delta_3)(\gamma_1\mu_1+\delta_1).\label{d3}
\end{eqnarray}
Multiplying two sides of the above equations, respectively, yields
\begin{eqnarray}\label{ine_zss}
0\leq\frac{1}{8}\gamma_1^2\gamma_2^2\gamma_3^2\sigma_{12}\sigma_{23}\sigma_{31}=-(\gamma_1\mu_1+\delta_1)^2(\gamma_2\mu_2+\delta_2)^2(\gamma_3\mu_3+\delta_3)^2\leq0.
\end{eqnarray}

Notice that the inequality on the left side of (\ref{ine_zss}) holds since the correlation is strictly positive according to the assumption. By (\ref{ine_zss}), we have $\gamma_1\gamma_2\gamma_3=0$ since $\sigma_{12}\sigma_{23}\sigma_{31}>0$. Without loss of generality, we assume $\gamma_1=0$, and then $\delta_1\neq0$ according to the assumption. Equation (\ref{d1}) implies that $\gamma_2\mu_2+\delta_2=0$. Replacing it into the equation (\ref{d2}) derives $\gamma_2\gamma_3=0$. If $\gamma_2=0$, we have $\delta_2=0$ from $\gamma_2\mu_2+\delta_2=0$, which is contradictory. So, there must be $\gamma_3=0$ and $\delta_3\neq0$. Replacing it into equation (\ref{d3}) yields $\gamma_1\mu_1+\delta_1=0$. Combining with $\gamma_1=0$, we derive $\delta_1=0$, which is also contradictory to the assumption. The proof is completed.
\end{proof}

\subsection*{Appendix C: Proof of Corollary \ref{frontier}}
\begin{proof}
For a given stock portfolio $\bm{y}$, set the optioned portfolio $\bm{z}=\bm{y}$. By Proposition \ref{dominate}, for any $i\in\{1,\cdots,m\}$, there exists $(\delta_i,\gamma_i,\theta_i)$ satisfying
\begin{eqnarray*}
&&\mathbb{E}(\Delta\varphi_i)>\mathbb{E}(\Delta p_i)~\mbox{and}~
\mathbb{V}(\Delta\varphi_i)<\mathbb{V}(\Delta p_i).
\end{eqnarray*}
Notice that $\bm{y}>0$ from the assumption,
\begin{eqnarray*}
&&\mathbb{E}(\Delta v(\bm{y})) = \sum_{i=1}^my_i\mathbb{E}(\Delta p_i)<\sum_{i=1}^mz_i\mathbb{E}(\Delta\varphi_i)=\mathbb{E}(\Delta v(\bm{z})).
\end{eqnarray*}
Furthermore, according to Proposition {\ref{corr}} and $\rho_{ij}>0$ for any $i,j\in\{1,\cdots,m\}$, we have $\rho_{ij}\geq|\varrho_{ij}|\geq\varrho_{ij}$ and then
\begin{eqnarray*}
\mathbb{V}(\Delta v(\bm{y}))&=&\sum_{i=1}^m\sum_{j=1}^my_iy_j\mathbb{V}^{\frac{1}{2}}(\Delta p_i)\mathbb{V}^{\frac{1}{2}}(\Delta p_j)\rho_{ij}\\
&>&\sum_{i=1}^m\sum_{j=1}^mz_iz_j\mathbb{V}^{\frac{1}{2}}(\Delta\varphi_i)\mathbb{V}^{\frac{1}{2}}(\Delta\varphi_j)\varrho_{ij}=\mathbb{V}(\Delta v(\bm{z})).
\end{eqnarray*}

Given $\mathscr{G}=\{\Delta p_i=-k_i,i\in\mathcal{I}\}$, by proposition {\ref{dominate}} we obtain
\begin{eqnarray*}
&&\mathbb{E}(\Delta\varphi_i|\mathscr{G})>\mathbb{E}(\Delta p_i|\mathscr{G})~\mbox{and}~\mathbb{V}(\Delta\varphi_i|\mathscr{G})<\mathbb{V}(\Delta p_i|\mathscr{G})
\end{eqnarray*}
for $i\in\mathcal{J}$ and
\begin{eqnarray*}
&&\mathbb{E}(\Delta\varphi_i|\mathscr{G})>\mathbb{E}(\Delta p_i|\mathscr{G})
\end{eqnarray*}
for $i\in\mathcal{I}$.
Recalling $\bm{y}>0$ by assumption, we obtain
\begin{eqnarray*}
&&\mathbb{E}(\Delta v(\bm{y})|\mathscr{G}) = \sum_{i=1}^my_i\mathbb{E}(\Delta p_i)|\mathscr{G})<\sum_{i=1}^mz_i\mathbb{E}(\Delta\varphi_i)|\mathscr{G})=\mathbb{E}(\Delta v(\bm{z})|\mathscr{G}).
\end{eqnarray*}
By proposition \ref{corr} and $\rho^{\mathscr{G}}_{ij}>0$ for any $i,j\in\{1,\cdots,m\}$, we obtain $\rho^{\mathscr{G}}_{ij}\geq|\varrho^{\mathscr{G}}_{ij}|\geq\varrho^{\mathscr{G}}_{ij}$ and then
\begin{eqnarray*}
\mathbb{V}(\Delta v(\bm{y})|\mathscr{G}) &=& \sum_{i\in\mathcal{J}}\sum_{j\in\mathcal{J}}y_iy_j\mathbb{V}^{\frac{1}{2}}(\Delta p_i|\mathscr{G})\mathbb{V}^{\frac{1}{2}}(\Delta p_j|\mathscr{G})\rho_{ij}^{\mathscr{G}}\\
&>&
\sum_{i\in\mathcal{J}}\sum_{j\in\mathcal{J}}z_iz_j\mathbb{V}^{\frac{1}{2}}(\Delta\varphi_i|\mathscr{G})\mathbb{V}^{\frac{1}{2}}(\Delta\varphi_j|\mathscr{G})\varrho_{ij}^{\mathscr{G}}=\mathbb{V}(\Delta v(\bm{z})|\mathscr{G}).
\end{eqnarray*}
Finally, combining the above two inequalities yields
\begin{eqnarray*}
CoVaR_q^{\Delta v(\bm{y})|\mathscr{G}} &=& \alpha_q\mathbb{V}^{\frac{1}{2}}(\Delta v(\bm{y})|\mathscr{G})-\mathbb{E}(\Delta v(\bm{y})|\mathscr{G})\\
&>& \alpha_q\mathbb{V}^{\frac{1}{2}}(\Delta v(\bm{z})|\mathscr{G})-\mathbb{E}(\Delta v(\bm{z})|\mathscr{G}) = CoVaR_q^{\Delta v(\bm{z})|\mathscr{G}}.
\end{eqnarray*}
The proof is completed.
\end{proof}

\subsection*{Appendix D: Reformulation of the conditional value change of portfolio}
A general version is provided by \cite{zhu2020}. However, to make this paper self-contained, we provide the details of how the conditional value change of the portfolio is translated to equation (\ref{reform_quad}).

Conditioning on $\mathscr{G}=\{\Delta p_i=-k_i,i\in\mathcal{I}\}$, the conditional value change of the portfolio can be written as
\begin{eqnarray*}
&&\Delta v^{\mathscr{G}}=\frac{1}{2}\Delta\bm{p}_{\mathcal{J}}^\top\Gamma_{\mathcal{J}\mathcal{J}}\Delta \bm{p}_{\mathcal{J}}+
\bm{s}^\top\Delta\bm{p}_{\mathcal{J}}+c_0,
\end{eqnarray*}
where $\bm{s}=\bm{\delta}_{\mathcal{J}}-\Gamma_{\mathcal{J}\mathcal{I}}\bm{k}$ and $c_0=\theta\Delta t-\bm{\delta}_{\mathcal{I}}^\top\bm{k}+\frac{1}{2}\bm{k}^\top\Gamma_{\mathcal{I}\mathcal{I}}\bm{k}$. Denote $\Delta\overline{\bm{p}}=E^{-\frac{1}{2}}\Delta\bm{p}_{\mathcal{J}}$. Since $\Delta\bm{p}_{\mathcal{J}}\sim\mathcal{N}(\bm{c},E)$, we have $\Delta\overline{\bm{p}}\sim\mathcal{N}(E^{-\frac{1}{2}}\bm{c},{\rm I})$.

The conditional value change of the portfolio can be rewritten as
\begin{eqnarray*}
&&\Delta v^{\mathscr{G}}=\frac{1}{2}\Delta\overline{\bm{p}}^{\top}E^{\frac{1}{2}}\Gamma_{\mathcal{J}\mathcal{J}}E^{\frac{1}{2}}\Delta\overline{\bm{p}}+
(E^{\frac{1}{2}}\bm{s})^\top\Delta\overline{\bm{p}}+c_0.
\end{eqnarray*}
Notice that $E^{\frac{1}{2}}\Gamma_{\mathcal{J}\mathcal{J}}E^{\frac{1}{2}}$ is a symmetric matrix. It can be decomposed as $E^{\frac{1}{2}}\Gamma_{\mathcal{J}\mathcal{J}}E^{\frac{1}{2}}=D\Lambda D^{\top}$, where $D$ is an orthogonal matrix and $\Lambda$ is a diagonal matrix whose diagonal elements consisting of eigenvalues $(\lambda_1,\cdots,\lambda_{m'})$. We assume $\lambda_1,\cdots,\lambda_h$ $(h\leq m')$ are the nonzero eigenvalues. Denote $\bm{q}=D^{\top}\Delta\overline{\bm{p}}$. Thus we have $\bm{q}\sim\mathcal{N}(D^{\top}E^{-\frac{1}{2}}\bm{c},{\rm I})$.

Furthermore, denote $\bm{\iota}=D^{\top}E^{\frac{1}{2}}\bm{s}$. Then the conditional value change of the portfolio can be reformulated as
\begin{eqnarray}
\Delta v^{\mathscr{G}}&=&\frac{1}{2}\Delta\overline{\bm{p}}^{\top}D\Lambda D^{\top}\Delta\overline{\bm{p}}+
(E^{\frac{1}{2}}\bm{s})^\top\Delta\overline{\bm{p}}+c_0 \nonumber\\
&=&\frac{1}{2}\bm{q}^{\top}\Lambda\bm{q}+\bm{\iota}^{\top}\bm{q}+c_0 \nonumber\\
&=&\frac{1}{2}\sum_{i=1}^h\left(\lambda_iq_i^2+2\iota_iq_i\right)+\sum_{i=h+1}^{m'}\iota_iq_i+c_0 \nonumber\\
&=&\frac{1}{2}\sum_{i=1}^h\lambda_i(q_i+\frac{\iota_i}{\lambda_i})^2+\sum_{i=h+1}^{m'}\iota_iq_i+\tau,\nonumber
\end{eqnarray}
where $\tau=c_0-\frac{1}{2}\sum_{i=1}^h\frac{\iota_i^2}{\lambda_i}$. The conditional value change of the portfolio is the linear weighted sum of independent noncentral Chi-square random variables, independent normal random variables and a constant, where the first two terms are independent.

\subsection*{Appendix E: Parameter generation of 4.1.1}
To generate a correlation matrix, we firstly generate a lower triangular matrix $L$ with elements generated from uniform distribution. Then we normalize the matrix $L$ so that the norm of each row is equal to 1. The correlation matrix is calculated by $LL^{\top}$. The price of stocks are set 1. For each stock, the annual standard deviation $\sigma$ are generated from uniform distribution $U(0.1,0.3)$, and the annual return $\mu = 0.15+0.5\sigma$. Accordingly, we estimate the option parameter using B-S model (\cite{black1973}). The moneyness (the ratio of the exercise price to the stock price) is generated from uniform distribution U(0.8, 1.2). The type of option and the corresponding underlying asset are uniformly randomly determined. The expiration of each option is set to one year.

\subsection*{Appendix F: GARCH model in empirical study}
We assume that each stock follows GJR-GARCH model that considers the leverage effect. For stock $i\in\{1,\cdots,m\}$, we assume
\begin{eqnarray*}
r_{it} &=& r+\kappa_i\sqrt{\sigma_{it}}-\frac{1}{2}\sigma_{it}+\epsilon_{it},\quad \epsilon_{it}|\mathscr{F}_{t-1}\sim\mathcal{N}(0,\sigma_{it}),\\
\sigma_{it} &=& \alpha_{i0}+\alpha_{i1}\epsilon_{it-1}^2+\alpha_{i2}\epsilon_{it-1}^2\mathbbm{1}_{\{\epsilon_{it-1}<0\}}+\alpha_{i3} \sigma_{it-1},
\end{eqnarray*}
where $r_{it}$ and $\sigma_{it}$ are the return and the variance of return of stock $i$ at time $t$, respectively, $r$ is the risk-free interest rate,  $\kappa_i$ is the risk premium, $\alpha_{ij}$ ($j=0,1,2,3$) are parameters, $\mathbbm{1}$ is an indicator function and $\mathscr{F}_{t}$ is the information set at time $t$.

We further assume that the correlation matrix of stocks follows a DCC type model
\begin{eqnarray*}
&&\bm{\epsilon}_t|\mathscr{F}_{t-1}\sim N(0,H_tC_tH_t),\\
&&D_t = C(1-\beta_1-\beta_2)+\beta_1\bm{\epsilon}_{t-1}\bm{\epsilon}_{t-1}^{\top}+\beta_2D_{t-1},\\
&&C_t = \text{diag}\{D_t\}^{-1}D_t\text{diag}\{D_t\}^{-1},
\end{eqnarray*}
where $H_t={\rm diag}\left(\sqrt{\sigma_{1t}},\cdots,\sqrt{\sigma_{mt}}\right)$, $\text{diag}(D_t)$ is a diagonal matrix with the elements being determined by the diagonal elements of matrix $D_t$ accordingly, $C_t$ is the correlation matrix at time $t$, and $C$ is the average correlation of the stocks. Using the Radon-Nikodym derivative of \cite{Rombouts2011}, we can obtain the DCC GARCH process under risk neutral measure that is used for option pricing as follows
\begin{eqnarray*}
&&r_{it} = r-\frac{1}{2}\sigma_{it}+\epsilon^Q_{it},\quad \epsilon^Q_{it}|\mathscr{F}_{t-1}\sim \mathcal{N}(0,\sigma_{it}),\quad i=1,\cdots,m,\\
&&\sigma_{it} = \alpha_{i0}+\alpha_{i1}(\epsilon^Q_{it-1}-\kappa_i\sqrt{\sigma_{it-1}})^2+\alpha_{i2}(\epsilon^Q_{it-1}-\kappa_i\sqrt{\sigma_{it-1}})^2\mathbbm{1}_{\{\epsilon^Q_{it-1}<\kappa_i\sqrt{\sigma_{it-1}}\}}+\alpha_{i3} \sigma_{it-1},\\
&&\bm{\epsilon}^Q_t|\mathscr{F}_{t-1}\sim \mathcal{N}(0,H_tC_tH_t),\\
&&D_t = C(1-\beta_1-\beta_2)+\beta_1(\bm{\epsilon}^Q_{t-1}-\bm{\kappa\sqrt{\sigma_{t-1}}})(\bm{\epsilon}^Q_{t-1}-\bm{\kappa\sqrt{\sigma_{t-1}}})^{\top}+\beta_2D_{t-1},\\
&&C_t = \text{diag}\{D_t\}^{-1}D_t\text{diag}\{D_t\}^{-1},
\end{eqnarray*}
where $\epsilon_{it}^Q=\epsilon_{it}+\kappa_i\sqrt{\sigma_{it}}$ and $\bm{\kappa\sqrt{\sigma_{t}}}=(\kappa_i\sqrt{\sigma_{it}})_{m}$ is a column vector.

\bibliographystyle{ormsv080}
\bibliography{hedgerisk}

\end{document}